\documentclass[12pt,reqno]{amsart}
\usepackage{mathtools} 
\usepackage{graphicx}
\usepackage{amssymb,amsmath}
\usepackage{amsthm}
\usepackage{color}
\usepackage[pdf]{pstricks}
\usepackage{hyperref}
\usepackage{empheq}
\usepackage{pgfplots}
\usepackage{amsfonts,amscd,cleveref}

\numberwithin{equation}{section}

\usepackage{epsfig}
\usepackage{extarrows}
\usepackage{caption}
\usepackage{subcaption}
\usepackage{tikz}

\newtheorem{remark}{Remark}

\newtheorem{assumption}{Assumption}

\newtheorem{lemma}{Lemma}
\newtheorem{theorem}{Theorem}

\newcommand{\cO}{\mathcal{O}}

\newcommand{\R}{\mathbb{R}}             
\newcommand{\ve}{\varepsilon}            
\newcommand{\Span}[1]{\mathrm{span}\left\{#1\right\}} 
\newcommand{\tr}[1]{\mathrm{Tr}\left(#1\right)} 
\newcommand{\hot}{h.o.t.}

\title[Bifurcations of stability spectra for steep Stokes waves]{\bf Recurrent bifurcations of stability spectra \\for steep Stokes waves in a deep fluid}

\author[S. Dyachenko]{Sergey Dyachenko}
\address[S. Dyachenko]{Department of Mathematics, State University of New York at Buffalo, Buffalo, New York, USA, 14260}
\email{sergeydy@buffalo.edu}
		
\author[R. Marangell]{Robert Marangell}
\address[R. Marangell]{School of Mathematics and Statistics F07, University of Sydney, NSW 2006, Australia}
\email{robert.marangell@sydney.edu.au}

\author[D. E. Pelinovsky]{Dmitry E. Pelinovsky}
\address[D. E. Pelinovsky]{Department of Mathematics and Statistics, McMaster University, Hamilton, Ontario, Canada, L8S 4K1}
\email{pelinod@mcmaster.ca}

\begin{document}
\maketitle

\begin{abstract} 
We study the modulational stability problem for the traveling periodic waves (called Stokes waves) 
in an infinitely deep fluid by using pseudo-differential operators in conformal variables. We derive 
the criteria and the normal forms for four bifurcations which are repeated recurrently when the steepness of the Stokes wave 
is increased towards the highest wave with the peaked profile. The four bifurcations are observed in the following order: 
(a) new figure-8 bands appearing at each extremal point of speed, 
(b) degeneration of figure-8 bands resulting in vertical slopes, 
(c) new circular bands around the origin appearing at each period-doubling bifurcation, and 
(d) reconnection of figure-$\infty$ bands at each extremal point of energy. 
Our work uses the analytic theory of Stokes waves developed previously for Babenko's equation. 
The novelty of our work is the analytic extension of the modulational stability problem for singular pseudo-differential operators 
in terms of the Floquet parameter.
The derivation of the normal form uses some structural assumptions which are known to be true for the Stokes waves. 
For the first and second bifurcation cycles, we compute numerically with a higher-order accuracy the actual values of wave steepness 
for which the structural assumptions are satisfied and the numerical coefficients of the normal forms
to show the excellent agreement between the normal form theory and the numerical approximations of the spectral bands. 
\end{abstract}

\section{Introduction}

Traveling periodic waves on the surface of an irrotational fluid under the gravity force are called Stokes waves ~\cite{stokes1847theory,stokes1880theory}. The mathematical theory of their existence was elaborated in ~\cite{levi1925determination,nekrasov1921waves,Struik1926}. The family of Stokes waves with the smooth profiles is continued towards a limit, in which the wave profile becomes peaked with the specific 120$^{o}$ angle at the crest. 
The existence of the limiting wave with the peaked profile was proven in~\cite{amick1982stokes,plotnikov2002proof,toland1978existence}. 

The potential motion of a two-dimensional fluid is described by the classical Euler's equations. Although many equivalent formulations are known for 
Euler's equations, we will particularly address the formulation involving pseudo-differential operators in conformal variables. 
When the fluid is considered to be infinitely deep, the conformal transformation maps the vertical semi-strip to 
the fluid domain and does not depend on the mean value of the fluid surface in conformal variable. This formulation leads to Babenko's pseudo-differential equation~\cite{babenko1987some}, which was used in the studies 
of traveling waves with both smooth and peaked profiles \cite{Balk1996,DyachenkoEtAl1996,dyachenko2016branch,locke2024peaked}. 

Based on global bifurcation theory and analysis of the pseudo-differential operators, several important 
results on the existence of Stokes waves were rigorously justified in \cite{BDT2000a,BDT2000b,Plotnikov1992}. 
It was shown in \cite{BDT2000b} that there exist infinitely many bifurcation points for the family of traveling periodic waves 
when it is extended towards the highest peaked wave. These points are identified as zero eigenvalues 
of the associated self-adjoint (Hessian) operator related to the variational formulation. They correspond to the extremal points 
of the wave speed versus wave steepness. The Morse index, which is the number of negative eigenvalues of the Hessian operator, becomes arbitrarily large 
in the limit of the highest peaked wave \cite{Plotnikov1992}. 
Based on the band-gap structure of the spectrum of the linearized operator with periodic coefficients, it 
was also shown in \cite{BDT2000b} that there are infinitely many subharmonic bifurcation points. These bifurcation points 
are identified as zero eigenvalues of the Hessian operators with respect to perturbations of multiple periods, in particular, 
with respect to double-periodic (antiperiodic) perturbations. Construction of traveling waves past the double--period bifurcation via continuation was recently shown in \cite{semenova2025}.

The global bifurcation theory is also related to the stability theory for Stokes waves. Stability of periodic waves with respect to co-periodic perturbations was studied numerically in~\cite{longuet1997crest,longuet1978instabilities,longuet1978theory,saffman1985,tanaka1983stability,tanaka1985}. It was realized in \cite{saffman1985,tanaka1985}
that the new unstable eigenvalues in the linearized stability problem appear at the extremal points of the wave energy, 
rather than at the extremal points of the wave speed. In \cite{DP2025}, we proved this conjecture based on analysis of the pseudo-differential operators in conformal variables and also showed that the extremal points of the wave energy coincide with the extremal points of the horizontal momentum. 

Finally, the modulational stability problem is defined by the spectrum of the linearized operators with respect to perturbations of 
longer periods. The modulational instability of small-amplitude Stokes waves, also known as the Benjamin-Feir instability~\cite{benjamin1967disintegration,zakharov1968stability}, was studied in detail~\cite{berti2023,berti2022full,bridgesmielke,creedon2022ahigh,hur2023,nguyen2020proof}. Additionally, the high-frequency instabilities 
of small-amplitude Stokes waves were studied for the same spectral problem in ~\cite{berti2022full,creedon2021high2,deconinck2011instability}.

Numerical solutions of the modulational stability problem for Stokes waves of increasing steepness 
have been recently developed with high accuracy in~\cite{dyachenko_semenova2022,dyachenko2023quasiperiodic,dyachenko2022almost,korotkevich2022superharmonic,murashige2020stability}. 
The recent comprehensive studies in~\cite{DDS2024,deconinck2022instability} showed that there exist four fundamental 
bifurcations of the spectral bands near the origin, which are repeated recurrently
when the family of Stokes waves is increased towards the highest peaked wave. The four bifurcations are observed in the following 
order: 
(a) new figure-8 bands appearing at each extremal point of speed, 
(b) degeneration of figure-8 bands resulting in their vertical slope, 
(c) new circular bands around the origin appearing at each period-doubling bifurcation, and 
(d) reconnection of figure-$\infty$ bands at each extremal point of energy. 
The necessity of infinitely many bifurcation points for (a) and (c) is proven in \cite{BDT2000b}. The bifurcation points for (d) were studied in \cite{DP2025}. The bifurcation points for (b) have not been seen in the previous 
work on the modulational stability problem in toy models for fluids \cite{BHM24,CuiP2025}.

It is the purpose of this work to study rigorously the modulational stability problem associated with the pseudo-differential operators in conformal variables and to derive the criteria and the normal forms for the four bifurcations 
near the origin. We also compute coefficients of the normal forms numerically for the first and second bifurcation cycles to show an excellent agreement between the normal forms and the numerical approximations of spectral bands near the origin.

\subsection{Equations of motion in conformal variables}
\label{sec-1-1}

Let $y = \eta(x,t)$ be the profile for the free surface of an incompressible and irrotational deep fluid in the $2\pi$-periodic domain $\mathbb{T}$ that depends on time $t \in \mathbb{R}$. For a proper definition of the free surface, we add the zero-mean constraint $\oint \eta(x,t) dx = 0$, which is invariant in the time evolution of Euler's equations. Here and in what follows, the symbol $\oint$ denotes an integral of a periodic function over the period, which is independent of the limits of integration. We use notations $L^2_{\rm per}$ for $2\pi$-periodic squared-integrable functions and similarly, $H^1_{\rm per}$ for $2\pi$-periodic squared-integrable functions with squared-integrable first derivatives.

Let $z(u,t) = \xi(u,t) + i\eta(u,t)$ be a holomorphic function, which maps the vertical strip in the lower complex half-plane $\mathbb{T} \times i(-\infty,0]$ to the fluid domain beneath the free surface $y = \eta(x,t)$. The top boundary $\mbox{Im}\,u = 0$ gives the surface profile in the parametric form: $x = \xi(u,t)$ and $y = \eta(u,t)$, where $u \in \mathbb{T}$ and $t \in \mathbb{R}$. The conformal mapping yields the relation 
\begin{equation}
\label{relation-xi-eta}
\xi = u - \mathcal{H} \eta,
\end{equation} 
where $\mathcal{H}$ is the periodic Hilbert transform, a skew-adjoint bounded operator in $L^2_{\rm per}$, normalized by the Fourier symbol 
\begin{equation}
\label{Hilbert-operator}
\hat{\mathcal{H}}_n = \left\{ \begin{array}{ll} 
i \; {\rm sgn}(n), \quad & n \in \mathbb{Z} \backslash \{0\}, \\
0, \quad & n = 0. \end{array} \right.
\end{equation}
It follows from (\ref{relation-xi-eta}) that 
$$
\xi_u = 1 + \mathcal{K} \eta \quad \mbox{\rm and} \quad \xi_t = -\mathcal{H}\eta_t,
$$
where $\mathcal{K} = - \mathcal{H} \partial_u = |\partial_u|$ is a positive self-adjoint operator in $L^2_{\rm per}$ with 
the domain $H^1_{\rm per}$ and the Fourier symbol 
\begin{equation}
\label{operator-K}
\hat{\mathcal{K}}_n = |n|, \quad n \in \mathbb{Z}.
\end{equation}

Derivation of equations of motion in conformal variables can be found 
in Appendices of \cite{dyachenko2016branch,locke2024smooth}.
These equations can be written as the system of pseudo-differential equations for 
the velocity potential $\psi$ and the free surface $\eta$ defined at the top boundary $\mbox{Im}\,u = 0$:
\begin{equation}
\label{time-Bab-eq}
\left\{ \begin{array}{l}
\eta_t (1 + \mathcal{K} \eta)  + \eta_u \mathcal{H} \eta_t + \mathcal{H} \psi_u = 0, \\
\psi_t \eta_u - \psi_u \eta_t + \eta \eta_u + \mathcal{H} \left( (1 + \mathcal{K} \eta) \psi_t + 
\psi_u \mathcal{H} \eta_t + \eta  (1 + \mathcal{K}  \eta) \right) = 0.
\end{array}
\right. 
\end{equation}
The system (\ref{time-Bab-eq}) admits the following conserved quantities, e.g., see \cite{benjamin1982hamiltonian} and \cite{DP2025}:
\begin{align}
\label{conserved-mass}
M(\eta) &= \oint \eta (1+ \mathcal{K} \eta) du, \\
\label{conserved-momentum}
P(\psi,\eta) &= -\oint \psi \eta_u du, \\
\label{conserved-mass-2}
Q(\psi,\eta) &= \oint \psi (1 + \mathcal{K} \eta) du, \\
\label{conserved-Ham}
H(\psi,\eta) &= \frac{1}{2}\oint \left( \psi \mathcal{K} \psi + \eta^2 (1 + \mathcal{K} \eta) \right) du,
\end{align}
which are referred to as {\em mass}, {\em horizontal momentum}, {\em vertical momentum}, and {\em energy}, respectively. Due to the zero-mean constraint $\oint \eta(x,t) dx = 0$ on the surface elevation $\eta$ and the chain rule $dx = (1 + \mathcal{K} \eta) du$, we get the constraint $M(\eta) = 0$, or explicitly, from (\ref{conserved-mass}),
\begin{equation}
\label{zero-mean}
\oint \eta (1+ \mathcal{K} \eta) du = 0.
\end{equation}

In the reference frame moving with the wave speed $c$, the system (\ref{time-Bab-eq}) with the constraint (\ref{zero-mean}) can be rewritten as 
\begin{align}
\left\{ \begin{array}{l}
\eta_t (1 + \mathcal{K} \eta) + \eta_u \mathcal{H} \eta_t = \mathcal{K} \zeta, \\
\zeta_t (1 + \mathcal{K} \eta) + \zeta_u \mathcal{H} \eta_t - \mathcal{H} (\zeta_t \eta_u - \zeta_u \eta_t ) - 2 c \mathcal{H} \eta_t  = \left(c^2 \mathcal{K} - 1\right)\eta  - \eta \mathcal{K} \eta - \frac{1}{2} \mathcal{K} \eta^2,
\end{array}
\right. 
\label{full-single-eq}
\end{align}
where the potential $\psi$ has been transformed to $\zeta$ by using $\psi = -c \mathcal{H} \eta + \zeta$ and the variable $u$ stands now for 
$u - ct$. 

{\em Traveling waves} correspond to the reduction $\zeta = 0$ and $\eta_t = 0$ in the system (\ref{full-single-eq}). This gives the scalar pseudo-differential Babenko's equation for the $2\pi$-periodic profile $\eta = \eta(u)$:
\begin{equation}
\label{trav-Bab-eq}
(c^2 \mathcal{K} - 1) \eta  = \frac{1}{2} \mathcal{K} \eta^2 + \eta \mathcal{K} \eta.
\end{equation}
This equation can be obtained as the Euler--Lagrange equation for the action functional
\begin{equation}
\label{aug-energy}
\Lambda_c(\eta) := \frac{1}{2} \langle (c^2 \mathcal{K} - 1) \eta, \eta \rangle - 
\frac{1}{2} \langle \mathcal{K} \eta^2, \eta \rangle, 
\end{equation}
where $\langle f, g \rangle := \frac{1}{2\pi} \oint \bar{f}(u) g(u) du$ is the normalized inner product in $L^2_{\rm per}$. 

It was shown in \cite{BDT2000a,BDT2000b} that {\em there exists a family of smooth traveling waves with the even profile $\eta \in C^{\infty}_{\rm per}$ satisfying the Babenko equation (\ref{trav-Bab-eq}) for $c \in (1,c_*)$ with some $c_* > 1$.} The left end $c = 1$ of the interval $(1,c_*)$ is a local bifurcation point of the $2\pi$-periodic solutions from the zero solution. The even single-lobe profile $\eta \in C^{\infty}_{\rm per}$  is approximated near $c = 1$ by the following small-amplitude expansion \cite{creedon2022ahigh,DyachenkoEtAl1996}:
	\begin{align}
	\label{eta-small}
	\eta = a \cos(u) + a^2 \left[ \cos(2u) - \frac{1}{2} \right] + \frac{3}{2} a^3 \cos(3u) + \mathcal{O}(a^4)
	\end{align}
	and 
	\begin{align}
	\label{c-small}
	c^2 = 1 + a^2 + \mathcal{O}(a^4),	
	\end{align}
where $a \in \mathbb{R}$ is the amplitude parameter.

\subsection{Spectral and modulational stability of traveling waves}
\label{sec-1-2}

Expanding system (\ref{full-single-eq}) near the traveling wave with the profile $(\eta,0)$ and truncating the system at the linear terms with respect to the co-periodic perturbation $(v,w)$, we obtain the following system of linearized equations:
\begin{equation}
\left\{ \begin{array}{ll}
\mathcal{M} v_t \qquad \qquad & = \mathcal{K} w, \\
\mathcal{M}^* w_t - 2 c \mathcal{H} v_t & = \mathcal{L} v,
\end{array}
\right. 
\label{lin-Bab-eq}
\end{equation}
where 
\begin{align}
\label{operator-M}
\mathcal{M} := 1 + \mathcal{K} \eta + \eta' \mathcal{H} \quad\mbox{and}\quad
\mathcal{M}^* := 1 + \mathcal{K} \eta - \mathcal{H} (\eta' \; \cdot ), 
\end{align}
and
\begin{align}
\mathcal{L} := c^2 \mathcal{K} - (1 + \mathcal{K} \eta) - \eta \mathcal{K} - \mathcal{K} (\eta \; \cdot ). \label{linBop}
\end{align}
The linear operators $\mathcal{K}$ and $\mathcal{L}$ are unbounded and self-adjoint in $L^2_{\rm per}$ with the domains 
$$
{\rm Dom}(\mathcal{K}) = {\rm Dom}(\mathcal{L}) = H^1_{\rm per}.
$$ 
Since $\mathcal{H}$ is a bounded skew-adjoint operator in $L^2_{\rm per}$,  the linear operators $\mathcal{M}$ and $\mathcal{M}^*$ are bounded and adjoint to each other in $L^2_{\rm per}$. 

Separating variables in the linearized system (\ref{lin-Bab-eq}) yields the spectral stability problem with respect to co-periodic perturbations, 
\begin{equation}
\left\{ \begin{array}{ll}
\mathcal{K} w &= \lambda \mathcal{M} v, \\
\mathcal{L} v &= \lambda (\mathcal{M}^* w - 2 c \mathcal{H} v),
\end{array}
\right. 
\label{spec-Bab-eq}
\end{equation}
where $(v,w) \in H^1_{\rm per} \times H^1_{\rm per}$ is an eigenfunction and $\lambda \in \mathbb{C}$ is an eigenvalue. 
The following lemma describes the Hamiltonian symmetry of eigenvalues. 

\begin{lemma}
    \label{lem-spectrum}
    The spectrum of the spectral problem (\ref{spec-Bab-eq}) in $L^2_{\rm per} \times L^2_{\rm per}$ 
    consists of eigenvalues $\lambda$ of finite algebraic multiplicity. If $\lambda \in \mathbb{C}$ is an eigenvalue, 
    so are $\bar{\lambda}$, $-\lambda$, and $-\bar{\lambda}$. 
\end{lemma}

\begin{proof}
The first assertion holds since $\mathbb{T}$ is compact, $\mathcal{K}$ and $\mathcal{L}$ are unbounded operators in $L^2_{\rm per}$, 
and $\mathcal{H}$, $\mathcal{M}$, and $\mathcal{M}^*$ are bounded operators in $L^2_{\rm per}$. The second assertion 
follows from the properties:
    \begin{itemize}
\item $\mathcal{K}$, $\mathcal{L}$, $\mathcal{M}$, $\mathcal{M}^*$, and $\mathcal{H}$ have real-valued coefficients.
\item $\mathcal{K}$, $\mathcal{L}$, $\mathcal{M}$, and $\mathcal{M}^*$ are invariant and $\mathcal{H}$ is covariant with respect to 
the parity transformation $(\Pi f)(x) = f(-x)$ such that 
$$
\Pi \mathcal{H} \Pi = -\mathcal{H}, \quad \Pi \mathcal{M} \Pi = \mathcal{M}, \quad \Pi \mathcal{M}^* \Pi = \mathcal{M}^*, \quad \Pi \mathcal{L} \Pi = \mathcal{L}, \quad 
\Pi \mathcal{K} \Pi = \mathcal{K}.
$$
\end{itemize}
As a consequence of these properties, if $(v,w) \in H^1_{\rm per} \times H^1_{\rm per}$ is an eigenfunction of the spectral problem (\ref{spec-Bab-eq}) for an eigenvalue $\lambda \in \mathbb{C}$, then $(\bar{v},\bar{w})$ is an eigenfunction for $\bar{\lambda}$, 
$(\Pi v, -\Pi w)$ is an eigenfunction for $-\lambda$, and $(\Pi \bar{v}, -\Pi \bar{w})$ is an eigenfunction for $-\bar{\lambda}$.
\end{proof}

\begin{remark}
    The nonlinear system (\ref{full-single-eq}) and the linearized system (\ref{lin-Bab-eq}) admit a Hamiltonian 
    formulation with a nonlocal symplectic structure. Since we study the spectral stability of traveling waves,
    we will not write the Hamiltonian structure explicitly. 
\end{remark}

To introduce the modulational stability problem, we extend the spectral stability problem (\ref{spec-Bab-eq}) 
from co-periodic perturbations $(v,w) \in H^1_{\rm per} \times H^1_{\rm per}$ to perturbations defined on the real line 
$(v,w) \in H^1(\mathbb{R}) \times H^1(\mathbb{R})$. To do so, we use the 
Floquet--Bloch theory and introduce the characteristic exponent $\mu \in \mathbb{B} := \left(-\frac{1}{2},\frac{1}{2} \right]$, where $\mathbb{B}$ is reciprocal to $\mathbb{T} = [0,2\pi)$. By using the Bloch transform 
$$
v(u) = \int_{\mathbb{B}} \hat{v}(u,\mu) e^{i \mu u} d \mu, \quad 
w(u) = \int_{\mathbb{B}} \hat{w}(u,\mu) e^{i \mu u} d \mu, 
$$
where $(\hat{v},\hat{w})(\cdot,\mu) \in H^1_{\rm per} \times H^1_{\rm per}$ for every fixed $\mu \in \mathbb{B}$, 
we obtain
\begin{equation}
\left\{ \begin{array}{ll}
\mathcal{K}_{\mu} \hat{w} &= \lambda \mathcal{M}_{\mu} \hat{v}, \\
\mathcal{L}_{\mu} \hat{v} &= \lambda (\mathcal{M}^*_{\mu} \hat{w} 
- 2 c \mathcal{H}_{\mu} \hat{v}),
	\end{array}
	\right. 
	\label{mod-Bab-eq-aux}
	\end{equation}
which is again defined in the periodic domain $\mathbb{T}$ for every $\mu \in \mathbb{B}$. The expressions for each operator $\mathcal{K}_{\mu}$, $\mathcal{L}_{\mu} $, $\mathcal{M}_{\mu} $, $\mathcal{M}^*_{\mu} $, and $\mathcal{H}_{\mu}$ are obtained below for $\mu \neq 0$. \\

\underline{$\mathcal{H}_{\mu} = e^{-i \mu u} \mathcal{H} e^{i \mu u}$:} 
Since $\mathcal{H} \sum\limits_{n \in \mathbb{Z}} \hat{f}_n e^{i (\mu + n) u} 
= i \sum\limits_{n \in \mathbb{Z}} {\rm sgn}(\mu + n) \hat{f}_n e^{i (\mu + n) u}$ and $\mu \in \mathbb{B}$, we get 
\begin{equation*}
\mathcal{H}_{\mu} f = \mathcal{H} f + i {\rm sgn}(\mu) \hat{f}_0, \quad \mu \neq 0, 
\end{equation*}
where $\hat{f}_0 = \langle 1, f \rangle$ due to normalization of the inner product in $L^2_{\rm per}$. Thus, we can define 
\begin{equation}
    	\label{H-mu}
  \mathcal{H}_{\pm}  = \mathcal{H} \pm i \langle 1, \cdot \rangle, \quad \pm \mu > 0,   
\end{equation}
where $\mathcal{H}$ is defined by (\ref{Hilbert-operator}).

\underline{$\mathcal{K}_{\mu} = e^{-i \mu u} \mathcal{K} e^{i \mu u}$:} Since $\mathcal{K}_{\mu} = - \mathcal{H}_{\mu} (\partial_u + i \mu)$, we get the expansion 
\begin{equation}
\label{K-mu}
\mathcal{K}_{\mu} = \mathcal{K} - i \mu \mathcal{H}_{\pm}, \quad \pm \mu > 0,
\end{equation}
where $\mathcal{K}$ is defined by (\ref{operator-K}).

\underline{$\mathcal{M}_{\mu} = e^{-i \mu u} \mathcal{M} e^{i \mu u}$:} It follows from (\ref{operator-M}) 
and (\ref{H-mu}) that $\mathcal{M}_{\mu} = \mathcal{M}_{\pm}$ given by 
\begin{align}
\label{M-mu}
\mathcal{M}_{\pm} = 1 + \mathcal{K}\eta + \eta' \mathcal{H}_{\pm}, \quad  \pm \mu > 0.
\end{align}
Similarly, we get $\mathcal{M}^*_{\mu} = \mathcal{M}^*_{\pm}$ given by 
\begin{align}
\label{M-star-mu}
\mathcal{M}^*_{\pm} = 1 + \mathcal{K}\eta  - \mathcal{H}_{\pm}(\eta' \; \cdot), \quad  \pm \mu > 0.
\end{align}

\underline{$\mathcal{L}_{\mu} = e^{-i \mu u} \mathcal{L} e^{i \mu u}$:} It follows from (\ref{linBop}) and (\ref{K-mu}) that 
\begin{align}
\label{L-mu}
\mathcal{L}_{\mu} &= \mathcal{L} - i \mu \mathcal{L}_{\pm}, \quad \pm \mu > 0,
\end{align}
where 
\begin{equation}
\label{L-plus-minus}
\mathcal{L}_{\pm} = c^2 \mathcal{H}_{\pm} - \eta \mathcal{H}_{\pm} - \mathcal{H}_{\pm} (\eta \; \cdot).
\end{equation}

As a result of (\ref{H-mu}), (\ref{K-mu}), (\ref{M-mu}), (\ref{M-star-mu}), and (\ref{L-mu}), 
we can rewrite the modulational stability problem (\ref{mod-Bab-eq-aux}) in the equivalent form
\begin{equation}
\left\{ \begin{array}{ll}
(\mathcal{K} - i \mu \mathcal{H}_{\pm}) \hat{w} &= \lambda \mathcal{M}_{\pm} \hat{v}, \\
(\mathcal{L} - i \mu \mathcal{L}_{\pm}) \hat{v} &= \lambda (\mathcal{M}^*_{\pm} \hat{w} 
- 2 c \mathcal{H}_{\pm} \hat{v}),
\end{array}
\right. 
\label{mod-Bab-eq}
\end{equation}
where $\mu \neq 0$ and the two signs refer to $\pm \mu > 0$. The following lemma is similar 
to Lemma \ref{lem-spectrum} but is extended to the modulational stability problem (\ref{mod-Bab-eq}).

\begin{lemma}
    \label{lem-mod-spectrum}
    The spectrum of the spectral problem (\ref{mod-Bab-eq}) in $L^2_{\rm per} \times L^2_{\rm per}$ 
    consists of eigenvalues $\lambda$ of finite algebraic multiplicity for each fixed $\mu \in \mathbb{B}$. 
    Consequently, we can enumerate the spectral bands 
\begin{equation}
\label{spectral-bands}
\sigma_S = \cup_{j \in \mathbb{N}} \{ \lambda_j(\mu), \;\; \mu \in \mathbb{B} \}.
\end{equation}
If $\lambda_j(\mu)$ is a spectral band of the spectral stability problem (\ref{mod-Bab-eq}) for $\mu \in \mathbb{B}$, 
so are $\bar{\lambda}_j(-\mu)$, $-\lambda_j(-\mu)$, and $-\bar{\lambda}_j(\mu)$.
\end{lemma}

\begin{proof}
        The first assertion holds again because $\mathbb{T}$ is compact, $\mathcal{K}$ and $\mathcal{L}$ are unbounded operators in $L^2_{\rm per}$, whereas $\mathcal{H}_{\pm}$, $\mathcal{M}_{\pm}$, $\mathcal{M}^*_{\pm}$, and $\mathcal{L}_{\pm}$ are bounded operators in $L^2_{\rm per}$. 
The second assertion follows from the properties of $\mathcal{K}$, $\mathcal{L}$, $\mathcal{M}$, $\mathcal{M}^*$, and $\mathcal{H}$ in the proof of Lemma \ref{lem-spectrum} and the additional properties 
$$
    \mathcal{H}_{\pm} = \overline{\mathcal{H}}_{\mp}, \quad \mathcal{H}_{\pm} |_{\mu \to -\mu} = \mathcal{H}_{\mp}, \quad 
    \Pi \mathcal{H}_{\pm} \Pi = -\mathcal{H}_{\mp},
    $$ 
Consequently, if  $(\hat{v}(\mu),\hat{w}(\mu)) \in H^1_{\rm per} \times H^1_{\rm per}$ is an eigenfunction of the spectral problem (\ref{mod-Bab-eq}) for an eigenvalue $\lambda(\mu) \in \mathbb{C}$ with fixed $\mu \in \mathbb{B}$, then $(\bar{v}(-\mu),\bar{w}(-\mu))$ is an eigenfunction for $\bar{\lambda}(-\mu)$, $(\Pi v(-\mu), -\Pi w(-\mu))$ is an eigenfunction for $-\lambda(-\mu)$, and $(\Pi \bar{v}(\mu), -\Pi \bar{w}(\mu)$ is an eigenfunction for $-\bar{\lambda}(\mu)$.
\end{proof}

\begin{remark}
We say that the periodic wave with the profile $\eta \in C^{\infty}_{\rm per}$ is modulationally unstable if the spectral bands in $\sigma_S$ intersecting the origin do not coincide with a subset of $i\R$ for $\lambda \neq 0$. 
\end{remark}

\begin{remark}
	\label{rem-anal}
The original spectral  problem (\ref{mod-Bab-eq-aux}) is not analytic 
at $\mu = 0$ due to the different behavior of the linear operators in the limit $\mu \to 0^{\pm}$. However, after extending $\mathcal{H}_{\mu}$, $\mathcal{K}_{\mu}$, $\mathcal{M}_{\mu}$, $\mathcal{M}_{\mu}^*$, and $\mathcal{L}_{\mu}$ differently for $\mu > 0$ and $\mu < 0$, the reformulated spectral problem (\ref{mod-Bab-eq}) is now analytic in $\mu$ and $\lambda$ for both signs.
\end{remark}

\begin{figure}[htb!]
	\centering
	\includegraphics[width=0.985\textwidth]{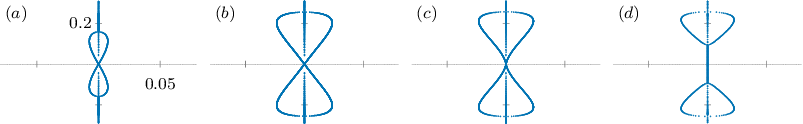} 
	\includegraphics[width=0.985\textwidth]{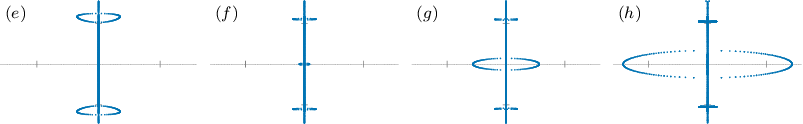} 
	\includegraphics[width=0.985\textwidth]{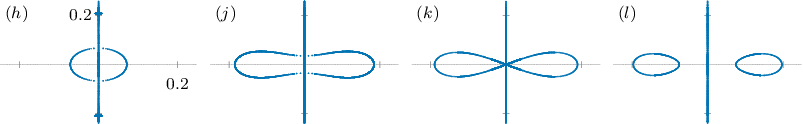} 
	\includegraphics[width=0.985\textwidth]{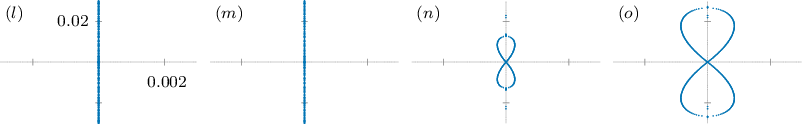} 
	\caption{Transformation of the spectral bands in the modulational stability problem (\ref{mod-Bab-eq}) when the steepness parameter $s$ of Stokes waves is increased.}
	\label{fig-bif}
\end{figure}

\subsection{Main results and organization of the paper}
\label{sec-1-3}

Figure \ref{fig-bif} shows numerical approximations of the spectral bands (\ref{spectral-bands}) in the modulational stability problem (\ref{mod-Bab-eq}) 
when the traveling wave is continued in the steepness parameter $s$ from the small-amplitude limit given by 
(\ref{eta-small})--(\ref{c-small}) towards the limiting wave with the peaked profile. 
We observe four bifurcations of the spectral bands near the origin described below. 
These four bifurcations were discovered in \cite{DDS2024}, see their Figure 3. Based on the first three cycles, it was conjectured in \cite{DDS2024} that the four bifurcations are repeated recurrently in the same sequence infinitely many times before the wave profile becomes peaked.

\begin{remark}
    Panels (h) and (l) are repeated twice in Figure \ref{fig-bif} to show that between (a)--(h), (h)--(l), and (l)--(o), we recalibrate the scales along the real and imaginary 
    axis of the complex $\lambda$-plane, for better presentation of bifurcations. The scaling units are explicitly shown in the first panel of the calibration and are kept for all subsequent panels until the next calibration. 
\end{remark}

Our main results consist of the derivation and justification of the criteria and the normal forms for the four bifurcations near the origin shown in Figure \ref{fig-bif}. To be precise, we describe each bifurcation separately in the list below. 

\begin{itemize} 
\item The first figure-$8$ bands are seen in (a). This bifurcation, also known as 
the Benjamin--Feir instability, is studied in Appendix \ref{app-A}. The normal form 
for the stability bands is given by 
	\begin{equation}
	\left( \lambda - \frac{i \mu}{2} \right)^2 = \frac{1}{64} \mu^2 (8 a^2 - \mu^2) + \mathcal{O}(a^6 + \mu^6),
	\label{fig-8-small-intro}
	\end{equation}
where $a \in \R$ is a small parameter of the asymptotic expansions (\ref{eta-small}) and (\ref{c-small}). 
The normal form was derived and justified in the previous works \cite{berti2022full,creedon2022ahigh}.\\

\item Panels (b)--(d) show the second bifurcation, for which the slopes of the figure-$8$ bands become vertical at $\lambda = 0$, 
after which the figure-$8$ bands break into two instability bubbles along the imaginary axis. This bifurcation is 
analyzed in Section \ref{sec-3} by expanding solutions of the spectral problem (\ref{mod-Bab-eq}) near $(\mu,\lambda) = (0,0)$. 
The normal form for the stability bands is given by 
    \begin{align}
		\label{normal-form-one-intro}
\mathcal{P}'(c_0) \left( \lambda - \frac{i \mu \mathcal{N}'(c)}{2 \mathcal{P}'(c)} \right)^2 + \frac{\mu^2 \Delta'(c_0) (c - c_0)}{4 \mathcal{P}'(c_0)} + \mathcal{D} |\mu|^3 = \mathcal{O}(\mu^4),
	\end{align}	
closed with the consistency condition $c - c_0 = \mathcal{O}(\mu)$, 
where $\mathcal{P}(c)$ and $\mathcal{N}(c)$  are given by (\ref{wave-momentum}), $\Delta(c)$ is the bifurcation parameter in (\ref{deltaC}), 
$\mathcal{D}$ is the numerical coefficient in (\ref{num-coeff-D}), and $c_0$ is a simple root of $\Delta(c)$. 
See Theorems \ref{theorem-bif-1} and \ref{theorem-bif-1-normal-form}. 
This bifurcation originates from the tangential intersection of the double-degenerate spectral bands at  $(\mu,\lambda) = (0,0)$. It has not 
been considered in the previous literature.\\

\item Panels (e)--(g) show the third bifurcation, for which a new circular instability bubble emerges near $\lambda = 0$. This bifurcation is analyzed in Section \ref{sec-4} by expanding 
solutions of the spectral problem (\ref{mod-Bab-eq}) near $(\mu,\lambda) = \left(\frac{1}{2},0\right)$. 
The normal form for the stability bands is given by 
    \begin{align}
		\label{normal-form-two-intro}
	\lambda^2\mathcal{A}_{1} + 2 i c_0 \tilde{\mu} \lambda \mathcal{A}_{2} + \tilde{\mu}^2 \mathcal{A}_{3} 
	- \Lambda'(c_0) (c-c_0) \| \psi_0 \|^2_{L^2} = \mathcal{O}\left( |\tilde{\mu}|^3 \right),
\end{align}
closed with the consistency condition $c - c_0 = \mathcal{O}(\tilde{\mu}^2)$ for $\tilde{\mu} = \frac{1}{2} - \mu$, 
where $\mathcal{A}_{1,2,3}$ are numerical coefficients in (\ref{coefficients-A}), 
$\Lambda(c)$ is the simple eigenvalue of 
$$
\mathcal{L}_{\rm aper} : H^1_{\rm aper} \subset L^2_{\rm aper} \to  L^2_{\rm aper},
$$
associated with the eigenfunction $\psi_0 \in H^1_{\rm aper}$, 
and $c_0$ is a simple root of $\Lambda(c)$. See Theorems \ref{theorem-bif-2} and \ref{theorem-bif-2-normal-form}.
This bifurcation originates from the one-dimensional kernel of $\mathcal{L}_{\rm aper}$ 
in the space of double-periodic (antiperiodic) perturbations. Such subharmonic 
bifurcations for fixed rational values of $\mu$, including $\mu = \frac{1}{2}$, have been studied in \cite{BDT2000b}. \\

\item Panels (j)--(l) show the fourth bifurcation, for which the instability band surrounding $\lambda = 0$ reconnects into the figure-$\infty$ bands, which then break into two instability bubbles along the real axis. This bifurcation is analyzed in Section \ref{sec-5} by expanding solutions of the spectral problem (\ref{mod-Bab-eq}) near $(\mu,\lambda) = (0,0)$. The normal form for the stability bands is given by 
    \begin{align}
		\label{normal-form-three-intro}
	\lambda^4 \mathcal{B} - i \mu \lambda \mathcal{N}'(c_0) + \lambda^2 \mathcal{P}''(c_0) (c - c_0) = \mathcal{O}(\mu^2),	
\end{align}
closed with the consistency condition $c - c_0 = \mathcal{O}(|\mu|^{2/3})$, 
where $\mathcal{P}(c)$ and $\mathcal{N}(c)$  are given by (\ref{wave-momentum}), $\mathcal{B}$ is the numerical coefficient in (\ref{B-coef}), 
and $c_0$ is a simple root of $\mathcal{P}'(c)$. See Theorem \ref{theorem-bif-3}. This bifurcation originates from the extreme points of energy $\mathcal{H}$ (equivalently, the horizontal momentum), at which the spectral stability problem (\ref{spec-Bab-eq}) for co-periodic perturbations admits the zero eigenvalue bifurcation. The zero-eigenvalue bifurcation for $\mu = 0$ has been studied in \cite{DP2025}.\\

\item Panels (m)--(o) show the first bifurcation again, for which new figure-$8$ bands arise near $\lambda = 0$. This bifurcation is analyzed 
in Section \ref{sec-6} by expanding solutions of the spectral problem (\ref{mod-Bab-eq}) near $(\mu,\lambda) = (0,0)$. 
 The normal form for the stability bands is given by 
    \begin{align}
		\label{normal-form-four-intro}
        \left( \lambda - \frac{i \mu c_0}{2} \right)^2 = \mu^2 \left[ \mathcal{N}(c_0) - \frac{5c_0}{4} \mathcal{P}(c_0) \right] 
        \left[ |\mu| + \frac{\delta_0 \varepsilon}{c_0^2 \langle 1, v_0 \rangle} \right] + \mathcal{O}(|\mu|^{7/2}),
\end{align}
closed with the consistency condition $c^2 = c_0^2 + \delta_0 \varepsilon^2 + \mathcal{O}(\varepsilon^3)$ and $\varepsilon = \mathcal{O}(\mu)$,
where $\mathcal{P}(c)$ and $\mathcal{N}(c)$  are given by (\ref{wave-momentum}),
$\delta_0$ is given by (\ref{delta}), and $c_0$ is the fold bifurcation point, for which 
the kernel of 
$$
\mathcal{L} : H^1_{\rm per} \subset L^2_{\rm per} \to  L^2_{\rm per}
$$ 
in the space of co-periodic perturbations is two-dimensional with the eigenfunctions $\eta'$ and $v_0 \in H^1_{\rm per}$ 
satisfying $\langle 1, v_0 \rangle \neq 0$. See Theorems 
\ref{theorem-bif-4} and \ref{theorem-bif-4-normal-form}. 
This bifurcation originates from the extreme points of speed $c$, 
at which the family of smooth traveling waves 
parameterized by steepness $s$ has a fold bifurcation. 
Fold bifurcations for $\mu = 0$ have been studied in \cite{BDT2000b}. 
\end{itemize}

\begin{remark}
It is remarkable that the normal form (\ref{normal-form-four-intro}) for the steep Stokes waves 
is different from the normal form (\ref{fig-8-small-intro}) for the small-amplitude Stokes waves, 
despite both describing a figure-8 bifurcation. It is seen from (\ref{normal-form-one-intro}) and (\ref{normal-form-four-intro}) 
that the original spectral problem (\ref{mod-Bab-eq-aux}) is not analytic at $\mu = 0$ because the pseudo-differential operators $\mathcal{H}_{\mu}$
and $\mathcal{K}_{\mu}$ are singular at $\mu = 0$. Nevertheless, we derive the normal forms by analytic extension of these operators 
for $\mu > 0$ and $\mu < 0$ in the reformulated spectral problem (\ref{mod-Bab-eq}). 
\end{remark}

The numerical method used in our computations is described in Appendix \ref{app-B}.
The profile $\eta \in C^{\infty}_{\rm per}$ of the traveling wave 
is computed from Babenko's equation (\ref{trav-Bab-eq}) as a function of 
the wave steepness $s$, see Appendix \ref{app-B-1}. The spectral bands (\ref{spectral-bands}) of the modulational stability problem (\ref{mod-Bab-eq}) are computed for each steepness $s$ as functions of the Floquet parameter $\mu$, see Appendix \ref{app-B-2}. We compare the spectral bands computed from the normal form equations for each bifurcation described above with the spectral bands zoomed near the origin in the complex $\lambda$-plane. Table~\ref{tab:bif} provides 
numerical data for the steepness $s$, speed $c$, and energy $\mathcal{H}$ at each bifurcation point in the first two cycles of the four bifurcations.

\begin{table}[htb!]
	\centering
	\renewcommand{\arraystretch}{1.3}
	\begin{tabular}{c|c|c|c|l|l}
		$s$ & $c$ & $\mathcal H$ & Bifurcation & Figure \\ \hline
		$0$           & $1$           &  $0$             & Figure-$8$ I &        \\ 
        $0.10921308$  & $1.06052068$  &  $0.34626647037$ & Hourglass I &          ~Fig.~\ref{fig:bcd1} \\      
		$0.12890310$  & $1.08414013$  &  $0.44648936798$ & Ellipse I &            ~Fig.~\ref{fig:efg1} \\
		$0.13660355$  & $1.09213793$  &  $0.46517718146$ & Figure-$\infty$ I &    ~Fig.~\ref{fig:jkl1} \\
		$0.13875329$  & $1.09295138$  &  $0.46256418841$ & Figure-$8$ II &     ~Fig.~\ref{fig:mno1} \\  
		$0.14010224$  & $1.09256435$  &  $0.45870019516$ & Hourglass II &       ~Fig.~\ref{fig:bcd2} \\ 
		$0.14048675$  & $1.09238125$  &  $0.45792453975$ & Ellipse II &         ~Fig. \ref{fig:efg2} \\
		$0.14079655$  & $1.09228678$  &  $0.45770578965$ & Figure-$\infty$ II & ~Fig. \ref{fig:jkl2} \\
	\end{tabular}
	\caption{Parameters of Stokes waves at the four bifurcation points: steepness $s$, speed $c$, energy $\mathcal{H}$, and the bifurcation type in the first two cycles.}
	\label{tab:bif}
\end{table}

The main computations of the normal forms are described in Sections \ref{sec-3}, \ref{sec-4}, \ref{sec-5}, and \ref{sec-6}. Section \ref{sec-2} contains two preliminary results about integral characteristics of the traveling waves with profile $\eta \in C^{\infty}_{\rm per}$ and about the behavior of the spectral problem (\ref{mod-Bab-eq}) in the limit $\mu \to 0^{\pm}$. Section \ref{sec-1-4} describes the Newton polytope method, which is used to justify the balance of asymptotic terms in the direct Puiseux expansions.

\section{Preliminary results}
\label{sec-2}

We describe two results which clarify the behavior of the spectral problem (\ref{mod-Bab-eq}) near $(\mu,\lambda) = (0,0)$. 
Section \ref{sec-2-1} discusses smoothness of continuation 
of the traveling wave with the even profile $\eta \in C^{\infty}_{\rm per}$ with respect to the wave speed $c \in (1,c_*)$ and 
derives some relations between the integral quantities related to the horizontal momentum (\ref{conserved-momentum}) and the energy (\ref{conserved-Ham}). Both 
facts are useful in the explicit computations of the projection matrix 
at $(\mu,\lambda) = (0,0)$ used in the bifurcation theory of Sections \ref{sec-3}, \ref{sec-5}, and \ref{sec-6}. Section \ref{sec-2-2} compares the limit 
$\mu \to 0^{\pm}$ in the spectral problem (\ref{mod-Bab-eq}) 
with the co-periodic stability problem (\ref{spec-Bab-eq}), where we prove that 
the eigenfunctions of the two problems are related by an invertible transformation.  

\subsection{Integral quantities for traveling waves}
\label{sec-2-1}

Recall the linear operator 
$$
\mathcal{L} : H^1_{\rm per} \subset L^2_{\rm per} \to L^2_{\rm per},
$$ 
given by (\ref{linBop}). It represents a linearization of Babenko's equation (\ref{trav-Bab-eq}) at the traveling wave with the even profile $\eta \in C^{\infty}_{\rm per}$. Due to the translational symmetry, we have $\mathcal{L} \eta' = 0$ with $\eta' \in H^1_{\rm per}$. 
Also recall the linear operator 
$$
\mathcal{K} : H^1_{\rm per} \subset L^2_{\rm per} \to L^2_{\rm per},
$$
given by (\ref{operator-K}). 
By using Fourier series, we have $\mathcal{K} 1 = 0$ with $1 \in H^1_{\rm per}$. Furthermore, we obtain the following relations by explicit computations:
\begin{align}
\label{L-on-1}
\mathcal{L} 1 &= -\left(1 + 2 \mathcal{K} \eta\right), \\
\label{L-on-eta}
\mathcal{L} \eta &= \eta - c^2 \mathcal{K} \eta, 
\end{align}
and
\begin{align}
\label{L-on-eta-prime}
\mathcal{L}_{\pm} \eta' &= -c^2 \mathcal{K} \eta + \eta \mathcal{K}\eta + \frac{1}{2} \mathcal{K}\eta^2 = -\eta,
\end{align} 
where $\mathcal{L}_{\pm} : L^2_{\rm per} \to L^2_{\rm per}$ is given by 
(\ref{L-plus-minus}). 

The following lemma gives the sufficient condition for differentiability of the mapping $c \to \eta \in C^{\infty}_{\rm per}$ at any $c \in (1,c_*)$. 

\begin{lemma}
	\label{lem-diff}
	If ${\rm Ker}(\mathcal{L}) = {\rm span}(\eta')$, then the mapping 
	$c \to \eta \in C^{\infty}_{\rm per}$ is continuously differentiable for this value of $c \in (1,c_*)$ and 
	\begin{align}
	\label{der-speed}
	\mathcal{L} \partial_c \eta &= - 2 c \mathcal{K} \eta.
	\end{align}
\end{lemma}

\begin{proof}
If $\eta$ is even, then $\eta'$ is odd. Restricting solutions 
to Babenko's equation (\ref{trav-Bab-eq}) to even functions 
is allowed since $\mathcal{K}$ is parity-preserving 
in the subspace of $H^1_{\rm per}$ spanned by even functions, 
which we denote by $H^1_{\rm per,e}$. 
The linearized operator $\mathcal{L}$ acting on $H^1_{\rm per,e} \subset H^1_{\rm per}$ has  a trivial kernel 
and its spectrum consists of eigenvalues which do not 
accumulate to $0$, see \cite{BDT2000a,BDT2000b}. 
Hence, by the implicit function theorem for $(\eta,c) \in H^1_{\rm per,e} \times \mathbb{R}$, there exists a $C^1$ continuation 
of the solution $\eta \in H^1_{\rm per,e}$ 
to Babenko's equation (\ref{trav-Bab-eq}) in $c$ for this value of 
$c \in (1,c_*)$. By differentiating (\ref{trav-Bab-eq}) in $c$, 
we obtain (\ref{der-speed}), where $\mathcal{K} \eta \perp {\rm Ker}(\mathcal{L})$ with the unique solution $\partial_c \eta \in  H^1_{\rm per,e}$. 
\end{proof}

\begin{remark}
	\label{rem-fold}
	We say that the profile $\eta \in C^{\infty}_{\rm per}$ is 
	at the fold bifurcation if ${\rm Ker}(\mathcal{L}) \neq  {\rm span}(\eta')$ 
	for this value of $c \in (1,c_*)$. The fold bifurcation is included 
	in our study in Section \ref{sec-6}. Bifurcations 
	in Sections \ref{sec-3}, \ref{sec-4}, and \ref{sec-5} hold for ${\rm Ker}(\mathcal{L}) = {\rm span}(\eta')$.
\end{remark}

Related to the profile $\eta \in C^{\infty}_{\rm per}$ of the traveling wave, we define the wave momentum $\mathcal{P}(c)$ and 
the wave energy $\mathcal{H}(c)$ by 
evaluating $P(\psi,\eta)$ and $H(\psi,\eta)$ in (\ref{conserved-momentum}) and (\ref{conserved-Ham}) at $\psi= -c \mathcal{H} \eta$:
\begin{equation}
\label{wave-momentum}
\mathcal{P}(c) = c \langle K \eta, \eta \rangle
\end{equation}
and 
\begin{equation}
\label{wave-energy}
\mathcal{H}(c) = \frac{c^2}{2} \langle K \eta, \eta \rangle + \frac{1}{2}\langle \eta^2, (1+K\eta) \rangle.
\end{equation}
In addition, we introduce the squared $L^2$ norm, which is a part of the energy (\ref{wave-energy}) but is not a conservative quantity in the time evolution of the system (\ref{time-Bab-eq}):
\begin{equation}
\label{wave-norm}
\mathcal{N}(c) := \langle \eta, \eta \rangle. 
\end{equation}
These integral quantities are used in the Taylor expansion of solutions of the spectral problem (\ref{mod-Bab-eq}) near $(\mu,\lambda) = (0,0)$. Some useful relations are derived in the following lemma.

\begin{lemma}
\label{lem-identity}
It follows that
\begin{align}
\label{identity2}    
\mathcal{H}(c) = \frac{5}{6} c\mathcal{P}(c) + \frac{1}{6}\mathcal{N}(c),
\end{align}
where $\mathcal{P}(c)$, $\mathcal{H}(c)$,  and $\mathcal{N}(c)$ are defined by 
(\ref{wave-momentum}), (\ref{wave-energy}), and (\ref{wave-norm}). If ${\rm Ker}(\mathcal{L}) = {\rm span}(\eta')$, then we also have 
\begin{align}
    \label{identity1}
\mathcal{H}'(c) = c\mathcal{P'}(c) = \mathcal{N}'(c) + 5\mathcal{P}(c).
\end{align}
\end{lemma}

\begin{proof}
It follows from (\ref{trav-Bab-eq}) that 
\begin{align*}
c \mathcal{P}(c) - \mathcal{N}(c) = \langle (c^2 \mathcal{K}\eta - \eta), \eta \rangle = \frac{3}{2} \langle \eta^2, \mathcal{K} \eta \rangle.
\end{align*}
Eliminating $\langle \eta^2, \mathcal{K} \eta \rangle$ from 
$$
2 \mathcal{H}(c) = c \mathcal{P}(c) + \mathcal{N}(c) + \langle \eta^2, \mathcal{K} \eta \rangle,
$$
we obtain (\ref{identity2}). 

To derive (\ref{identity1}) under the assumption ${\rm Ker}(\mathcal{L}) = {\rm span}(\eta')$, see Lemma \ref{lem-diff}, we obtain from (\ref{aug-energy}) that 
$$
\Lambda_c(\eta)  = c \mathcal{P}(c) - \mathcal{H}(c).
$$
Differentiating $\Lambda_c(\eta)$ in $c$ and using (\ref{trav-Bab-eq}), yields $\frac{d}{dc} \Lambda_c(\eta) = \mathcal{P}(c)$, 
which implies the first identity in (\ref{identity1}). For the second identity in (\ref{identity1}), we obtain 
    \begin{align*}
c\mathcal{P}'(c) - \mathcal{N}'(c) &= c\langle \mathcal{K}\eta, \eta \rangle + 
2 \langle (c^2 \mathcal{K} \eta - \eta), \partial_c \eta \rangle  \\
        &= \mathcal{P}(c) - 2 \langle \mathcal{L} \eta, \partial_c \eta \rangle \\
        &= \mathcal{P}(c) + 4 c \langle \mathcal{K} \eta, \eta \rangle \\
        &= 5\mathcal{P}(c), 
    \end{align*}
    where we have used~\eqref{L-on-eta}, ~\eqref{der-speed}, and self-adjointness of $\mathcal{L}$ in $L^2_{\rm per}$.
\end{proof}

\begin{remark}
Let $\sigma_0$ be the smallest eigenvalue of $\mathcal{L} : H^1_{\rm per} \subset L^2_{\rm per} \to L^2_{\rm per}$.
Then, it follows from ~\eqref{L-on-eta} that  
\begin{align*}
        1 - \frac{c\mathcal{P}(c)}{\mathcal{N}(c)} = \frac{\langle \mathcal{L} \eta, \eta \rangle}{\langle \eta, \eta \rangle} \geq \sigma_0. 
\end{align*}  
Since it follows from (\ref{L-on-1}) that $\langle \mathcal{L} 1, 1 \rangle = -1 < 0$, we have $\sigma_0 < 0$, which implies that 
$$
c \mathcal{P}(c) \leq (1 - \sigma_0) \mathcal{N}(c).
$$
\end{remark}

\begin{remark}
	By using the second identity in (\ref{identity1}), we obtain 
\begin{align}
\Delta(c) &= 4 \mathcal{P}'(c) \mathcal{N}(c) - c \mathcal{P}'(c) \mathcal{N}'(c) + [\mathcal{N}'(c)]^2 \notag \\
&= 4 \mathcal{P}'(c) \mathcal{N}(c) - 5 \mathcal{P}(c) \mathcal{N}'(c), 
\label{deltaC}
\end{align}
where $\Delta(c)$ appears in Section \ref{sec-3}, see the proof of Theorem \ref{theorem-bif-1}.
\end{remark}

Figure \ref{fig:delta} displays the numerically computed values of $\mathcal{H}'(c)$ defined by (\ref{identity1}) (red solid line) and $\Delta(c)$ defined by (\ref{deltaC}) (green solid line) as functions of the steepness parameter $s$ in the logarithmic scale. 
The vertical solid lines mark zeros of $c'(s)$, for which the mapping $c \to \eta \in C^{\infty}_{\rm per}$ is not continuously differentiable in $c \in (1,c_*)$.
Zeros of $\Delta(c)$ (green line) mark the hourglass bifurcation, see Section \ref{sec-3}. Zeros of $\mathcal{H}'(c)$  (red line) mark the figure-$\infty$ bifurcation, see Section \ref{sec-5}. Zeros of $c'(s)$ (vertical lines) mark the figure-$8$ bifurcation, see Section \ref{sec-6}. 

\begin{figure}[htb!]
	\centering
	\includegraphics[width=0.82\linewidth]{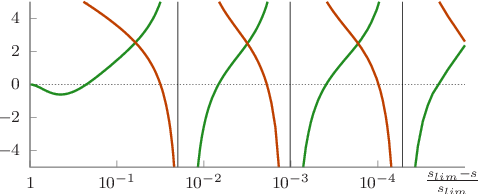}
	\caption{The dependence of $\mathcal{H}'(c)$ (red solid) and $\Delta(c)$ (green solid) versus the wave steepness $s$, with the vertical lines for zeros of $c'(s)$. }
	\label{fig:delta}
\end{figure}

In our numerical computations, we compute $\mathcal{P}(c)$, $\mathcal{P}'(c)$, $\mathcal{H}(c)$, and $\mathcal{H}'(c)$, after which we obtain $\mathcal{N}(c)$ and $\mathcal{N}'(c)$ from the relations~\eqref{identity2} and~\eqref{identity1}.
The values of $\mathcal{P}(c)$ and $\mathcal{H}(c)$ are obtained from 
(\ref{wave-momentum}) and (\ref{wave-energy}), 
The derivative $\mathcal{P}'(c)$ is obtained from 
\begin{align*}
\mathcal{P}'(c) = \langle \mathcal{K}\eta, \eta \rangle + 2c\langle \mathcal{K}\partial_c\eta, \eta \rangle,
\end{align*}
where $\partial_c\eta$ is determined from solving the linear equation (\ref{der-speed}) numerically via the MINRES method to a prescribed tolerance defined as 
$$
\frac{\|\mathcal{L}\partial_c\eta + 2c\mathcal{K}\eta \|_{L^2}}{
	\|2c\mathcal{K}\eta \|_{L^2}} \leq 10^{-10}.
$$ 
From $\mathcal{P}'(c)$, we obtain the derivative $\mathcal{H}'(c) = c \mathcal{P}'(c)$ by using (\ref{identity1}).

Figure \ref{fig:beigs} displays the first (smallest) eigenvalues of the operator $\mathcal{L}$ for co-periodic (green curves) and anti-periodic (red curves) perturbations versus the normalized steepness. In the limit of zero steepness, we obtain from (\ref{eta-small}) and (\ref{c-small}):
$$
\lim_{a \to 0} \mathcal{L} = \mathcal{K}-1.
$$
Therefore, we have a simple eigenvalue at $-1$  and a sequence of double eigenvalues 
at $\{ 0,1,2,\dots \}$ for co-periodic perturbations associated with the Fourier modes $1$ and $e^{\pm i n u}$ for $n \in \mathbb{N}$, respectively, and a sequence of double eigenvalues $\{ -\frac{1}{2}, \frac{1}{2}, \frac{3}{2}, \dots \}$ for anti-periodic perturbations associated with the Fourier modes $e^{\pm \frac{i}{2} (2n-1) u}$ for $n \in \mathbb{N}$. These eigenvalues are continued with respect to the steepness parameter. 

The red solid dots correspond to the subharmonic (double-period) bifurcation of traveling wave solutions, which results in the ellipse bifurcation in the modulational stability problem (\ref{mod-Bab-eq}) described in Section~\ref{sec-4}. The green solid dots correspond to the turning points of the speed $c$ versus steepness $s$, where $c'(s) = 0$, see vertical solid lines on Figure \ref{fig:delta}. At these points, the mapping $c \to \eta \in C^{\infty}_{\rm per}$ is not continuously differentiable, see Lemma \ref{lem-diff}, and the figure-$8$ bifurcation occurs in the spectral problem (\ref{mod-Bab-eq}) described in Section \ref{sec-6}. Each zero crossing for the red curve occurs between a zero crossing by the green curve.

\begin{figure}[htb!]
	\centering
	\includegraphics[width=0.75\linewidth]{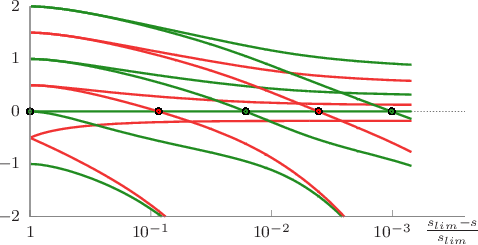}
	\caption{The dependence of the first (smallest) eigenvalues of $\mathcal{L}$ for co-periodic (green) and anti-periodic (red) perturbations. }
	\label{fig:beigs}
\end{figure}

\subsection{The limit $\mu \to 0^{\pm}$ in the modulational stability problem}
\label{sec-2-2}

It is obvious that the limit $\mu \to 0^{\pm}$ in the spectral problem (\ref{mod-Bab-eq}) is different from the co-periodic stability problem (\ref{spec-Bab-eq}) because of the jumps in the continuation of 
the linear operators $\mathcal{H}$, $\mathcal{M}$, and $\mathcal{M}^*$ for 
$\mu \neq 0$. However, the following lemma shows that eigenfunctions of (\ref{mod-Bab-eq}) in the limit $\mu \to 0^{\pm}$ are related to eigenfunctions of (\ref{spec-Bab-eq}) by using a simple invertible transformation. 

\begin{lemma}
	\label{lem-transformation}
	Let $(v,w) \in H^1_{\rm per} \times H^1_{\rm per}$ be an eigenfunction of (\ref{spec-Bab-eq}) for an eigenvalue $\lambda \in \mathbb{C} \backslash \{0\}$. Then, the eigenfunction  $(\hat{v},\hat{w}) \in H^1_{\rm per} \times H^1_{\rm per}$ of (\ref{mod-Bab-eq}) in the limit $\mu \to 0^{\pm}$ exists for the same eigenvalue $\lambda \in \mathbb{C} \backslash \{0\}$ and is given by 
	\begin{equation} 
	\label{transformation}
	\hat{v} = \mp i \langle 1, v \rangle \eta' + v, \qquad 
	\hat{w} = \pm i c \langle 1, v \rangle + w.
	\end{equation}
\end{lemma}

\begin{proof}
	In the limit $\mu \to 0^{\pm}$, we rewrite (\ref{mod-Bab-eq}) in the explicit form:
\begin{equation}
\left\{ \begin{array}{ll}
\mathcal{K} \hat{w} &= \lambda \left( \mathcal{M} \hat{v} \pm i \langle 1, \hat{v} \rangle \eta' \right), \\
\mathcal{L} \hat{v} &= \lambda \left( \mathcal{M}^*  \hat{w} 
- 2 c \mathcal{H} \hat{v} \mp i \langle \eta', \hat{w} \rangle \mp 2 c i \langle 1, \hat{v} \rangle \right).
\end{array}
\right. 
\label{mod-Bab-eq-limit}
\end{equation}
Since $\mathcal{K}1 = 0$ and $\mathcal{L} \eta' = 0$, we have the following orthogonality condition for every solution $(\hat{v},\hat{w}) \in H^1_{\rm per} \times H^1_{\rm per}$ of (\ref{mod-Bab-eq-limit}) with $\lambda \in \mathbb{C} \backslash \{0\}$:
\begin{align}
\label{Fredholm-constraints}
\left. \begin{array}{r}
\langle (1 + 2 \mathcal{K} \eta), \hat{v} \rangle = 0, \\
\langle \eta', \hat{w} \rangle - 2 c \langle  \mathcal{K} \eta, \hat{v} \rangle = 0, 
\end{array}
\right\}
\end{align}	
where we have used $\mathcal{M}^* 1 = 1 + 2 \mathcal{K} \eta$ and $\mathcal{M} \eta' = \eta'$ obtained from (\ref{operator-M}). The two constraints imply 
$\langle \eta', \hat{w} \rangle + c \langle 1, \hat{v} \rangle = 0$, which enables us to rewrite (\ref{mod-Bab-eq-limit}) in the equivalent form:
\begin{equation}
\left\{ \begin{array}{ll}
\mathcal{K} \hat{w} &= \lambda \left( \mathcal{M} \hat{v} \pm i \langle 1, \hat{v} \rangle \eta' \right), \\
\mathcal{L} \hat{v} &= \lambda \left( \mathcal{M}^*  \hat{w} 
- 2 c \mathcal{H} \hat{v} \mp c i \langle 1, \hat{v} \rangle \right).
\end{array}
\right. 
\label{eq-inhomog-problem}
\end{equation}
By using the transformation (\ref{transformation}), we obtain 
$\mathcal{K} \hat{w} = \mathcal{K} w$, 
$\mathcal{L} \hat{v} = \mathcal{L} v$, and $\langle 1, \hat{v} \rangle = \langle 1, v \rangle$. Furthermore, the inhomogeneous terms in (\ref{eq-inhomog-problem}) cancel out under the transformation (\ref{transformation}) and the spectral problem (\ref{eq-inhomog-problem}) reduces to the form (\ref{spec-Bab-eq}) for $(v,w) \in H^1_{\rm per} \times H^1_{\rm per}$.
\end{proof}

\begin{remark}
	\label{rem-multiplicity}
	Due to Lemma \ref{lem-transformation}, the algebraic and geometric multiplicities of every eigenvalue of the co-periodic spectral problem (\ref{spec-Bab-eq}) are equal to those in the spectral problem (\ref{mod-Bab-eq}) as $\mu \to 0^{\pm}$, that is, in the spectral problem (\ref{mod-Bab-eq-limit}). The result holds for $\lambda = 0$ due to the orthogonality conditions (\ref{Fredholm-constraints}) needed for generalized eigenfunctions.
\end{remark}

\section{Methodology}
\label{sec-1-4}

We explain the method that we use for the justification analysis of the four bifurcations considered in Sections \ref{sec-3}, \ref{sec-4}, \ref{sec-5}, and \ref{sec-6}. The spectral bands (\ref{spectral-bands}) are defined by the eigenvalues $\lambda$ of the spectral problem (\ref{mod-Bab-eq}) continued as functions of the Floquet parameter $\mu$. In particular, we are interested in the behavior of these functions near the origin for $\mu$ being close to either $0$ or $\frac{1}{2}$, and we are attempting to match the behavior of the spectral bands seen in Figure \ref{fig-bif} with the normal forms justified via some rigorous theory. 

If we investigate zeros of some Fredholm determinant which give eigenvalues of the purely discrete spectrum of the spectral problem (\ref{mod-Bab-eq}) for a finite $\mu \in \mathbb{B}$, we can use the Weierstrass preparation theorem to pull out the (finite) part that is vanishing, and ignore the rest \cite{Hor73}. This is somewhat standard, though it technically requires some analyticity to use. 
The spectral problem (\ref{mod-Bab-eq}) is analytic in $\mu$ and $\lambda$ for both signs, see Remark \ref{rem-anal}, and the analytic projections  
to the discrete spectrum near $\lambda = 0$ are justified. 

To explain ideas of the justification analysis, we study eigenvalues and eigenvectors of a matrix perturbed by a small parameter $\ve$. This parameter represents either small $|\mu|$ in Sections \ref{sec-3}, \ref{sec-5}, and \ref{sec-6} or $|\frac{1}{2} - \mu|$ in Section \ref{sec-4}. Furthermore, the small deviation of the speed $c$ from the critical value $c_0$ for the four bifurcations also produces perturbation terms, which are represented by another small parameter. In what follows, we require the second small parameter to be scaled consistently with the parameter $\varepsilon$ in order to obtain splittings of degenerate eigenvalues near $\lambda = 0$ at the leading order of the asymptotic theory. 

\subsection{Two double Jordan blocks}
\label{sec-case-1}

This situation represents the bifurcation studied in Section \ref{sec-3}. For simplicity, we take the following normalized matrix
\begin{equation}
\label{A-0}
A_0 = \begin{pmatrix} 0 & 1 & 0 & 0 \\ 0 & 0 & 0 &0 \\ 0 & 0 & 0& 1 \\ 0 & 0 & 0 &0 \end{pmatrix}
\end{equation}
with a quadruple zero eigenvalue. We want to consider how the eigenvalues of 
$$
A(\ve) = A_0 + \ve B + \ve^2 C
$$
unfold for $\ve \neq 0$, where $B$ and $C$ are some given generic matrices 
with entries $(b_{ij})_{1 \leq i,j \leq 4}$ and $(c_{ij})_{1 \leq i,j \leq 4}$.

We first explain the method of direct Puiseux expansions which is used 
in Section \ref{sec-3} (and other sections). Then, we justify the direct 
method by using the Newton polytope method for analysis of the characteristic polynomials. The advantage of the direct Puiseux expansions is that they avoid 
computations of all projections of perturbations to the basis of generalized eigenfunctions and single out only the perturbation terms which lead to splitting 
of eigenvalues near $\lambda = 0$. 

\underline{\bf Direct  Puiseux expansions.} We fix the standard basis $\{ e_1,e_2,e_3,e_4\}$ such that 
\begin{equation*}
A_0 e_1 = 0, \quad A_0 e_2 = e_1, \quad A_0 e_3 = 0, \quad  A_0 e_4 = e_3. 
\end{equation*}
Because of the Jordan block structure of $A_0$, the small eigenvalues of $A(\ve)$ scale like $\mathcal{O}(\ve^{1/2})$ as $\ve \to 0$ \cite{Hin91}. This theory allows us to use Puiseux expansions for eigenvalues 
$$
\lambda(\ve) = \ve^{1/2} \lambda_1 + \ve \lambda_2 + \cO(\ve^{3/2})
$$
and eigenvectors
$$
v = v_0 + \ve^{1/2} v_1 + \ve v_2 + \cO(\ve^{3/2}),
$$
which are inserted into the matrix eigenvalue problem $A(\ve) v = \lambda(\ve) v$.
Collecting powers of $\ve^{1/2}$, we get the following results. 

\underline{\bf Order $\mathcal{O}(\ve^0)$:}
Since 
$$
A_0 v_0 = 0 \quad \Rightarrow \quad  v_0 \in \ker A_0 = \Span{e_1,e_3},
$$
we can write 
\begin{equation}
\label{eq:vect0}
v_0 = \alpha e_1 + \beta e_3,
\end{equation}
with coefficients $\alpha, \beta \in \R$.

\underline{\bf Order $\mathcal{O}(\ve^{1/2})$:}
Since 
\begin{equation*}
A_0 v_1 = \lambda_1 v_0 \quad \Rightarrow \quad  
v_1 \in \ker A_0^2 = \Span{e_1,e_2,e_3,e_4}.
\end{equation*}
we can write 
\begin{equation}
\label{eq:vect1}
v_1 = \lambda_1 (\alpha e_2 + \beta e_4) + u_1,
\end{equation}
where $u_1 \in \ker{A_0}$ is arbitrary. Since $\alpha,\beta$ are arbitrary in (\ref{eq:vect0}), we can set $u_1 \equiv 0$ without loss of generality. 

\underline{\bf Order $\mathcal{O}(\ve^{1})$:}
We obtain 
\begin{equation*}
A_0 v_2 + B v_0 = \lambda_1 v_1 + \lambda_2 v_0. 
\end{equation*}
The Fredholm alternative theorem implies that there exists a solution for $v_2$ if and only if
\begin{equation}
\label{eq:solveps}
\lambda_1 v_1 - B v_0 \in {\rm ran}(A_0) = [\ker(A_0^T)]^{\perp},
\end{equation}
since $\lambda_2 v_0 \in {\rm ran}(A_0)$. Using (\ref{eq:vect0}), (\ref{eq:vect1}), and $\langle e_{2,4}, e_{1,3} \rangle  = 0$ (this is essentially the reason why we can set $u_1 \equiv 0$), we get 
from (\ref{eq:solveps}) due to  $\ker(A_0^T) = \Span{e_2,e_4}$,
\begin{align}
\label{alpha-beta}
\left\{ \begin{array}{l} 
(b_{21} \alpha + b_{23} \beta ) = \lambda_1^2 \alpha, \\ 
(b_{41} \alpha + b_{43} \beta) = \lambda_1^2 \beta.
\end{array} \right. 
\end{align}
A non-trivial pair $(\alpha,\beta)$ exists in (\ref{alpha-beta}) if and only if 
$\lambda_1^2$ is an eigenvalue of the projection matrix 
\begin{equation}
\label{G-1}
G_1 = \begin{pmatrix} b_{21} & b_{23} \\ b_{41} & b_{43} \end{pmatrix} = \begin{pmatrix} \langle e_2, B e_1\rangle &  \langle e_2, B e_3\rangle \\  \langle e_4, B e_1\rangle&  \langle e_4, B e_3\rangle \end{pmatrix}.
\end{equation}
$G_1$ provides a map for the image of $B$ acting on the right eigenspace $\{ e_1, e_3\}$, projected onto the left eigenspace $\{ e_2, e_4\}$. If we denote the eigenvalues of $G_1$ by $s_1$ and $s_2$, then the four eigenvalues $\lambda(\ve)$ of $A(\ve)$ split from the quadruple zero eigenvalue if $\ve \neq 0$ according to the following asymptotic balance:
\begin{equation}
\label{eigenvalue-split-1}
\left\{ \begin{array}{l} 
\lambda(\ve) = \ve^{1/2} \sqrt{ s_1} + \mathcal{O}(\ve), \\  
\lambda(\ve) = -\ve^{1/2}\sqrt{ s_1} + \mathcal{O}(\ve), \\ 
\lambda(\ve) = \ve^{1/2}\sqrt{ s_2} + \mathcal{O}(\ve), \\ 
\lambda(\ve) = -\ve^{1/2}\sqrt{ s_2} + \mathcal{O}(\ve). 
\end{array} \right. 
\end{equation}
For each of these eigenvalues, we have an associated eigenvector $v_0$ in (\ref{eq:vect0}) and an associate generalized eigenvector $v_1$ in (\ref{eq:vect1}). Furthermore, we can find the second-order correction $v_2$ by 
\begin{equation}
\label{eq:secondvec}
v_2 = A_0^T(\lambda_1 v_1 - B v_0) + \lambda_2 (\alpha e_2 + \beta e_4) + u_2
\end{equation}
where $u_2 \in \ker{A_0}$ is arbitrary.

For the splitting of the quadruple zero eigenvalue in Section \ref{sec-3}, we have a simplification compared to the general theory. Namely, the projection matrix $G_1$ in (\ref{G-1}) has 
\begin{equation}
\label{b-coefficients}
b_{21} = b_{23} = b_{41} = 0, \quad b_{43} \neq 0,
\end{equation}
see (\ref{matrix-projections}). As a result, $s_1 = 0$ and $s_2 \neq 0$ in (\ref{eigenvalue-split-1}) so that 
only two of the four eigenvalues $\lambda(\ve)$ of $A(\ve)$
split at the order of $\mathcal{O}(\ve^{1/2})$, whereas 
the other two eigenvalues split at the higher order of $\mathcal{O}(\ve)$.

To consider the higher-order splitting of the eigenvalue $\lambda(\ve) = \mathcal{O}(\ve)$ for $s_1 = 0$ and $s_2 \neq 0$, we introduce the eigenvector of $G_1$ for the zero eigenvalue $s_1 = 0$ as  $(\alpha_1, \beta_1)^T$:
$$
G_1 \begin{pmatrix} \alpha_1 \\ \beta_1 \end{pmatrix} = 0.
$$
This yields $v_0$ in (\ref{eq:vect0}) and since $\lambda_1 = 0$ and $u_1 \equiv 0$, 
then (\ref{eq:vect1}) implies that $v_1 = 0$. Furthermore, (\ref{eq:secondvec}) implies that  
\begin{equation}
\label{eq:vect2}
v_2 = - A_0^T B v_0 + \lambda_2 (\alpha_1 e_2 + \beta_1 e_4) + u_2,
\end{equation}
where $u_2 \in \ker{A_0}$. Unlike previously, we need to assume here that $u_2 \neq 0$ since we get constraints on $u_2$ from (\ref{eq:messy}) below, hence we write $u_2 = \gamma e_1 + \delta e_3$ with new coefficients $\gamma,\delta \in \R$.

\underline{\bf Order $\mathcal{O}(\ve^{3/2})$:} We obtain
$$
A_0 v_3 = \lambda_3 v_0 \quad  \Rightarrow \quad 
v_3 = \lambda_3(\alpha_1 e_2 + \beta_1 e_4) 
+ u_3,
$$
where $u_3 \in \ker{A_0}$ is arbitrary.

\underline{\bf Order $\mathcal{O}(\ve^{2})$:} We obtain 
\begin{equation*}
A_0 v_4 + B v_2 + C v_0 = \lambda_2 v_2 + \lambda_4 v_0. 
\end{equation*}
The Fredholm alternative theorem implies that there exists a solution for $v_4$ if and only if
\begin{equation*}
\lambda_2 v_2 - B v_2 - C v_0 \in {\rm ran}(A_0) = [\ker(A_0^T)]^{\perp}.
\end{equation*}
Substituting $v_2$ from \eqref{eq:vect2} and using the fact that $\ker{A_0^T} = \Span{e_2,e_4}$ we arrive at the pair of conditions: 
\begin{equation}
\label{eq:messy}
\lambda_2^2 \begin{pmatrix} \alpha_1 \\ \beta_1 \end{pmatrix} - \lambda_2 G_2 \begin{pmatrix} \alpha_1 \\ \beta_1 \end{pmatrix} - G_3 \begin{pmatrix} \alpha_1 \\ \beta_1 \end{pmatrix} =  G_1 \begin{pmatrix} \gamma \\ \delta \end{pmatrix},
\end{equation} 
where 
$$
G_2 = \begin{pmatrix} \langle e_2, (A_0^T B + B A_0^T) e_1\rangle &  \langle e_2, (A_0^T B + B A_0^T) e_3\rangle \\  \langle e_4, (A_0^T B + B A_0^T) e_1\rangle&  \langle e_4,(A_0^T B + B A_0^T) e_3\rangle \end{pmatrix}
$$
and
$$
G_3 = \begin{pmatrix} \langle e_2, (B A_0^{T} B - C) e_1\rangle &  \langle e_2, (B A_0^{T} B - C) e_3\rangle \\  \langle e_4, (B A_0^{T} B - C) e_1\rangle&  \langle e_4, (B A_0^{T} B - C) e_3\rangle \end{pmatrix}.
$$
The left-hand side of \eqref{eq:messy} is a pair of polynomials in $\lambda_2$, dependent on the eigenvector $(\alpha_1,\beta_1)^T$ of the zero eigenvalue of $G_1$, which we can denote as $(p_1(\lambda_2),p_2(\lambda_2))^T$. There exists a solution of \eqref{eq:messy} for $(\gamma,\delta)^T$ if and only if $(p_1(\lambda_2),p_2(\lambda_2))^T$ is in the range of $G_1$ i.e. orthogonal to a eigenvector of $G_1^T$ for the zero eigenvalue $s_1 = 0$, denoted by $(\ell_1,\ell_2)^T$:
\begin{equation}
\label{eq:quad}
\ell_1 p_1(\lambda_2) + \ell_2 p_2(\lambda_2) = 0.
\end{equation}
This gives a quadratic equation on $\lambda_2$ with two roots labeled as $\lambda_2^{\pm}$. Combined with the other pair related to $s_2 \neq 0$, see 
(\ref{eigenvalue-split-1}), we obtain four eigenvalues splitting from the quadruple zero eigenvalue if $\ve \neq 0$ according to the following asymptotic balance:
\begin{equation}
\label{eigenvalue-split-2}
\left\{ \begin{array}{l} 
\lambda(\ve) = \ve \lambda_2^+ + \mathcal{O}(\ve^{3/2}), \\  
\lambda(\ve) = \ve \lambda_2^- + \mathcal{O}(\ve^{3/2}), \\ 
\lambda(\ve) = \ve^{1/2}\sqrt{ s_2} + \mathcal{O}(\ve), \\ 
\lambda(\ve) = -\ve^{1/2}\sqrt{ s_2} + \mathcal{O}(\ve). 
\end{array} \right. 
\end{equation}
For the case (\ref{b-coefficients}), the eigenvector of $G_1$ in (\ref{G-1}) for $s_1 = 0$ is $(\alpha_1,\beta_1)^T = (1,0)^T$ so that the second equation 
in the system (\ref{eq:messy}) defines uniquely $\delta$, whereas the first equation in the system (\ref{eq:messy}) gives a quadratic equation 
on $\lambda_2$, which is equivalent to (\ref{eq:quad}). The correction $\gamma$ is not uniquely defined and can be set to $0$.

\underline{\bf Justification with the Newton polytope method.} A suitable algorithm which justifies the direct expansions leading to 
the splitting formulas (\ref{eigenvalue-split-1}) or (\ref{eigenvalue-split-2}) 
is based on the {\em Newton polytope method} for the characteristic polynomials \cite{Kir92}. Let $cp(\lambda,\ve)$ denote the characteristic polynomial of the perturbed matrix $A(\ve)$:
$$
cp(\lambda,\ve) := \det\left(A_0 + \ve B + \ve^2 C - \lambda I \right).
$$
The polynomial $cp(\lambda,\ve)$ is a polynomial in two variables $\lambda$ and $\ve$. Since we are looking at perturbations near zero, we can get a leading approximation of this polynomial via the lower convex hull of its Newton polytope in the $(\lambda,\ve)$ lattice, where points correspond to integer powers \cite{Kir92}. That is, we expand $cp(\lambda,\ve)$ in powers of $\lambda$ and $\ve$ and keep track of the non-zero coefficient in front of powers. For each of the power  $\lambda^i\ve^j$ with a non-zero coefficient, we plot points on an integer lattice $(i,j)$. Then we take the lower-convex hull of these points.

\begin{figure}[htb!]
	\begin{center}
		\includegraphics[scale=0.4]{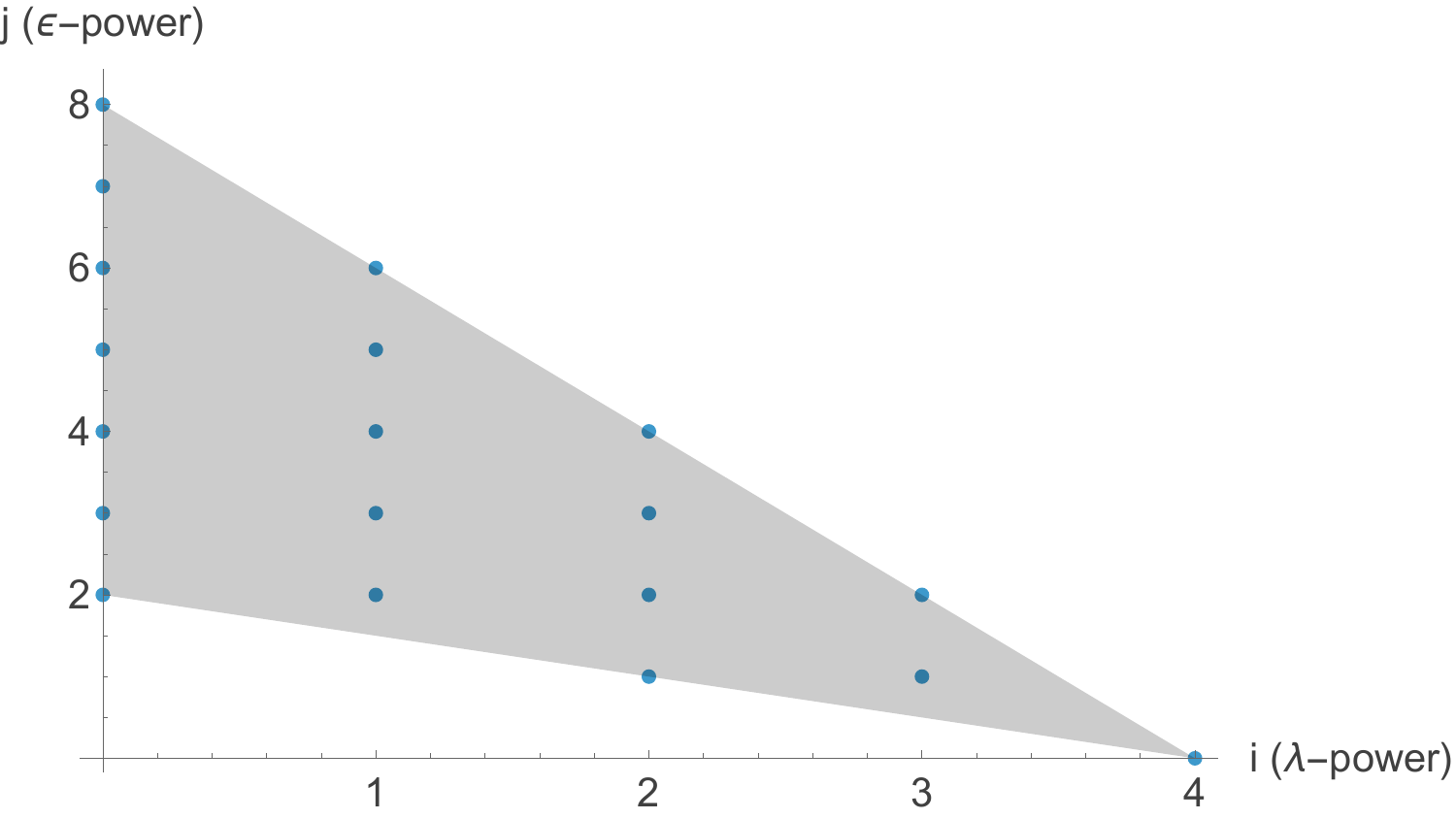}
		\caption{The Newton polytope of $cp(\lambda,\ve) = \det\left(A_0 + \ve B + \ve^2 C - \lambda I \right)$. The lower convex hull is the points in the grey triangle. The boundary consists of the lattice points at $(0,2),(2,1)$ and $(4,0).$}
		\label{fig:newtpoly22gen}
	\end{center}
\end{figure}

Figure \ref{fig:newtpoly22gen} represents the generic case for the matrix $A_0$ in (\ref{A-0}), for which we have 
\begin{equation}
\label{eq:cp1approx}
cp(\lambda,\ve) = \lambda^4 + a_1 \ve\lambda^2  + a_2 \ve^2 + \hot, 
\end{equation}
where $\hot$ denotes terms of the order of $\ve \lambda^3$, $\ve^2 \lambda^2$, $\ve^2 \lambda$, and $\ve^3$. When the coefficients $a_1$ and $a_2$ are nonzero, the lower-convex hull with the boundary connecting $(0,2)$, $(2,1)$, and $(4,0)$ gives the lowest-order approximation to the characteristic polynomial.

The coefficients $a_1$ and $a_2$ are given in terms of the entries of $B$. In certain cases \cite{ME98,MBO97} some formulas are known to exist, though not for the general forms of the matrices. However in our examples, the dimensions are low enough that we can write out the characteristic polynomial and collect terms explicitly. In particular, for the matrix $A_0$ in (\ref{A-0}) perturbed by $\ve B + \ve^2 C$, we can write
\begin{equation*}
\begin{cases}
a_1 = -(b_{21} + b_{43}) = -\tr{G_1}, \\ 
a_2 = b_{21}b_{43} - b_{23}b_{41} = \det(G_1),
\end{cases}
\end{equation*}
in agreement with our previous direct computations, see (\ref{alpha-beta}).  Substituting $\lambda = \lambda_1 \ve^{1/2}$ into \eqref{eq:cp1approx} yields
$$
cp(\ve^{1/2} \lambda_1, \ve) = \ve^2 \left( \lambda_1^4 - \tr{G_1} \lambda_1^2 + \det(G_1) \right) + \cO(\ve^{>2}) 
$$
from which we obtain $\lambda_1 \sim \pm \sqrt{s_{1,2}}$, where $s_1$ and $s_2$ are eigenvalues of $G_1$, and recover the eigenvalue splitting in (\ref{eigenvalue-split-1}).

If we now assume that $s_1 = 0$ and $s_2 \neq 0$, then $a_1 \neq 0$ and $a_2 = 0$, so that we remove the point $(0,2)$ in the lower convex hull and re-compute the approximating polynomial. The characteristic polynomial is now expanded as 
\begin{equation}
\label{eq:cp22dega20}
cp(\lambda,\ve) = \lambda^4 + a_1 \ve \lambda^2  +a_3 \ve^2 \lambda + a_4 \ve^3 +  \hot,
\end{equation}
where $a_1 = -\tr{G_1}$ as before, $a_3$ and $a_4$ have been computed with the symbolic software,
\begin{equation*}
\begin{cases}
a_3 = (b_{11}+b_{22})\,b_{43}-b_{41}(b_{13}+b_{24})
-b_{23}(b_{31}+b_{42})
+ b_{21}(b_{33}+b_{44}), \\
a_4 = b_{41}(b_{13} b_{22} + b_{24} b_{33} - c_{23}) -b_{43}(b_{11} b_{22} + b_{24} b_{31} - c_{21}) \\
\quad \quad 
+ b_{23}(b_{11} b_{42} + b_{31} b_{44} - c_{41})
- b_{21}(b_{13} b_{42} + b_{33} b_{44} - c_{43}),
\end{cases}
\end{equation*}
and $\hot$ denote terms of the order of $\ve \lambda^3$, $\ve^2 \lambda^2$, $\ve^3 \lambda$, and $\ve^4$.

\begin{figure}[htb!]
	\begin{center}
		\includegraphics[scale=0.4]{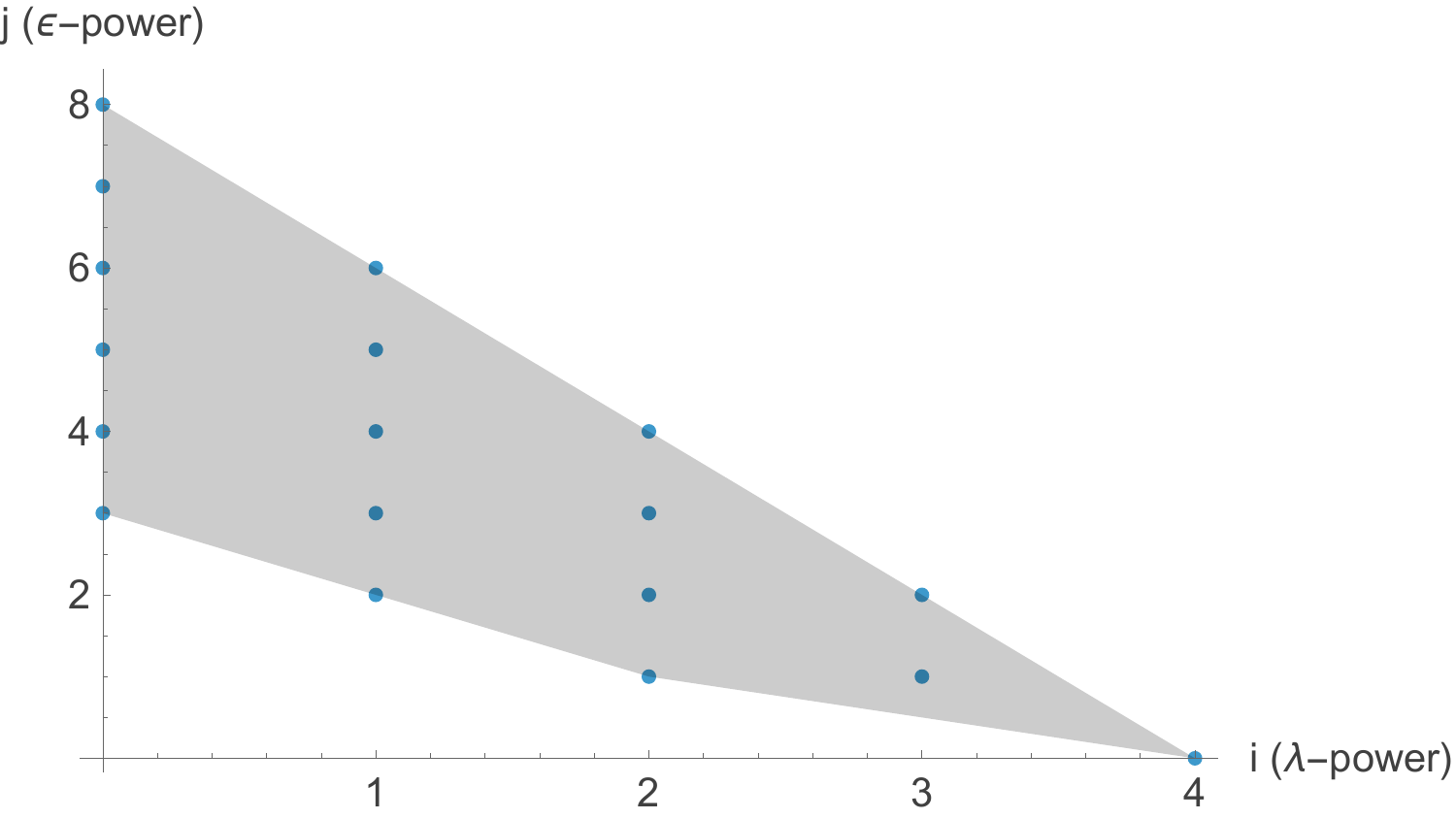} 
		\caption{The Newton polytope of $cp(\lambda,\ve) = \det\left(A_0 + \ve B + \ve^2 C - \lambda I \right)$ when one of the eigenvalues of $G_1$ is zero. The lower convex hull is the points in the grey triangle. The boundary consists of the lattice points at $(0,3),(1,2), (2,1)$ and $(4,0)$.}
		\label{fig:newtpoly22dega20}
	\end{center}
\end{figure}

Figure \ref{fig:newtpoly22dega20} shows the lower-convex hull of the Newton polytope if $a_1, a_3, a_4 \neq 0$. In this case, we have two slopes, so there will be two branches of the eigenvalues. The first is again found by setting $\lambda = \lambda_1 \ve^{1/2}$ which corresponds to the slope $-\frac{1}{2}$ connecting $(2,1)$ and $(4,0)$ in Figure \ref{fig:newtpoly22dega20}. Substituting this into \eqref{eq:cp22dega20}, we have 
$$
cp(\ve^{1/2} \lambda_1,\ve) = \ve^2\left(\lambda_1^4 + a_1 \lambda_1^2\right) + \cO(\ve^{>2}), 
$$
which recovers the first non-zero correction $\lambda_1 \sim \pm \sqrt{s_2}$ for the non-zero eigenvalue of the matrix $G_1$. There are also two eigenvalues that are $0$ here, indicating that the slope $-\frac{1}{2}$ can not see the eigenvalue splitting at this level. We next make the substitution $\lambda = \lambda_2 \ve$, which corresponds to the slope $-1$ connecting $(0,3)$ and $(2,1)$ in Figure \ref{fig:newtpoly22dega20}, and we arrive at 
\begin{align*}
cp(\ve \lambda_2, \ve)  &= \ve^4  \lambda_2^4 + \ve^3 \left( a_1 \lambda_2^2 + a_3 \lambda_2 + a_4 \right) + h.o.t \\ 
&=  \ve^3 \left( a_1 \lambda_2^2 + a_3 \lambda_2 + a_4 \right) + \cO(\ve^{>3}). 
\end{align*}
This recovers the quadratic equation
$$
a_1 \lambda_2^2 + a_3 \lambda_2 + a_4 =0,
$$
which is equivalent to  \eqref{eq:quad}. By denoting the two roots of the quadratic equatio by $\lambda_2^{\pm}$ and combining with the previous case, we have recovered the eigenvalue splitting in (\ref{eigenvalue-split-2}). The eigenvalue splitting as in (\ref{eigenvalue-split-2}) is obtained with direct Puiseux expansions, see (\ref{expansion-1-eps}) with (\ref{quad-eq}) and (\ref{expansion-1-second-eps}) with (\ref{single-eq}) in Section \ref{sec-3}.

\subsection{One double Jordan block}
\label{sec-case-2}

This situation represents the bifurcation studied in Section \ref{sec-4} 
and corresponds to the following normalized matrix
\begin{equation*}
A_0 = \begin{pmatrix} 0 & 1 \\ 0 & 0 \end{pmatrix}
\end{equation*}
with a double zero eigenvalue. Since this case is much simpler than the previous case of two double Jordan blocks, the eigenvalue splitting is obtained from the quadratic characteristic polynomial by using simpler computations, which we do not write here. We only poit out that $G_1 = 0$ (where $G_1$ is a scalar) so that 
the double zero eigenvalue splits as $\lambda = \mathcal{O}(\ve)$ instead of 
$\lambda = \mathcal{O}(\ve^{1/2})$, see (\ref{expansion-2-eps}) with (\ref{quad-eq-anti}) in Section \ref{sec-4}.

\subsection{One quadruple and one double Jordan blocks}
\label{sec-case-3}

This situation represents the bifurcation studied in Section \ref{sec-5} 
and corresponds to the following normalized matrix
\begin{equation}
\label{eq:a042}
A_0 = \begin{pmatrix} 0 & 1 & 0 & 0 & 0 & 0 \\ 0 & 0 & 1 & 0 & 0 & 0 \\0 & 0 & 0 & 1 & 0 & 0 \\ 0 & 0 & 0 & 0 & 0 & 0 \\ 0 & 0 & 0 & 0 & 0 & 1 \\ 0 & 0 & 0 & 0 & 0 & 0 \end{pmatrix}
\end{equation}
with the zero eigenvalue of algebraic multiplicity $6$. The characteristic polynomial is expanded as 
\begin{equation}
\label{eq:cp42gen}
cp(\lambda, \ve) = \lambda^6 + a_1 \ve \lambda^2 + a_2 \ve^2 + h.o.t,
\end{equation}
where 
\begin{equation*}
\begin{cases}
a_1 = - b_{41}, \\ 
a_2 = b_{41} b_{65} - b_{45} b_{61} = \det \begin{pmatrix} b_{41} & b_{61} \\ b_{45} & b_{65} \end{pmatrix}. 
\end{cases}
\end{equation*}

\begin{figure}[htb!]
	\begin{center}
		\includegraphics[scale=0.4]{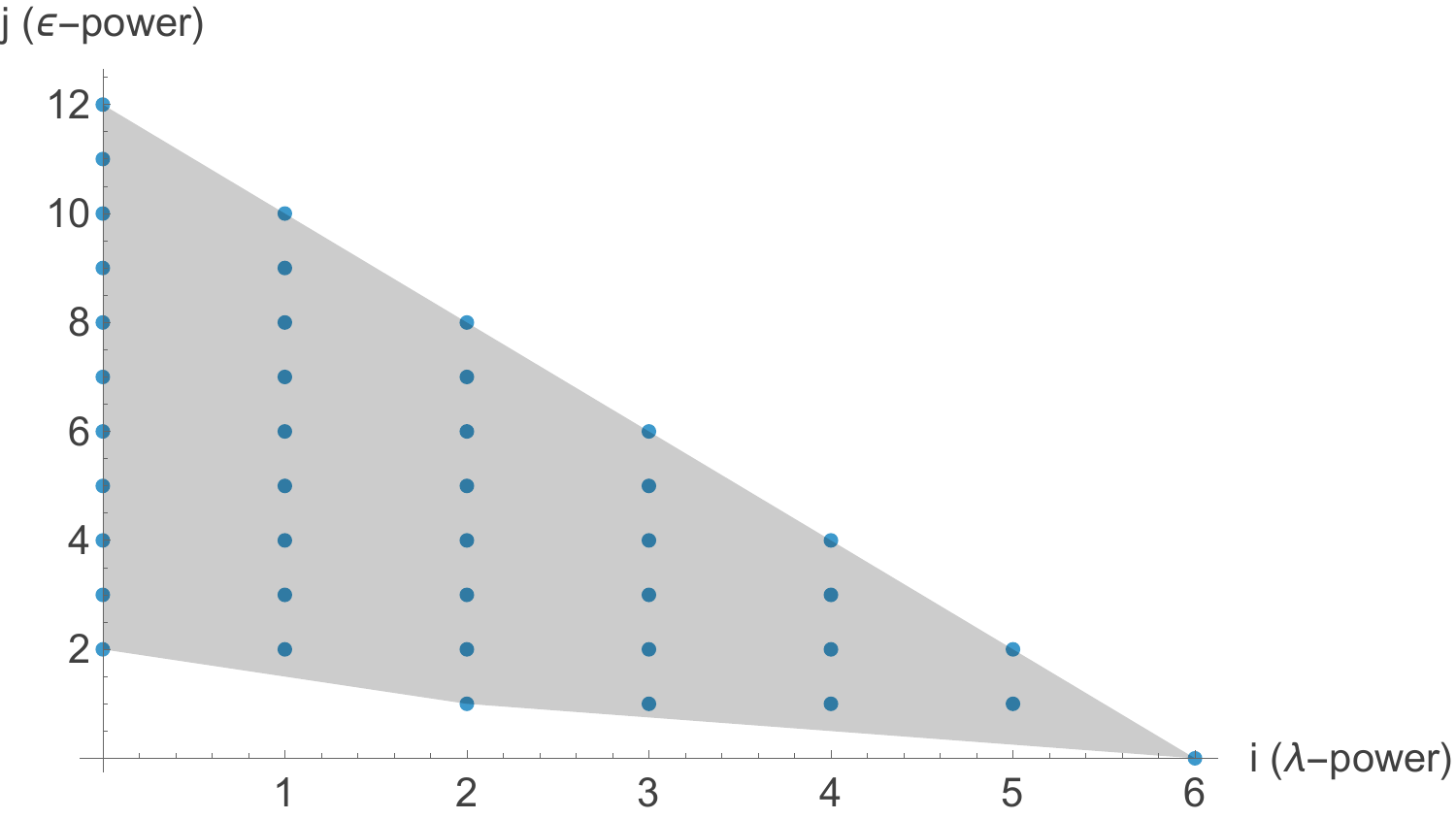} 
		\caption{The Newton polytope of $cp(\lambda,\ve) = \det\left(A_0 + \ve B + \ve^2 C - \lambda I \right)$ with $A_0$ given by \eqref{eq:a042}. The lower convex hull is the points in grey. The boundary consists of the lattice points at $(0,2),(2,1)$ and $(6,0).$}
		\label{fig:newtpolygen42}
	\end{center}
\end{figure}

Figure \ref{fig:newtpolygen42} shows the lower convex hull of the Newton polytope 
for the characteristic polynomial (\ref{eq:cp42gen})  with the boundary consisting of two lines if $a_1, a_2 \neq 0$. Taking into account the two slopes $-\frac{1}{4}$ and $-\frac{1}{2}$ of the lower boundary, we can get four eigenvalues by using the scaling $\lambda = \ve^{1/4} \lambda_{1/2}$ and two eigenvalues by using the scaling $\lambda = \ve^{1/2} \lambda_1$ with the leading-order approximations given by 
\begin{equation}
\label{eigenvalue-split-3}
\lambda_{1/2} \sim  e^{ 2k \pi/4} \sqrt[4]{-a_1} \quad k = 1,2,3,4 \quad \quad \textrm{and} \quad \quad \lambda_1 \sim \pm \sqrt{-a_2/a_1},
\end{equation}
respectively. However, our case of study gives $b_{41} = 0$, 
which is the same coefficient as $b_{21}$ in (\ref{G-1}), see (\ref{b-coefficients}). Therefore, we need to consider the case of $a_1 = 0$, for which the eigenvalue 
splitting (\ref{eigenvalue-split-3}) is not applicable.

\begin{figure}[htb!]
	\begin{center}
		\includegraphics[scale=0.4]{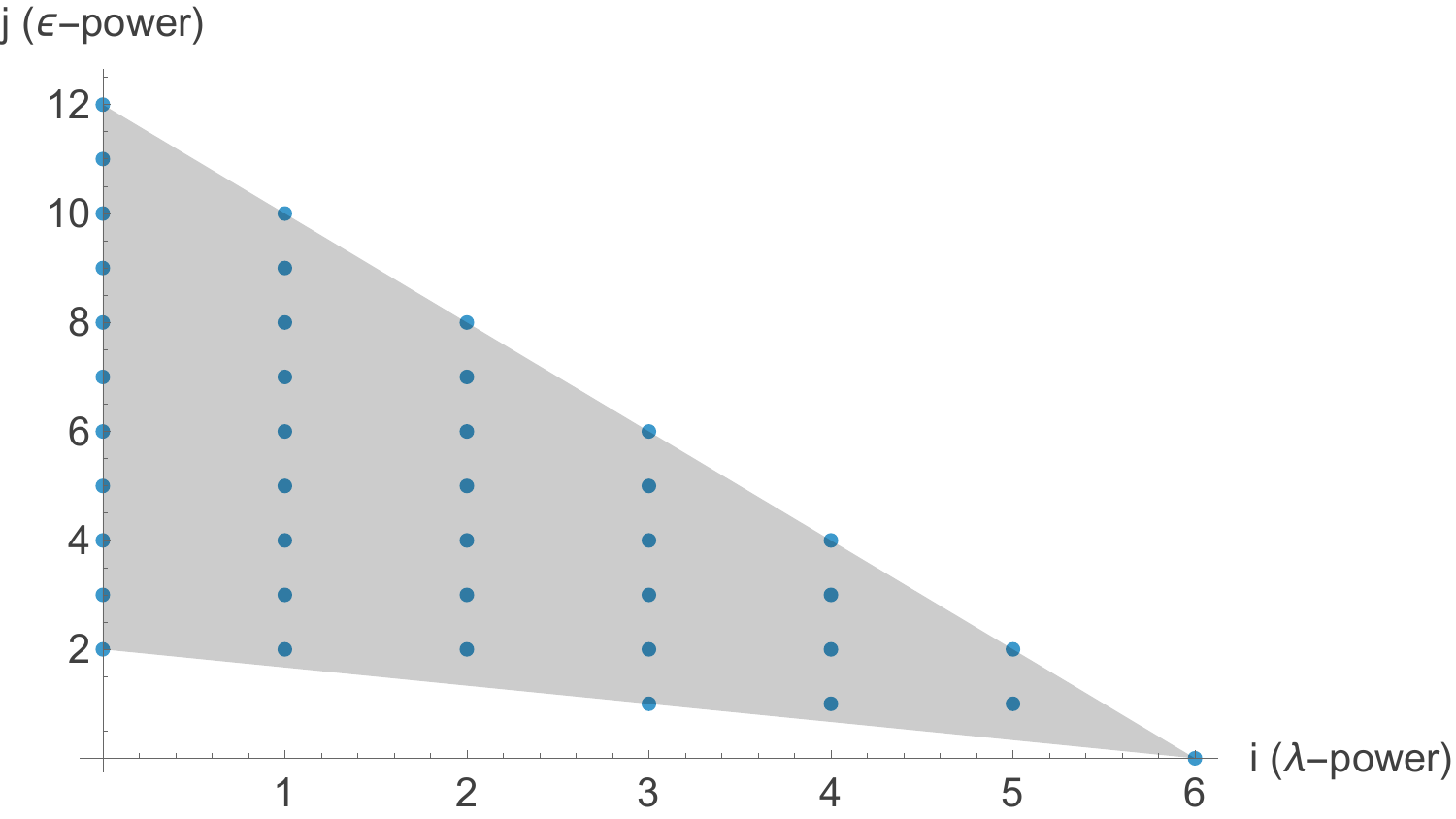} 
		\caption{The same as in Figure \ref{fig:newtpolygen42} but in the case $a_1 = 0$. The lower convex hull is the points in grey. The boundary consists of the lattice points at $(0,2),(3,1)$ and $(6,0).$}
		\label{fig:newtpoly42deg}
	\end{center}
\end{figure}

If $a_1 = 0$, then the characterisic polynomial is expanded as 
\begin{equation}
\label{eq:cp42gen-expand}
cp(\lambda, \ve) = \lambda^6 + a_3 \lambda^3 \ve + a_2 \ve^2 + \hot,
\end{equation}
where $a_2 = - b_{45} b_{61}$ as before and $a_3 = -b_{31}-b_{42}$. 
Figure  \ref{fig:newtpoly42deg} shows the lower convex hull of the Newton polytope for the characteristic polynomial (\ref{eq:cp42gen-expand}) with the boundary given by a single line. This suggests that the six eigenvalues have the correct scaling $\lambda = \ve^{1/3} \lambda_1$, with the leading-order approximation for $\lambda_1$ given by
\begin{equation}
\label{eigenvalue-split-4}
\lambda_1 \sim \frac{e^{2 \pi i k/3}}{2} \left( -a_3 \pm \sqrt{a_3^2 - 4 a_2} \right) \quad k = 1, 2,3.
\end{equation}
However, our case of study gives $b_{61} = 0$, 
which is the same coefficient as $b_{41}$ in (\ref{G-1}), see (\ref{b-coefficients}).  Therefore, we need to consider the double-degenerate case of $a_1 = a_2 = 0$, for which the eigenvalue 
splitting (\ref{eigenvalue-split-4}) is not applicable.

If $a_1 = a_2 = 0$, then the characterisic polynomial is expanded as 
\begin{equation}
\label{eq:cp42gen-expand2}
cp(\lambda, \ve) = \lambda_1^6 + a_3 \ve \lambda^3 + a_4 \ve^2 \lambda + a_5 \ve^3 + \hot,
\end{equation}
where $a_3 = -b_{31} - b_{42}$ as before and $a_4$, $a_5$ are obtained with the symbolic software:
\begin{equation*}
\begin{cases}
a_4 = (-b_{35} - b_{46})\, b_{61}
+ b_{45}\,(-b_{51} - b_{62})
+ (b_{31} + b_{42})\, b_{65}\\
a_5 = b_{65}\Big(
- b_{11} b_{42} - b_{21} b_{43} - b_{31} b_{44}
- b_{46} b_{51} + c_{41}
\Big) \\
\qquad 
+ b_{61}\Big(
b_{15} b_{42} + b_{25} b_{43} + b_{35} b_{44}
+ b_{46} b_{55} - c_{45}
\Big) \\
\qquad 
+ b_{45}\Big(
b_{11} b_{62} + b_{21} b_{63} + b_{31} b_{64}
+ b_{51} b_{66} - c_{61}
\Big).
\end{cases}
\end{equation*}

\begin{figure}[htb!]
	\begin{center}
		\includegraphics[scale=0.4]{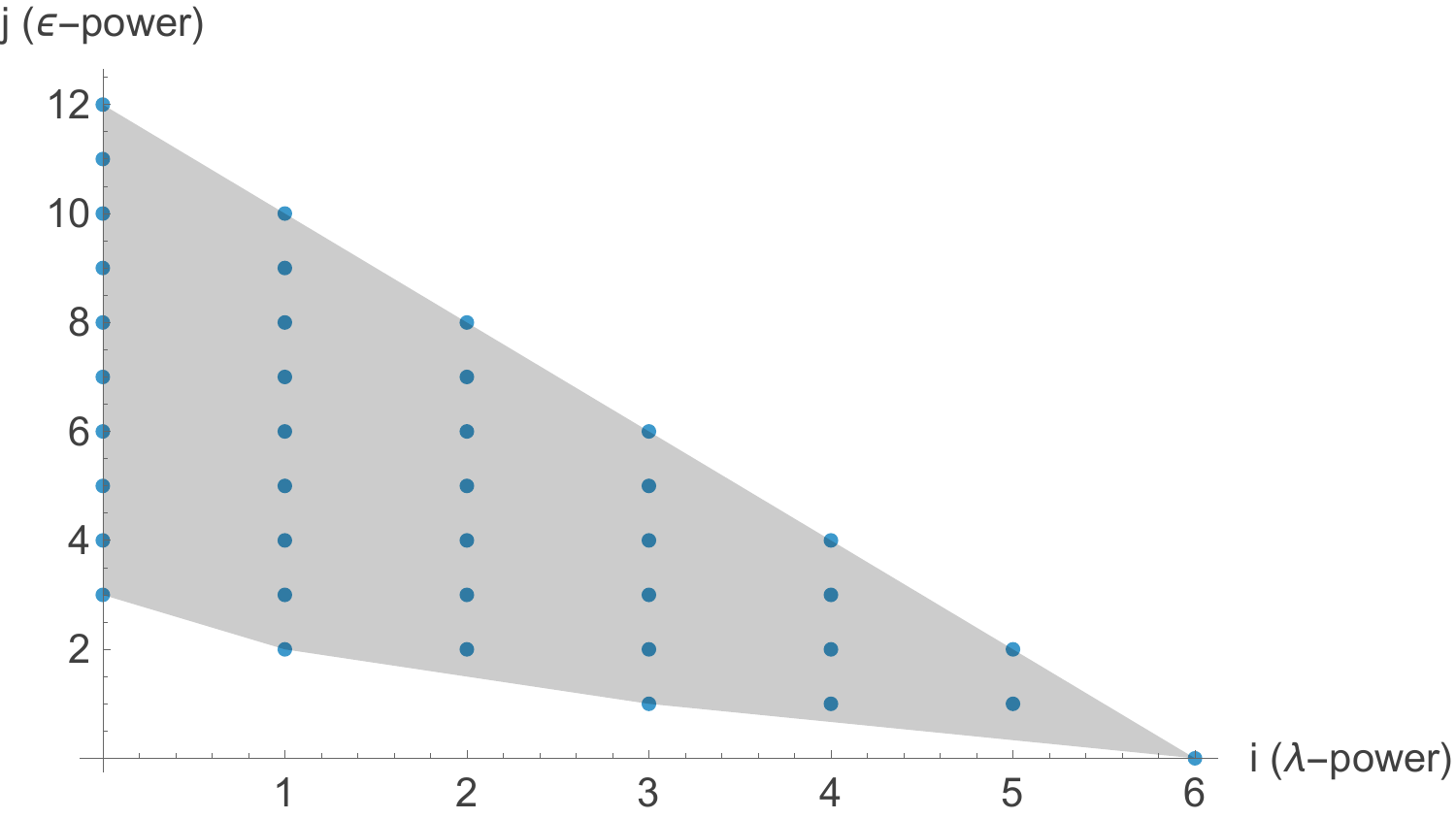} 
		\caption{The same as in Figure \ref{fig:newtpolygen42} but in the case $a_1  = a_2 = 0$. The lower convex hull is the points in grey. The boundary consists of the lattice points at $(0,3),(1,2),(3,1)$ and $(6,0).$}
		\label{fig:newtpoly42deg2}
	\end{center}
\end{figure}

Figure  \ref{fig:newtpoly42deg2} shows the lower convex hull of the Newton polytope for the characteristic polynomial (\ref{eq:cp42gen-expand2}) with the boundary given by three lines. This suggests three different scalings for the six eigenvalues with the leading-order approximation given by 
\begin{equation}
\label{eigenvalue-split-5}
\lambda \sim \ve^{1/3} e^{2 \pi i k/3} \sqrt[3]{-a_3} \quad k = 1,2,3 \quad \textrm{and } \quad \lambda \sim \pm \ve^{1/2} \sqrt{-\frac{a_4}{a_3}} \quad \textrm{and} \quad \lambda \sim - \ve \frac{a_5}{a_4}.
\end{equation}
The eigenvalue splitting as in (\ref{eigenvalue-split-5}) is obtained with direct Puiseux expansions, see (\ref{expansion-3-eps}) with (\ref{normal-form-three}) in Section \ref{sec-5}.

\subsection{Two simple and one double Jordan blocks}
\label{sec-case-4}

This situation represents the bifurcation studied in Section \ref{sec-6} 
and corresponds to the following normalized matrix
\begin{equation}
\label{eq:a0mat211}
A_0 = \begin{pmatrix} 0 & 0 & 0 & 0 \\ 0 & 0 & 0 & 0 \\ 0 & 0 & 0 & 1 \\ 0 & 0 & 0 & 0 \end{pmatrix}
\end{equation}
with the quadruple zero eigenvalue. The characteristic polynomial is expanded as 
\begin{equation}
\label{eq:cp211}
cp(\lambda;\ve) = \lambda^4 + a_1 \ve \lambda^2 + a_2 \ve^2 \lambda + a_3 \ve^3 + \hot,
\end{equation}
where 
\begin{equation*}
\begin{cases}
a_1 = - b_{43} \\ 
a_2 = b_{43}(b_{11} + b_{22}) - b_{13}b_{41} - b_{23}b_{42} \\ 
a_3 = b_{13} b_{22} b_{41}
- b_{12} b_{23} b_{41}
- b_{13} b_{21} b_{42}
+ b_{11} b_{23} b_{42}
+ b_{12} b_{21} b_{43}
- b_{11} b_{22} b_{43},
\end{cases}
\end{equation*}
from which 
$$
a_3 = - \det\begin{pmatrix}
b_{11} & b_{12} & b_{13}\\
b_{21} & b_{22} & b_{23}\\
b_{41} & b_{42} & b_{43}\\
\end{pmatrix}.
$$

\begin{figure}[htb!]
	\begin{center}
		\includegraphics[scale=0.4]{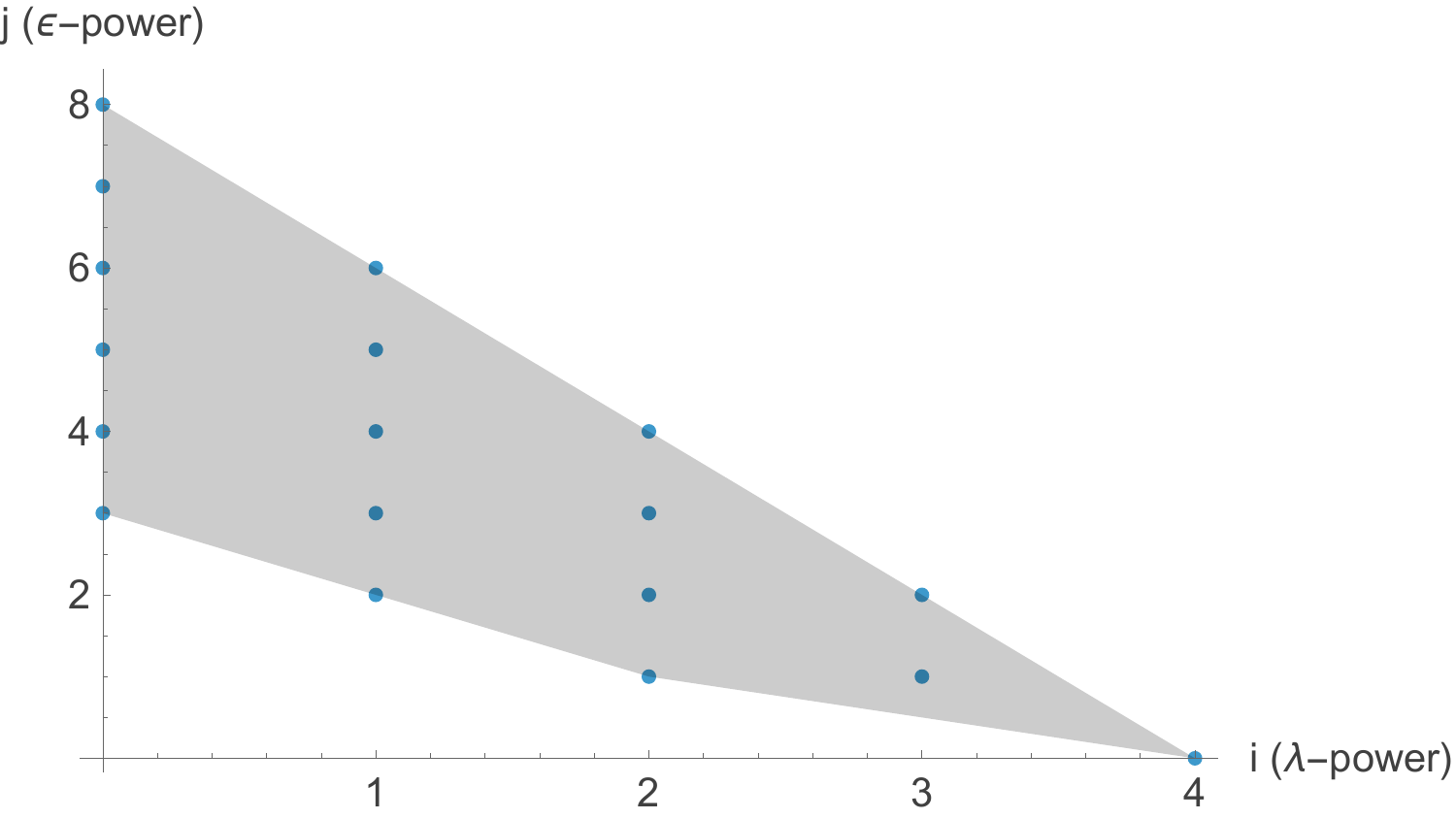} 
		\caption{The Newton polytope of $cp(\lambda,\ve) = \det\left(A_0 + \ve B + \ve^2 C - \lambda I \right)$ where $A_0$ is given by \eqref{eq:a0mat211}. The lower convex hull is the points in the grey. The boundary consists of the lattice points at $(0,3),(1,2),(2,1)$ and $(4,0).$}
		\label{fig:newtonpoly211}
	\end{center}
\end{figure}

Figure \ref{fig:newtonpoly211} shows the lower convex hull of the Newton polytope for the characteristic polynomial (\ref{eq:cp211}) with the boundary consisting of two lines if $a_1, a_2, a_3 \neq 0$. For the line with the slope $-\frac{1}{2}$, we  set $\lambda = \lambda_1 \ve^{1/2}$ and substitute into \eqref{eq:cp211} to get 
$$
cp(\ve^{1/2} \lambda_1, \ve) = \ve^2 \lambda_1^2( \lambda_1^2 + a_1) + \cO(\ve^{>2}),
$$
which defines two roots $\lambda_1 \sim \pm \sqrt{b_{43}}$. For the line with the slope $-1$, we set $\lambda = \lambda_2 \ve$ and substitute into \eqref{eq:cp211} to get
$$
cp(\ve \lambda_2, \ve) = \ve^3( a_1 \lambda_2^2 + a_2 \lambda_2 + a_3) + \cO(\ve^{>3}),  
$$ 
which defines two roots near the two roots of $a_1 \lambda_2^2+ a_2 \lambda_2 + a_3 = 0$. Denoting these two roots by $\lambda_2^{\pm}$ and combining with the previous case, we obtain the eigenvalue splitting given by 
\begin{equation}
\label{eigenvalue-split-6}
\lambda \sim \pm \ve^{1/2} \sqrt{b_{43}} \quad \textrm{and } \quad \lambda \sim  \ve \lambda_2^{\pm}.
\end{equation}
In our case of study in Section \ref{sec-6}, we have $a_1 = -b_{43} \neq 0$, see (\ref{b-coefficients}). However, we also have $a_2^2 - 4 a_1 a_3 = 0$ so that $\lambda_2^+ = \lambda_2^-$ is a double root of $a_1 \lambda_2^2+ a_2 \lambda_2 + a_3 = 0$, see (\ref{matrix-projections-extended}) and (\ref{double-eigenvalue}). 
As a result, there exists a degeneracy in the eigenvalue splitting as in (\ref{eigenvalue-split-6}) at the order of $\mathcal{O}(\ve)$. 

To unfold the degeneracy of the double root, we look at the next terms in the expansion of the characteristic polynomial $cp(\lambda,\ve)$. Let 
$$
\lambda_2^* := - \frac{a_2}{2 a_1}
$$
denote this double root for $a_2^2 - 4 a_1 a_3 = 0$. The characteristic polynomial can be rewritten in this setting as
\begin{align}
cp(\ve \lambda_2, \ve) & = \ve^3 \left( a_1 \lambda_2^2 + a_2 \lambda_2 + a_3 + \cO(\ve) \right) \notag \\
& = \ve^3 \left( a_1 (\lambda_2 - \lambda_2^*)^2 + \ve R(\lambda_2) + \cO(\ve^{>1}) \right),\label{eq:higherord}
\end{align}
where $R(\lambda_2)$ is a quartic polynomial corresponding to the vertices $(0,4),(1,3),(2,2),(3,1),(4,0)$ in the Newton polytope. Writing 
$$
R(\lambda_2) := \lambda_2^4 + d_1 \lambda_2^3 + d_2 \lambda_2^2 + d_3 \lambda_2 + d_4
$$
the coefficients $d_i$ are found explicitly by using the symbolic software:
\begin{equation*}
\begin{cases}
d_1 & = - b_{11} - b_{22} - b_{33} - b_{44}, \\ 
d_2 & =
- b_{12} b_{21}
+ b_{11} b_{22}
- b_{13} b_{31}
- b_{23} b_{32}
- b_{14} b_{41}
- b_{24} b_{42}
- b_{34} b_{43}
+ b_{33} b_{44} \\ 
& \quad \quad 
+ b_{11} ( b_{33} + b_{44} )
+ b_{22} ( b_{33} + b_{44} )
- c_{43} \\  
d_3 & = b_{13} b_{22} b_{31}
- b_{12} b_{23} b_{31}
- b_{13} b_{21} b_{32}
+ b_{11} b_{23} b_{32} \\[4pt]
& \quad \quad  + b_{14} b_{22} b_{41}
- b_{12} b_{24} b_{41}
- b_{13} b_{34} b_{41}
- b_{14} b_{21} b_{42}
+ b_{11} b_{24} b_{42}
- b_{23} b_{34} b_{42} \\[4pt]
& \quad \quad  + b_{11} b_{34} b_{43}
+ b_{22} b_{34} b_{43}
+ b_{11} b_{22} ( - b_{33} - b_{44} )
+ b_{12} b_{21} ( b_{33} + b_{44} ) \\[6pt]
& \quad \quad  + b_{41} ( b_{14} b_{33} - c_{13} )
+ b_{43} ( - b_{14} b_{31} - b_{24} b_{32} + c_{11} + c_{22} )
+ b_{42} ( b_{24} b_{33} - c_{23} ) \\[4pt]
& \quad \quad  + b_{13} ( b_{31} b_{44} - c_{41} )
+ b_{23} ( b_{32} b_{44} - c_{42} )
+ b_{11} ( - b_{33} b_{44} + c_{43} )
+ b_{22} ( - b_{33} b_{44} + c_{43} )\\ 
d_4 & = b_{13} b_{22} b_{34} b_{41}
- b_{12} b_{23} b_{34} b_{41}
- b_{13} b_{21} b_{34} b_{42}
+ b_{11} b_{23} b_{34} b_{42}
+ b_{12} b_{21} b_{34} b_{43}
- b_{11} b_{22} b_{34} b_{43} \\[4pt]
& + b_{22} b_{43} ( b_{14} b_{31} - c_{11} )
+ b_{23} b_{42} ( - b_{14} b_{31} + c_{11} )
+ b_{23} b_{41} ( b_{14} b_{32} - c_{12} )
+ b_{21} b_{43} ( - b_{14} b_{32} + c_{12} ) \\[4pt]
& + b_{21} b_{42} ( b_{14} b_{33} - c_{13} )
+ b_{22} b_{41} ( - b_{14} b_{33} + c_{13} )
+ b_{13} b_{42} ( b_{24} b_{31} - c_{21} )
+ b_{12} b_{43} ( - b_{24} b_{31} + c_{21} ) \\[4pt]
& + b_{11} b_{43} ( b_{24} b_{32} - c_{22} )
+ b_{13} b_{41} ( - b_{24} b_{32} + c_{22} )
+ b_{12} b_{41} ( b_{24} b_{33} - c_{23} )
+ b_{11} b_{42} ( - b_{24} b_{33} + c_{23} ) \\[4pt]
& + b_{12} b_{23} ( b_{31} b_{44} - c_{41} )
+ b_{13} b_{22} ( - b_{31} b_{44} + c_{41} )
+ b_{13} b_{21} ( b_{32} b_{44} - c_{42} )
+ b_{11} b_{23} ( - b_{32} b_{44} + c_{42} ) \\[4pt]
& + b_{11} b_{22} ( b_{33} b_{44} - c_{43} )
+ b_{12} b_{21} ( - b_{33} b_{44} + c_{43} )
\end{cases}
\end{equation*}
Now we set $\lambda_2 = \lambda_2^* + \eta$ and expand (\ref{eq:higherord}): 
$$
cp(\ve \lambda_2, \ve) = \ve^3\left( a_1 \eta^2 + \ve \left( R(\lambda_2^*) + \eta R'(\lambda_2^*) + \frac{1}{2} \eta^2 R''(\lambda_2^*) + \frac{1}{6}\eta^3R'''(\lambda_2^*) + \eta^4 \right) \right) + \mathcal{O}(\ve^{>3}).
$$ 
Collecting terms yields the balance  $\eta^2 \sim \ve$ or $\eta \sim \ve^{1/2} \lambda_3$. Substituting this in the previous expression and collecting the leading order, we obtain 
$$
a_1 \lambda_3^2 + R(\lambda_2^*) = 0 \quad \quad \Rightarrow \quad \quad  \lambda_3 = \pm \sqrt{\frac{R(\lambda_2^*)}{a_1}}. 
$$
Putting this all together now, we have that the double eigenvalue at the order of $\mathcal{O}(\ve)$ splits at the order of $\mathcal{O}(\ve^{3/2})$ if $R(\lambda_2^*) \neq 0$ like 
\begin{equation} 
\label{eigenvalue-split-7}
\lambda \sim \pm \ve^{1/2} \sqrt{b_{43}} \quad \textrm{and } \quad  \lambda \sim \ve \lambda_2^*  \pm \ve^{3/2} \sqrt{\frac{R(\lambda_2^*)}{a_1}}.
\end{equation}
The eigenvalue splitting (\ref{eigenvalue-split-7}) is obtained with direct Puiseux expansions, see (\ref{expansion-mu-4}) with (\ref{double-eigenvalue-again}) and (\ref{lambda-3}) in Section \ref{sec-6}.

\section{Degeneration of figure-8 bands resulting in their vertical slope}
\label{sec-3}

The corresponding  bifurcation is seen in panels (b)--(d) of Figure \ref{fig-bif}.
Due to its shape, it can be called the hourglass bifurcation.
Before the bifurcation [panel (b)], the spectral bands near $\lambda = 0$ form a typical figure-$8$ pattern. At the bifurcation [panel (c)], the slopes of the figure-$8$ bands become vertical at $\lambda = 0$. After the bifurcation 
[panel (d)], the figure-$8$ bands break into two instability bubbles along the imaginary axis. 

Section \ref{sec-3-1} develops the criterion for this bifurcation. Section \ref{sec-3-2} gives the derivation of the normal form, 
which explains how the bifurcation is unfolded. Section \ref{sec-3-3} describes 
numerical approximations of the coefficients of the normal form, 
from which we are able to reproduce the bifurcation pattern seen in panels (b)--(d) of Figure \ref{fig-bif}.

\subsection{Criterion of the bifurcation} 
\label{sec-3-1}

We study first the geometric and algebraic multiplicity of the 
eigenvalue $\lambda = 0$ in the spectral problem (\ref{mod-Bab-eq}) 
for $\mu = 0$. Recall that the two signs in (\ref{mod-Bab-eq}) are defined for the two signs of $\mu$, hence, $\mu = 0$ stands for the limit $\mu \to 0^{\pm}$, 
given by the spectral problem (\ref{mod-Bab-eq-limit}), see Remark \ref{rem-multiplicity}.

For the geometric multiplicity of the zero eigenvalue, we add the following generic assumption. 

\begin{assumption} 
	\label{ass-kernel}
	${\rm Ker}(\mathcal{L}) = {\rm span}(\eta')$ for this value of $c \in (1,c_*)$.
\end{assumption}

Since $\mathcal{K} 1 = 0$ and $\mathcal{L} \eta' = 0$ are satisfied for every $c \in (1,c_*)$, Assumption \ref{ass-kernel} implies that the spectral 
problem (\ref{mod-Bab-eq-limit}) for $\lambda = 0$  admits a
two-dimensional kernel spanned as 
\begin{equation}
\label{kernel}
\left( \begin{array}{c} \hat{v}_0 \\ \hat{w}_0 \end{array} \right) = a_1 \left( \begin{array}{c} \eta' \\ 0 \end{array} \right) + a_2 \left( \begin{array}{c} 0 \\ 1 \end{array} \right),
\end{equation}
where $(a_1,a_2) \in \mathbb{R}^2$. Due to the Hamiltonian symmetry in Lemma \ref{lem-spectrum}, the generalized kernel of the spectral problem (\ref{mod-Bab-eq-limit}) is at least four-dimensional with at least two generalized eigenfunctions. These generalized eigenfunctions were obtained in \cite{DP2025} for the spectral problem (\ref{spec-Bab-eq}) and can be extended to the spectral problem (\ref{mod-Bab-eq-limit}) by using the transformation (\ref{transformation}) in Lemma \ref{lem-transformation} as follows:
\begin{equation}
\label{kernel-generalized}
\left( \begin{array}{c} \hat{v}_g \\ \hat{w}_g \end{array} \right) = a_1 
\left[ \left( \begin{array}{c} -\partial_c \eta \\ \mathcal{H} \eta \end{array} \right) \pm i \langle 1, \partial_c \eta \rangle 
\left( \begin{array}{c} \eta' \\ -c \end{array} \right) 
\right] 
+ a_2 \left[ \left( \begin{array}{c} -1 \\ 0 \end{array} \right) 
\pm i \left( \begin{array}{c} \eta' \\ -c \end{array} \right) \right],
\end{equation}
where the smoothness of the mapping $c \mapsto \eta \in C^{\infty}_{\rm per}$ is due to Assumption \ref{ass-kernel}, see Lemma \ref{lem-diff}. It can be checked by using (\ref{L-on-1}) and (\ref{der-speed}) that $(\hat{v}_g,\hat{w}_g)$ satisfies the linear inhomogeneous system
\begin{equation}
\left\{ \begin{array}{ll}
\mathcal{K} \hat{w}_g &=  \mathcal{M} \hat{v}_0, \\
\mathcal{L} \hat{v}_g &= \mathcal{M}^* \hat{w}_0 -  2 c \mathcal{H} \hat{v}_0.
\end{array}
\right. 
\label{gen-Bab-eq}
\end{equation}	
The system (\ref{gen-Bab-eq}) follows by differentiating (\ref{mod-Bab-eq-limit}) in $\lambda$ at $\lambda = 0$ and using $\langle 1, \eta' \rangle = 0$. For the algebraic multiplicity of the zero eigenvalue, we add the following generic assumption.

\begin{assumption}
	\label{ass-first-bif}
	$\mathcal{P}'(c) \neq 0$ for this value of $c \in (1,c_*)$.
\end{assumption}

It is shown in \cite{DP2025} that Assumptions \ref{ass-kernel} and \ref{ass-first-bif} imply that the generalized kernel of the spectral problem (\ref{mod-Bab-eq-limit}) is exactly four-dimensional.

We now formulate and prove the criterion for the corresponding bifurcation. 

\begin{theorem}
	\label{theorem-bif-1}
Let Assumptions \ref{ass-kernel} and \ref{ass-first-bif} be true and set 
\begin{equation}
\label{Delta}
\Delta(c) := 4 \mathcal{P}'(c) \mathcal{N}(c) - 5 \mathcal{P}(c) \mathcal{N}'(c).
\end{equation}
There exist four spectral bands of the modulational stability problem (\ref{mod-Bab-eq}) intersecting $\lambda = 0$. 
Two of the spectral bands are subsets of $i \mathbb{R}$ and non-degenerate if $\Delta(c) > 0$, tangential to $i \mathbb{R}$ and degenerate if $\Delta(c) = 0$, and cross $\lambda = 0$ along the straight lines in all four complex quadrants of the complex $\lambda$-plane and non-degenerate if $\Delta(c) < 0$. The other two spectral bands are subsets of $i \mathbb{R}$ and non-degenerate, independently of $c \in (1,c_*)$.
\end{theorem}

\begin{proof}
Substituting (\ref{kernel}) into (\ref{mod-Bab-eq}) yields the residual terms 
$$
\left\{ \begin{array}{l} 
i \mu \mathcal{H}_{\pm} \hat{w}_0 + \lambda \mathcal{M}_{\pm} \hat{v}_0, \\
i \mu \mathcal{L}_{\pm} \hat{v}_0 + \lambda (\mathcal{M}_{\pm}^* \hat{w}_0 - 2 c \mathcal{H}_{\pm} \hat{v}_0).
\end{array} \right.
$$
Projecting the first equation to ${\rm Ker}(\mathcal{K}) = {\rm span}(1)$ in the first equation  and the second equation to ${\rm Ker}(\mathcal{L}) = {\rm span}(\eta')$  yields 
\begin{align}
\label{system-tech}
\left\{ \begin{array}{l}
    i\mu \langle 1, \mathcal{H}_{\pm} \hat{w}_0 \rangle + \lambda \langle 1, \mathcal{M}_{\pm} \hat{v}_0 \rangle,  \\
    i\mu \langle \eta', \mathcal{L}_{\pm} \hat{v}_0 \rangle 
    + \lambda \langle \eta', (\mathcal{M}_{\pm}^* \hat{w}_0 - 2 c \mathcal{H}_{\pm} \hat{v}_0) \rangle,
    \end{array} \right.
\end{align}
We substitute (\ref{kernel}) into (\ref{system-tech}) to obtain 
\begin{align}
\label{matrix-projections}
\left( \begin{matrix} 
0 & \mp \mu \\
0 & 0 
\end{matrix} \right) \left( \begin{matrix} a_1 \\ a_2 \end{matrix} \right), 
\end{align}
where we have used $\langle \eta', \eta \rangle = \langle 1, \eta' \rangle = \langle \eta', \mathcal{K} \eta \rangle = 0$, 
$\mathcal{M}_{\pm} \eta' = \eta'$, $\mathcal{L}_{\pm} \eta' = -\eta$, 
see (\ref{H-mu}), (\ref{M-mu}), and (\ref{L-on-eta-prime}). 
The first-order terms in $\lambda$ have zero projections due to the 
existence of the generalized kernel given by (\ref{kernel-generalized}). 
The first-order terms in $\mu$ has one nonzero projection and one zero projection.

In terms of analysis in Section \ref{sec-case-1}, we can associate 
$\mu$ as $\ve$ and obtain (\ref{b-coefficients}) from (\ref{matrix-projections}). 
This suggest the eigenvalue splitting as in (\ref{eigenvalue-split-2}) 
with two eigenvalues of the order $\lambda = \mathcal{O}(\mu)$  
related to the subspace $(a_1,a_2) = (1,0)$ and
two eigenvalues of the order $\lambda = \mathcal{O}(\mu^{1/2})$ related to the subspace $(a_1,a_2) = (0,1)$. By using direct Puiseux expansions, it suffices to consider these two subspaces separately from each other, with two different expansions in powers of $\mu$.  In what follows, we give the computational details which imply the assertion. Justification of the asymptotic expansions is based on the methods from \cite{Hin91,Hor73,Kir92} explained in Section \ref{sec-case-1}.

\vspace{0.2cm}

\underline{Subspace $(a_1,a_2) = (1,0)$:} We introduce the asymptotic expansions
\begin{equation}
\label{expansion-1-eps}
\mu = \epsilon \mu_1, \quad \lambda = \epsilon \lambda_1 + \mathcal{O}(\epsilon^2),
\end{equation}
and 
\begin{equation}
\label{expansion-1}
\left( \begin{array}{c} \hat{v} \\ \hat{w} \end{array} \right) = \left( \begin{array}{c} \eta' \\ 0 \end{array} \right) + \epsilon \left( \begin{array}{c} \hat{v}_1 \\ \hat{w}_1 \end{array} \right) + \epsilon^2 \left( \begin{array}{c} \hat{v}_2 \\ \hat{w}_2 \end{array} \right) + \mathcal{O}(\epsilon^3),
\end{equation}
where $\epsilon > 0$ is a small parameter, $\mu_1 \neq 0$ is fixed, and the correction terms $(\hat{v}_1, \hat{w}_1)$ and $(\hat{v}_2, \hat{w}_2)$ are to be computed recursively in function space $H^1_{\rm per} \times H^1_{\rm per}$. 

\underline{\bf Order $\mathcal{O}(\epsilon)$:} We get the linear inhomogeneous system 
\begin{equation}
\left\{ \begin{array}{ll}
\mathcal{K} \hat{w}_1 &=  \lambda_1 \mathcal{M} \eta', \\
\mathcal{L} \hat{v}_1 &= -  2 c \lambda_1 \mathcal{H} \eta' + i \mu_1 \mathcal{L}_{\pm} \eta' ,
\end{array}
\right. 
\label{level-one}
\end{equation}
where the plus/minus signs correspond to ${\rm sgn}(\mu_1) = \pm 1$. 
Since $\mathcal{L}_{\pm} \eta' = -\eta$ by (\ref{L-on-eta-prime}) and 
\begin{equation}
\label{Upsilon}
\mathcal{L}	\Upsilon = - \eta, \quad \Upsilon := \frac{c}{2} \partial_c \eta - \eta,
\end{equation}
by (\ref{L-on-eta}) and (\ref{der-speed}), we use (\ref{kernel}) and (\ref{kernel-generalized}) and obtain the solution in the form 
\begin{align}
\left( \begin{array}{c} \hat{v}_1 \\ \hat{w}_1 \end{array} \right) = \lambda_1 
\left[ \left( \begin{array}{c} -\partial_c \eta \\ \mathcal{H} \eta \end{array} \right) \pm i \langle 1, \partial_c \eta \rangle 
\left( \begin{array}{c} \eta' \\ -c \end{array} \right) 
\right]
+ i \mu_1 \left( \begin{array}{c} \Upsilon \\ 0 \end{array} \right) 
 + a \left( \begin{array}{c} 0 \\ 1 \end{array} \right),\label{level-one-solution}
\end{align}
where $a \in \mathbb{R}$ is arbitrary at this order and the projection 
to $(\eta',0)^T$ is not included since the eigenfunctions are defined up to arbitrary scalar multiplication in (\ref{expansion-1}). 

\underline{\bf Order $\mathcal{O}(\epsilon^2)$:} We get the linear inhomogeneous system 
\begin{equation}
\left\{ \begin{array}{ll}
\mathcal{K} \hat{w}_2 &=  \lambda_1 \left( \mathcal{M} \hat{v}_1 \pm i \langle 1, \hat{v}_1 \rangle \eta' \right) + i \mu_1 \left( \mathcal{H} \hat{w}_1 \pm i \langle 1, \hat{w}_1 \rangle \right), \\
\mathcal{L} \hat{v}_2 &= \lambda_1 \left( \mathcal{M}^* \hat{w}_1 
-  2 c  \mathcal{H} \hat{v}_1 \mp i \langle \eta', \hat{w}_1 \rangle \mp 2 c i \langle 1, \hat{v}_1 \rangle \right) + i \mu_1 \mathcal{L}_{\pm} \hat{v}_1.
\end{array}
\right. 
\label{level-two}
\end{equation}
By Fredholm's theorem, there exists a periodic solution of the linear inhomogeneous system (\ref{level-two}) if and only if the two constraints are satisfied:
\begin{align}
\label{lin-alg-1}
\lambda_1 \langle 1, \mathcal{M} \hat{v}_1 \rangle \mp \mu_1 \langle 1, \hat{w}_1 \rangle &= 0, \\
\lambda_1 \langle \eta', \mathcal{M}^* \hat{w}_1 
-  2 c  \mathcal{H} \hat{v}_1 \rangle + i \mu_1 \langle \eta', \mathcal{L}_{\pm} \hat{v}_1 \rangle &= 0,
\label{lin-alg-2}
\end{align} 
where we have used $\langle 1, \eta' \rangle = \langle 1, \mathcal{H} \hat{w}_1 \rangle = 0$. Taking derivative of the constraint (\ref{zero-mean}) with respect to $c$ yields 
\begin{equation}
\label{der-c}
\langle (1+2K\eta), \partial_c \eta \rangle = 0.
\end{equation}
Since $\mathcal{M}^* 1 = 1 + 2 \mathcal{K}\eta$, it follows from (\ref{level-one-solution}) and (\ref{lin-alg-1}) with the help of (\ref{der-c}) that 
\begin{equation}
    \label{tech-eq-1}
i \mu_1 \lambda_1 \langle (1 + 2 \mathcal{K} \eta), \Upsilon \rangle 
+ i c \mu_1 \lambda_1 \langle 1, \partial_c \eta \rangle \mp a \mu_1 = 0.
\end{equation}
By using (\ref{zero-mean}), (\ref{L-on-1}), (\ref{Upsilon}), and self-adjointness of $\mathcal{L}$, we get 
\begin{equation}
\label{tech-eq-11}
\langle (1 + 2 \mathcal{K} \eta), \Upsilon \rangle = - \langle \mathcal{L} 1, \Upsilon \rangle = -\langle 1, \mathcal{L} \Upsilon \rangle 
= \langle 1, \eta \rangle = - \langle \mathcal{K} \eta, \eta \rangle. 
\end{equation}
Dividing (\ref{tech-eq-1}) by $\pm \mu_1 \neq 0$ and using (\ref{wave-momentum}),  (\ref{der-c}), and (\ref{tech-eq-11}), 
we obtain the unique expression for $a$:
\begin{align}
a &= \pm i \lambda_1 \left[ c \langle 1, \partial_c \eta \rangle - \langle \mathcal{K} \eta, \eta \rangle \right] \notag \\
&= \mp i \lambda_1 \left[ 2 c \langle K \eta, \partial_c \eta \rangle + \langle \mathcal{K} \eta, \eta \rangle \right] \notag \\
&= \mp i \lambda_1 \mathcal{P}'(c). \label{a-first}
\end{align}
Similarly, it follows from (\ref{level-one-solution}) and (\ref{lin-alg-2}) that $(\mu_1,\lambda_1)$ must be related according to the characteristic equation:
\begin{align}
\label{some-tech}
\lambda_1^2 \left[ \langle \mathcal{K} \eta, \eta \rangle + 2 c \langle \mathcal{K} \eta, \partial_c \eta \rangle \right] - i \mu_1 \lambda_1 \left[ \langle \eta, \partial_c \eta \rangle + 2 c \langle \mathcal{K}\eta, \Upsilon \rangle \right] - \mu_1^2 \langle \eta, \Upsilon \rangle = 0,
\end{align}
where we have used $\mathcal{M} \eta' = \eta'$ and $\mathcal{L}_{\pm} \eta' = -\eta$ together with skew-adjointness of 
$\mathcal{H}$ and $\mathcal{L}_{\pm}$ in $L^2_{\rm per}$. Note that $a$ in (\ref{a-first}) does not contribute to the equation (\ref{some-tech}). Substituting (\ref{Upsilon}) into (\ref{some-tech}) and using (\ref{wave-momentum}) and (\ref{wave-norm}) yields the quadratic characteristic equation 
\begin{equation}
\label{quad-eq}
\lambda_1^2 \mathcal{P}'(c) - i \mu_1 \lambda_1 \mathcal{N}'(c) + \mu_1^2 \left[ \mathcal{N}(c) - \frac{c}{4} \mathcal{N}'(c) \right] = 0,
\end{equation}
where we have used that 
$$
2c \langle \mathcal{K} \eta, \Upsilon \rangle = - \langle \mathcal{L} \partial_c \eta, \Upsilon \rangle = \langle \partial_c \eta, \eta \rangle,
$$
due to (\ref{der-speed}) and (\ref{Upsilon}) as well as self-adjointness of $\mathcal{L}$ in $L^2_{\rm per}$. Note that the quadratic equation (\ref{quad-eq}) is the same for both $\mu_1 > 0$ and $\mu_1 < 0$. It admits two solutions $\lambda_1^{\pm}(\mu_1)$ given by 
\begin{equation}
    \label{lambda-quad}
\lambda_1^{\pm}(\mu_1) = \frac{i \mu_1}{2 \mathcal{P}'(c)} \left[ \mathcal{N}'(c) \pm \sqrt{\Delta(c)}\right], 
\end{equation}
where $\Delta(c)$ is given by (\ref{Delta}). If $\Delta(c) > 0$, then 
$\lambda_1^{\pm}(\mu_1) \in i \mathbb{R}$ and $\lambda_1^+(\mu_1) \neq \lambda_2^-(\mu_1)$ for $\mu_1 \neq 0$. If $\Delta(c) = 0$, then $\lambda_1^{\pm}(\mu_1) \in i \mathbb{R}$ and $\lambda_1^+(\mu_1) = \lambda_2^-(\mu_1)$ for $\mu_1 \neq 0$. If $\Delta(c) < 0$, then $\lambda_1^{\pm}(\mu_1) \notin i \mathbb{R}$. 

Persistence of the two spectral bands crossing $\lambda = 0$ for $\Delta(c) \neq 0$ follows from the Hamiltonian symmetry of Lemma \ref{lem-mod-spectrum}. For fixed $\mu_1 \neq 0$, the two eigenvalues $\lambda_1^{\pm}$ are distinct and hence they remain either on $i \R$ if $\Delta(c) > 0$ or in each quadrant of the complex $\lambda$-plane if $\Delta(c) < 0$, in the asymptotic expansions (\ref{expansion-1-eps}) and (\ref{expansion-1}). 

\vspace{0.2cm}

\underline{Subspace $(a_1,a_2) = (0,1)$:} We introduce the asymptotic expansions
\begin{equation}
\label{expansion-1-second-eps}
\mu = \epsilon^2 \mu_2, \quad \lambda = \epsilon \lambda_1 + \mathcal{O}(\epsilon^2),
\end{equation}
and  
\begin{equation}
\label{expansion-1-second}
\left( \begin{array}{c} \hat{v} \\ \hat{w} \end{array} \right) = \left( \begin{array}{c} 0 \\ 1 \end{array} \right) + \epsilon \left( \begin{array}{c} \hat{v}_1 \\ \hat{w}_1 \end{array} \right) + \epsilon^2 \left( \begin{array}{c} \hat{v}_2 \\ \hat{w}_2 \end{array} \right) + \mathcal{O}(\epsilon^3).
\end{equation}
where $\epsilon > 0$ is a small parameter and $\mu_2 \neq 0$ is fixed.

\underline{\bf Order $\mathcal{O}(\epsilon)$:} We use (\ref{kernel}) and (\ref{kernel-generalized}) and obtain the solution in the form 
\begin{align}
\left( \begin{array}{c} \hat{v}_1 \\ \hat{w}_1 \end{array} \right) = \lambda_1 
\left[ \left( \begin{array}{c} -1 \\ 0 \end{array} \right) \pm i 
\left( \begin{array}{c} \eta' \\ -c \end{array} \right) 
\right]
 + a \left( \begin{array}{c} \eta' \\ 0 \end{array} \right),\label{level-one-solution-second}
\end{align}
where $a \in \mathbb{R}$ is arbitrary at this order and the projection 
to $(0,1)^T$ is not included since the eigenfunctions are defined up to arbitrary scalar multiplication in (\ref{expansion-1-second}). 

\underline{\bf Order $\mathcal{O}(\epsilon^2)$:} We get the linear inhomogeneous system 
\begin{equation}
\left\{ \begin{array}{ll}
\mathcal{K} \hat{w}_2 &=  i \mu_2 \mathcal{H}_{\pm} 1 + \lambda_1 \mathcal{M}_{\pm} \hat{v}_1, \\
\mathcal{L} \hat{v}_2 &= \lambda_1 \left( \mathcal{M}^*_{\pm} \hat{w}_1 
-  2 c  \mathcal{H}_{\pm} \hat{v}_1 \right).
\end{array}
\right. 
\label{level-two-second}
\end{equation}
Fredholm's orthogonality condition to the first equation of system (\ref{level-two-second}) gives 
\begin{equation}
\label{single-eq}
\mp \mu_2 - \lambda_1^2 \langle 1, \mathcal{M} 1 \rangle = 0 \quad \Rightarrow \quad \lambda_1^2 = -|\mu_2|, 
\end{equation}
where we have used $\mathcal{M} 1 = 1 + \mathcal{K} \eta$ and $\langle 1, \mathcal{K} \eta \rangle = 0$. 
Fredholm's orthogonality condition to the second equation of system (\ref{level-two-second}) is trivially satisfied 
because 
$$
\langle \eta', \mathcal{M}^*_{\pm} \hat{w}_1 -  2 c  \mathcal{H}_{\pm} \hat{v}_1 \rangle = 
\langle \eta', \hat{w}_1 \rangle - 2 c \langle \mathcal{K} \eta, \hat{v}_1 \rangle = 0.
$$
The two solutions $\lambda_1^{\pm}(\mu_2) = \pm i \sqrt{|\mu_2|} \in i \mathbb{R}$ satisfy $\lambda_1^+(\mu_2) \neq \lambda_1^-(\mu_2)$ for $\mu_2 \neq 0$. Under the orthogonality condition $\lambda_1^2 + |\mu_2| = 0$, the system (\ref{level-two-second}) can be solved explicitly:
\begin{align}
\left( \begin{array}{c} \hat{v}_2 \\ \hat{w}_2 \end{array} \right) = \lambda_1^2 
\left( \begin{array}{c} \pm i (\partial_c \eta - c) \\ -\eta \end{array} \right) 
 + a \lambda_1 \left[ \left( \begin{array}{c} -\partial_c \eta \\ \mathcal{H} \eta \end{array} \right) \pm i \langle 1, \partial_c \eta \rangle \left( \begin{array}{c} \eta' \\ -c \end{array} \right) \right],
 \label{level-two-solution-second}
\end{align}
where we have used $\mathcal{L}(\partial_c \eta - c) = c$, which follows from (\ref{L-on-1}) and (\ref{der-speed}). The parameter $a \in \R$ is not defined in (\ref{level-one-solution-second}) and (\ref{level-two-solution-second}) 
up to the order $\mathcal{O}(\epsilon^2)$. Since $\lambda_1^+(\mu_2) \neq \lambda_1^-(\mu_2)$ for $\mu_2 \neq 0$, the two spectral bands persist on $i \R$ 
in the asymptotic expansions (\ref{expansion-1-second-eps}) and (\ref{expansion-1-second}) due to the Hamiltonian symmetry of Lemma \ref{lem-mod-spectrum}. 
\end{proof}

\begin{remark}
	\label{rem-Delta}
	If $\Delta(c) < 0$, the spectral bands in Theorem \ref{theorem-bif-1} crossing $\lambda = 0$ can be parts of either figure-$8$ or  figure-$\infty$ or a more complicated pattern. If $\Delta(c) > 0$, the spectral bands on $i \mathbb{R}$ near $\lambda = 0$ may become complex away from $\lambda = 0$ due to additional bifurcations of the instability bubbles on $i \R$. If $\Delta(c) = 0$, we have the bifurcation for which the degenerate figure-8 bands have the vertical slope at $\lambda = 0$. Computation of the normal form of this bifurcation is required to unfold the degeneracy beyond the leading order and to obtain further information on whether the spectral bands are subsets of $i \mathbb{R}$ near $\lambda = 0$ or cross $\lambda = 0$ from the four complex quadrants of the complex $\lambda$-plane, see Theorem \ref{theorem-bif-1-normal-form} below.
\end{remark}

\subsection{Normal form for the bifurcation} 
\label{sec-3-2}

The normal form takes into account higher-order terms in $\mu$ which break the degeneracy of the double eigenvalue $\lambda_1^{\pm}(\mu_1) = \frac{i \mu_1 \mathcal{N}'(c)}{2 \mathcal{P}'(c)}$ in the case of $\Delta(c) = 0$, see Remark \ref{rem-Delta}. This is done by continuing Puiseux expansions to the higher orders in $\epsilon$, as in Section \ref{sec-case-4}. The following theorem gives the normal form with a numerical coefficient $\mathcal{D}$, which is assumed to be non-zero for the non-degeneracy.

\begin{theorem}
	\label{theorem-bif-1-normal-form}
    Under the assumptions of Theorem \ref{theorem-bif-1}, let $c_0 \in (1,c_*)$ be the root of $\Delta(c)$ given by (\ref{Delta}). 
    Let $\mathcal{D}$ be given by (\ref{num-coeff-D}) and assume that $\mathcal{D} \neq 0$. The two spectral bands intersecting $\lambda = 0$ with the vertical slope at $c = c_0$ are non-degenerate and obtained from roots of the characteristic equation
    \begin{align}
		\label{normal-form-one}
\mathcal{P}'(c_0) \left( \lambda - \frac{i \mu \mathcal{N}'(c_0)}{2 \mathcal{P}'(c_0)} \right)^2 + \mathcal{D} |\mu|^3 = \mathcal{O}(\mu^4).	
	\end{align}	
If ${\rm sgn}\mathcal{D} = {\rm sgn} \mathcal{P}'(c_0)$, the two spectral bands are subsets of $i \R$. If 
${\rm sgn}\mathcal{D} = -{\rm sgn} \mathcal{P}'(c_0)$, the two spectral bands cross $\lambda = 0$ along the straight lines 
in all four complex quadrants of the complex $\lambda$-plane.
\end{theorem}

\begin{proof}
By using Puiseux expansions for the double eigenvalue 
$\lambda_1^+(\mu_1) = \lambda_2^-(\mu_1)$ in (\ref{lambda-quad}) with $\Delta(c_0) = 0$, we modify the expansions (\ref{expansion-1-eps}) and (\ref{expansion-1})
as follows:
\begin{equation}
    \label{formal-exp-1}
\mu = \epsilon \mu_1, \quad \lambda = \epsilon \lambda_1 + \epsilon^{3/2} \lambda_{3/2} + \mathcal{O}(\epsilon^2),
\end{equation}
and 
\begin{align}
\left( \begin{array}{c} \hat{v} \\ \hat{w} \end{array} \right) &= \left( \begin{array}{c} \eta' \\ 0 \end{array} \right) + \epsilon \left( \begin{array}{c} \hat{v}_1 \\ \hat{w}_1 \end{array} \right) + \epsilon^{3/2} \left( \begin{array}{c} \hat{v}_{3/2} \\ \hat{w}_{3/2} \end{array} \right) + \epsilon^2 \left( \begin{array}{c} \hat{v}_2 \\ \hat{w}_2 \end{array} \right) \notag \\
& \quad + \epsilon^{5/2} \left( \begin{array}{c} \hat{v}_{5/2} \\ \hat{w}_{5/2} \end{array} \right) + \epsilon^2 \left( \begin{array}{c} \hat{v}_3 \\ \hat{w}_3 \end{array} \right) + \mathcal{O}(\epsilon^{7/2}).
\label{expansion-1-P}
\end{align}
We get the same solutions for the first two corrections terms as in (\ref{level-one-solution}): 
\begin{align}
\left( \begin{array}{c} \hat{v}_1 \\ \hat{w}_1 \end{array} \right) = \lambda_1 
\left( \begin{array}{c} \tilde{v}_1 \\ \tilde{w}_1 \end{array} \right) + i \mu_1 \left( \begin{array}{c} \Upsilon \\ 0 \end{array} \right) 
 + a_1 \left( \begin{array}{c} 0 \\ 1 \end{array} \right),\label{level-one-solution-P}
\end{align}
and 
\begin{align}
\left( \begin{array}{c} \hat{v}_{3/2} \\ \hat{w}_{3/2} \end{array} \right) = \lambda_{3/2}
\left( \begin{array}{c} \tilde{v}_1 \\ \tilde{w}_1 \end{array} \right) + a_{3/2} \left( \begin{array}{c} 0 \\ 1 \end{array} \right),
\label{level-three-half-solution-P}
\end{align}
where 
\begin{align*}
 \left\{ \begin{array}{l} \tilde{v}_1 = -\partial_c \eta \pm i \langle 1, \partial_c \eta \rangle \eta', \\
    \tilde{w}_1 = \mathcal{H} \eta  \mp i c \langle 1, \partial_c \eta \rangle.
    \end{array} \right.
\end{align*}
We recall from  (\ref{a-first}) and (\ref{lambda-quad}) with $\Delta(c_0) = 0$ that $a_1 = \mp i \lambda_1 \mathcal{P}'(c_0)$ and 
$\lambda_1 = i \gamma_1 \mu_1$ with $\gamma_1 = \frac{\mathcal{N}'(c_0)}{2 \mathcal{P}'(c_0)}$. The parameter $a_{3/2} \in \mathbb{R}$ is arbitrary at this order. The linear inhomogeneous system (\ref{level-two}) is solvable under the two conditions on $a$ and $\lambda_1$ so that we can define the solution in the implicit form:
\begin{align}
\left( \begin{array}{c} \hat{v}_2 \\ \hat{w}_2 \end{array} \right) = \mu_1^2
\left( \begin{array}{c} \tilde{v}_2 \\ \tilde{w}_2 \end{array} \right) 
+ a_2 \left( \begin{array}{c} 0 \\ 1 \end{array} \right), 
\label{level-two-solution-P}
\end{align}
where the coefficient $a_2\in \mathbb{R}$ is arbitrary at this order and $(\tilde{v}_2,\tilde{w}_2) \in H^1_{\rm per} \times H^1_{\rm per}$ is uniquely determined from solutions of he linear inhomogeneous system 
\begin{equation}
\left\{ \begin{array}{ll}
\mathcal{K} \tilde{w}_2 &=  -\gamma_1^2 \mathcal{M}_{\pm} \tilde{v}_1 - \gamma_1 \mathcal{M}_{\pm} \Upsilon - \gamma_1 \mathcal{H}_{\pm} 
\left( \tilde{w}_1 \mp i \mathcal{P}'(c_0) \right), \\
\mathcal{L} \tilde{v}_2 &=  -\gamma_1^2 \left( \mathcal{M}_{\pm}^* \tilde{w}_1 \mp i \mathcal{P}'(c_0) \mathcal{M}^* 1 
-  2 c  \mathcal{H}_{\pm} \tilde{v}_1 \right) + 2 c \gamma_1 \mathcal{H}_{\pm} \Upsilon - \gamma_1 \mathcal{L}_{\pm} \tilde{v}_1 - \mathcal{L}_{\pm} \Upsilon,
\end{array}
\right. 
\label{level-two-P}
\end{equation}
under the orthogonality conditions $\langle 1, \hat{w}_2 \rangle = \langle \eta', \hat{v}_2 \rangle = 0$.

\underline{\bf Order $\mathcal{O}(\epsilon^{5/2})$:} We get the linear inhomogeneous system 
\begin{equation}
\left\{ \begin{array}{ll}
\mathcal{K} \hat{w}_{5/2} &=  2 \lambda_1 \lambda_{3/2} \mathcal{M}_{\pm} \tilde{v}_1 + 
i \mu_1 \left( \lambda_{3/2} \mathcal{H}_{\pm} \tilde{w}_1 + a_{3/2} \mathcal{H}_{\pm} 1 \right), \\
\mathcal{L} \hat{v}_{5/2} &= 2 \lambda_1 \lambda_{3/2} \left( \mathcal{M}_{\pm}^* \tilde{w}_{1} 
-  2 c  \mathcal{H}_{\pm} \tilde{v}_{1} \right) + (\lambda_1 a_{3/2} + \lambda_{3/2} a_1) \mathcal{M}_{\pm}^* 1  \\
& \qquad 
- 2 c i \mu_1 \lambda_{3/2}  \mathcal{H}_{\pm} \Upsilon + i \mu_1 \lambda_{3/2} \mathcal{L}_{\pm} \tilde{v}_{1}.
\end{array}
\right. 
\label{level-five-half}
\end{equation}
The Fredholm's condition for the first equation of the system (\ref{level-five-half}) yields 
$$
a_{3/2} = \pm i \lambda_{3/2} c \langle 1, \partial_c \eta \rangle,
$$
which simply implies that the correction term in (\ref{level-three-half-solution-P}) can be rewritten in the form 
\begin{align}
\left( \begin{array}{c} \hat{v}_{3/2} \\ \hat{w}_{3/2} \end{array} \right) = \lambda_{3/2}
\left( \begin{array}{c} \tilde{v}_1 \\ \mathcal{H} \eta \end{array} \right).
\label{level-three-half-solution-P-equiv}
\end{align}
The Fredholm's condition for the second equation of the system (\ref{level-five-half}) yields 
\begin{align*}
   0 &= \langle \eta', 2 \lambda_1 \lambda_{3/2} \left( \mathcal{M}_{\pm}^* \tilde{w}_{1} 
-  2 c  \mathcal{H}_{\pm} \tilde{v}_{1} \right) 
- 2 c i \mu_1 \lambda_{3/2}  \mathcal{H}_{\pm} \Upsilon + i \mu_1 \lambda_{3/2} \mathcal{L}_{\pm} \tilde{v}_{1} \rangle \\
&=  2 \lambda_1 \lambda_{3/2} \left[ \langle \eta', \mathcal{H} \eta \rangle + 2 c \langle \eta', \mathcal{H} \partial_c \eta \rangle \right] 
- i \mu_1 \lambda_{3/2} \left[ \langle \eta, \partial_c \eta \rangle + 2 c \langle \mathcal{K} \eta, \Upsilon \rangle \right]\\
&= \lambda_{3/2} \left[ 2 \lambda_1 \mathcal{P}'(c_0) - i \mu_1 \mathcal{N}'(c_0) \right].
\end{align*}
This condition is equivalent to the derivative of the quadratic characteristic equation (\ref{quad-eq}) in $\lambda_1$ 
and it is trivially satisfied since $\lambda_1 = i \gamma_1 \mu_1$ is a double 
root of the characteristic equation (\ref{quad-eq}) with $\Delta(c_0) = 0$. Hence, we write the solution in the implicit form:
\begin{align}
\left( \begin{array}{c} \hat{v}_{5/2} \\ \hat{w}_{5/2} \end{array} \right) = \mu_1 \lambda_{3/2}
\left( \begin{array}{c} \tilde{v}_{5/2} \\ \tilde{w}_{5/2} \end{array} \right) + a_{5/2} \left( \begin{array}{c} 0 \\ 1 \end{array} \right), 
\label{level-five-two-solution-P}
\end{align}
where the coefficient $a_{5/2} \in \mathbb{R}$ is arbitrary at this order and $(\tilde{v}_{5/2},\tilde{w}_{5/2}) \in H^1_{\rm per} \times H^1_{\rm per}$ is uniquely determined from solution of 
the linear inhomogeneous system
\begin{equation*}
\left\{ \begin{array}{ll}
\mathcal{K} \tilde{w}_{5/2} &=  2 i \gamma_1 \mathcal{M}_{\pm} \tilde{v}_1 + 
+ i \mathcal{H}_{\pm} \mathcal{H} \eta, \\
\mathcal{L} \tilde{v}_{5/2} &= 2 i \gamma_1 \left( \mathcal{M}_{\pm}^* \tilde{w}_{1} 
-  2 c  \mathcal{H}_{\pm} \tilde{v}_{1} \right) \pm \gamma_1  (\mathcal{P}'(c_0) - c \langle 1, \partial_c \eta \rangle) \mathcal{M}_{\pm}^* 1  \\
& \qquad 
- 2 c i \mathcal{H}_{\pm} \Upsilon + i \mathcal{L}_{\pm} \tilde{v}_{1},
\end{array}
\right. 
\end{equation*}
under the orthogonality conditions $\langle 1, \hat{w}_{5/2} \rangle = \langle \eta', \hat{v}_{5/2} \rangle = 0$.

\underline{\bf Order $\mathcal{O}(\epsilon^3)$:} We get the linear inhomogeneous system 
\begin{equation}
\left\{ \begin{array}{ll}
\mathcal{K} \hat{w}_3 &=  \lambda_1 \mathcal{M}_{\pm} \hat{v}_2 + \lambda_{3/2} \mathcal{M}_{\pm} \hat{v}_{3/2} + 
i \mu_1 \mathcal{H}_{\pm} \hat{w}_{2}, \\
\mathcal{L} \hat{v}_3 &= \lambda_1 \left( \mathcal{M}_{\pm}^* \hat{w}_{2} 
-  2 c  \mathcal{H}_{\pm} \hat{v}_{2} \right) + \lambda_{3/2} \left( \mathcal{M}_{\pm}^* \hat{w}_{3/2} 
-  2 c  \mathcal{H}_{\pm} \hat{v}_{3/2} \right) + i \mu_1 \mathcal{L}_{\pm} \hat{v}_{2}.
\end{array}
\right. 
\label{level-three-P}
\end{equation}
The Fredholm's condition for the first equation of the system (\ref{level-three-P}) yields a unique value for the coefficient $a_2$. Since 
it does not affect the computations of the second equation, we do not 
write it explicitly. The Fredholm's condition for the second equation of the system (\ref{level-three-P}) yields the characteristic equation for $\lambda_{3/2}$:
\begin{align}
    0 &= \langle \eta', \lambda_1 \left( \mathcal{M}_{\pm}^* \hat{w}_{2} 
-  2 c  \mathcal{H}_{\pm} \hat{v}_{2} \right) + \lambda_{3/2} \left( \mathcal{M}_{\pm}^* \hat{w}_{3/2} 
-  2 c  \mathcal{H}_{\pm} \hat{v}_{3/2} \right) + i \mu_1 \mathcal{L}_{\pm} \hat{v}_{2} \rangle \notag \\
&= \lambda_1 \mu_1^2 \left[ \langle \eta', \tilde{w}_2 \rangle - 2 c \langle \eta', \mathcal{H} \tilde{v}_2 \rangle \right]
+ \lambda_{3/2}^2 \left[ \langle \eta', \mathcal{H} \eta \rangle + 2 c \langle \eta', \mathcal{H} \partial_c \eta \rangle \right]
+ i \mu_1^3 \langle \eta, \tilde{v}_2 \rangle \notag  \\
&= \lambda_{3/2}^2 \mathcal{P}'(c_0) + i \mu_1^3 \left[ \gamma_1 \langle \eta', \tilde{w}_2 \rangle - 2 c \gamma_1 \langle \eta', \mathcal{H} \tilde{v}_2 
\rangle + \langle \eta, \tilde{v}_2 \rangle \right].\label{tech-eq-2}
\end{align}
In order to get (\ref{normal-form-one}) with an explicit expression for the numerical coefficient $\mathcal{D}$, we need to extract the odd part of $\tilde{w}_2$ and the even part of $\tilde{v}_2$ from solutions of the linear inhomogeneous equation (\ref{level-two-P}). We compute
\begin{align*}
    \mathcal{M}_{\pm} \tilde{v}_1 = -\mathcal{M} \partial_c \eta, \quad
    \mathcal{M}_{\pm} \Upsilon &= \mathcal{M} \Upsilon \pm i \langle 1, \Upsilon \rangle \eta', 
\end{align*}
and
\begin{align*}
    \mathcal{H}_{\pm} \left( \tilde{w}_1 \mp i \mathcal{P}'(c_0) \right) &= \mathcal{H} \mathcal{H} \eta + \mathcal{P}'(c_0) + c \langle 1, \partial_c \eta \rangle.
\end{align*}
Since $\mathcal{K}$, $\mathcal{M}$, and $\mathcal{H} \mathcal{H}$ are parity preserving, the odd component of $\tilde{w}_2$ denoted as $\tilde{w}^{\rm odd}$ is given by a solution of the linear inhomogeneous equation 
$$
\mathcal{K} \tilde{w}_2^{\rm odd} =  \mp i \gamma_1 \langle 1, \Upsilon \rangle \eta', 
$$
which yields $\tilde{w}_2^{\rm odd} = \mp i \gamma_1 \langle 1, \Upsilon \rangle \mathcal{H} \eta$. Similarly, we compute 
\begin{align*}
    \mathcal{M}^*_{\pm} \tilde{w}_1 &= \mathcal{M}^* \mathcal{H} \eta \mp i c \langle 1, \partial_c \eta \rangle (1 + 2 \mathcal{K} \eta) 
    \pm i \langle \mathcal{K} \eta, \eta \rangle, \\
       \mathcal{H}^*_{\pm} \tilde{v}_1 &= - \mathcal{H} \partial_c \eta \mp i \langle 1, \partial_c \eta \rangle (1 + \mathcal{K} \eta), \\
       \mathcal{L}_{\pm} \tilde{v}_1 &= -c^2 \mathcal{H} \partial_c \eta \mp i c^2 \langle 1, \partial_c \eta \rangle (1 + \mathcal{K} \eta) 
       + \eta \mathcal{H} \partial_c \eta \pm i \langle 1, \partial_c \eta \rangle \eta (1 + \mathcal{K} \eta) \\
       & \qquad 
       + \mathcal{H} (\eta \partial_c \eta) \mp i \langle 1, \partial_c \eta \rangle \mathcal{H} (\eta \eta') \pm i \langle \eta, \partial_c \eta \rangle, 
\end{align*}
and
\begin{align*}
       \mathcal{L}_{\pm} \Upsilon &= c^2 \mathcal{H} \Upsilon \pm i c^2 \langle 1, \Upsilon \rangle - \eta \mathcal{H} \Upsilon \mp i \langle 1, \Upsilon \rangle \eta 
       - \mathcal{H} (\eta \Upsilon) \mp i \langle \eta, \Upsilon \rangle.
\end{align*}
We note again that $\mathcal{K}$, $\mathcal{L}$, and $\mathcal{M}^*$ are parity preserving, and that the even component in the right-hand sides alternate signs. Hence, the even component of $\tilde{v}_2$ denoted as $\tilde{v}^{\rm even}$ is given by $\tilde{v}_2^{\rm even} = \pm i \check{v}_2^{\rm even}$, where $\check{v}_2^{\rm even}$ is real, even, and uniquely defined from 
a solution of the linear inhomogeneous equation 
\begin{align*}
\mathcal{L} \check{v}_2^{\rm even} &=  \gamma_1^2 c \langle 1, \partial_c \eta \rangle (1 + 2 \mathcal{K} \eta) 
- \gamma_1^2 \langle \mathcal{K} \eta, \eta \rangle + \gamma_1^2  \mathcal{P}'(c_0)  (1 + 2 \mathcal{K} \eta) 
-  2 c \gamma_1^2 \langle 1, \partial_c \eta \rangle (1 + \mathcal{K} \eta) \\
& \quad + 2 c \gamma_1 \langle 1, \Upsilon \rangle + \gamma_1 c^2 \langle 1, \partial_c \eta \rangle (1 + \mathcal{K} \eta) 
- \gamma_1  \langle 1, \partial_c \eta \rangle \eta (1 + \mathcal{K} \eta) 
+ \gamma_1 \langle 1, \partial_c \eta \rangle \mathcal{H} (\eta \eta')  \\
& \quad - \gamma_1 \langle \eta, \partial_c \eta \rangle - c^2 \langle 1, \Upsilon \rangle - \langle 1, \Upsilon \rangle \eta + \langle \eta, \Upsilon \rangle.
\end{align*}
By using (\ref{tech-eq-2}), we obtain 
$$
\lambda_{3/2}^2 \mathcal{P}'(c_0) \pm \mu_1^3 \mathcal{D} = 0, 
$$
with 
\begin{equation}
    \label{num-coeff-D}
\mathcal{D} = \gamma_1^2 \langle \mathcal{K} \eta, \eta\rangle \langle 1, \Upsilon \rangle + 2 c \gamma_1 \langle \mathcal{K} \eta, \check{v}_2^{\rm even}
\rangle - \langle \eta, \tilde{v}_2^{\rm even}\rangle.
\end{equation}
In view of the expansion (\ref{formal-exp-1}), we obtain the normal form (\ref{normal-form-one}). The asymptotic expansions (\ref{expansion-1-P}) is 
given by (\ref{level-one-solution-P}), (\ref{level-two-solution-P}), (\ref{level-three-half-solution-P-equiv}), and (\ref{level-five-two-solution-P}). 
If ${\rm sgn}\mathcal{D} = {\rm sgn} \mathcal{P}'(c_0)$, the two spectral bands are non-degenerate and persist on subsets of $i \R$ due to the Hamiltonian symmetry of Lemma \ref{lem-mod-spectrum}. If 
${\rm sgn}\mathcal{D} = -{\rm sgn} \mathcal{P}'(c_0)$, the two spectral bands are non-degenerate and persist in all four complex quadrants of the complex $\lambda$-plane.
\end{proof}

\begin{remark}
	Combining (\ref{lambda-quad}) and (\ref{normal-form-one}) yields 
    \begin{align}
    \label{tech-eq}
\mathcal{P}'(c) \left( \lambda - \frac{i \mu \mathcal{N}'(c)}{2 \mathcal{P}'(c)} \right)^2 + \frac{\mu^2 \Delta(c)}{4 \mathcal{P}'(c)} 
+ \mathcal{D} |\mu|^3 = \mathcal{O}(\mu^4),
\end{align}		
	where $\mathcal{D}$ is only defined at $c = c_0$ and $c_0$ is a root of 
	$\Delta(c)$ in Theorem \ref{theorem-bif-1-normal-form}. 
	Assume that $c_0$ is a simple root of $\Delta(c)$ with $\Delta'(c_0) \neq 0$, see Figure \ref{fig:delta}. Expanding (\ref{tech-eq}) in $c$ near $c_0$ by using the smoothness of the mapping $c \mapsto \eta \in C^{\infty}_{\rm per}$, we obtain the extended normal form: 
    \begin{align}
		\label{normal-form-one-extended}
\mathcal{P}'(c_0) \left( \lambda - \frac{i \mu \mathcal{N}'(c)}{2 \mathcal{P}'(c)} \right)^2 + \frac{\mu^2 \Delta'(c_0) (c - c_0)}{4 \mathcal{P}'(c_0)} + \mathcal{D} |\mu|^3 = \mathcal{O}(\mu^4),	
	\end{align}	
where the leading-order balance implies that $c - c_0 = \mathcal{O}(\mu)$.
\end{remark}

\subsection{Numerical results}
\label{sec-3-3}
    
The criterion for the bifurcation is a simple zero of $\Delta(c)$ defined by~\eqref{Delta}. The graph of $\Delta(c)$ versus wave steepness $s$ is shown in Figure \ref{fig:delta} by green solid line. Newton's method is further used to
improve the accuracy of finding $c_0$ at which $\Delta(c)$ vanishes. The corresponding values of $s_0$ and $c_0$ for the first two zeros of $\Delta$ are recorded in Table \ref{tab:hourglass}, where we also show the values of $\mathcal{P}'(c_0)$, $\mathcal{N}'(c_0)$, $\mathcal{D}$, and $\Delta'(c_0)$  in the normal form (\ref{normal-form-one}).
We note that ${\rm sgn} \mathcal{D} = -{\rm sgn}\mathcal{P}'(c_0)$ so that 
the two spectral bands with the vertical slope cross $\lambda = 0$ along the straight lines in all four quadrants of the complex $\lambda$-plane, see 
Theorem \ref{theorem-bif-1-normal-form}. The sign alternation of $\Delta'(c_0)$ is due to the fact that the mapping $s \mapsto c(s)$ is increasing at the first hourglass bifurcation and decreasing at the second hourglass bifurcation.


\begin{table}[htb!]
	\centering
	\renewcommand{\arraystretch}{1.3}
	\begin{tabular}{c|c|c|c|c|c}
		$s_0$ & $c_0$ & $\mathcal{P}'(c_0)$ & $\mathcal{N}'(c_0)$ & $\mathcal D$ & $\Delta'(c_0)$ \\ \hline
		$0.10921308$  & $1.06052068$  &  $4.618$ & $3.2142$ & $-1.804\times10^{-2}$ & $42.51$ \\      
		$0.14010224$  & $1.09256435$  &  $4.653$ & $2.8766$ & $-2.562\times10^{-2}$ & $-8.818\times10^{3}$ \\ 
	\end{tabular}
	\caption{Parameters of the first two hourglass bifurcations.}
	\label{tab:hourglass}
\end{table}

Figures \ref{fig:bcd1} and \ref{fig:bcd2} show 
the instability bands of the spectral problem (\ref{mod-Bab-eq}) near the origin in the spectral $\lambda$-plane obtained numerically.  
The solid blue, green and red curves in each figure correspond to three consequent values of steepness before, at, and after the critical steepness, for which $\Delta(c_0) = 0$.  The dashed lines display the spectral bands obtained from the normal form~\eqref{normal-form-one-extended}, where the value of $\mathcal{D}$ was computed by fitting the real part of $\lambda$ to the numerical data at $\Delta(c_0) = 0$ from equation (\ref{normal-form-one}), instead of using the projection formulas in (\ref{num-coeff-D}). A good agreement between the spectral bands given by the normal form (\ref{normal-form-one-extended}) and the numerical approximations of the stability spectra is evident for the first two cycles of bifurcations. 

\begin{figure}[htb!]
	\centering
	\includegraphics[width=0.6\linewidth]{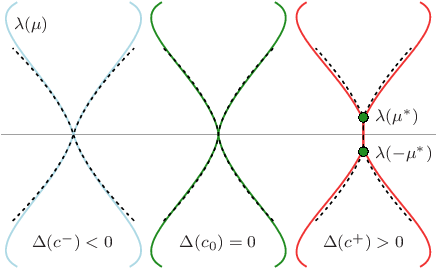}
	\caption{The first instance of the hourglass bifurcation,
		which corresponds to panels (b)--(d) of Figure \ref{fig-bif} zoomed near the origin in the $\lambda$-plane. The figure-$8$ for $c^- = 1.05970 < c_0$ (left panel) gets vertical slopes at the origin when $\Delta(c_0) = 0$ for $c_0 \approx 1.06052068$  (middle panel), and splits into a pair of isolas along $i \mathbb{R}$ for $c^+ = 1.06135 > c_0$ (right panel). The dashed lines show the approximations obtained from the normal form~\eqref{normal-form-one-extended}.}
	\label{fig:bcd1}
\end{figure}

\begin{figure}[htb!]
	\centering
	\includegraphics[width=0.6\linewidth]{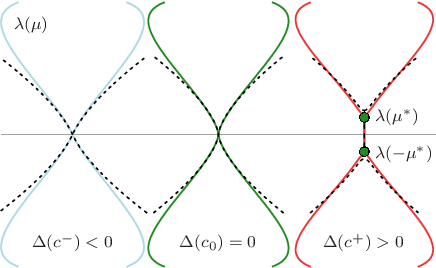}
	\caption{The same as Figure \ref{fig:bcd1} but for the second instance 
		of the hourglass bifurcation, which happens when $\Delta(c_0) = 0$ for $c_0 \approx 1.09256435$.}
	\label{fig:bcd2}
\end{figure}

\section{New circular bands around the origin}
\label{sec-4}

The corresponding  bifurcation is seen in panels (e)--(g) of Figure \ref{fig-bif}.
Due to its shape, it can be called the ellipse bifurcation.
Before the bifurcation [panel (e)], there are no circular bands around the origin. At the bifurcation [panel (f)], a new circular band appears from the origin and surrounds it.  After the bifurcation [panel (g)], the new circular band grows in size.

Section \ref{sec-4-1} develops the criterion for this bifurcation. Section \ref{sec-4-2} gives the derivation of the normal form, 
which explains how the bifurcation is unfolded. Section \ref{sec-4-3} describes 
numerical approximations of the coefficients of the normal form, 
from which we are able to reproduce the bifurcation pattern seen in panels (e)--(g) of Figure \ref{fig-bif}.

\subsection{Criterion of the bifurcation}
\label{sec-4-1}

This bifurcation arises due to the zero eigenvalue of the linearized operator $\mathcal{L}$ extended to the space of double-periodic (antiperiodic) perturbations and denoted as $\mathcal{L}_{\rm aper}$. This linearized operator has no zero eigenvalues generally. However, zero eigenvalues appear in the continuation of the traveling wave with respect to the steepness parameter $s$, as was shown in \cite{BDT2000b}. Hence, we add the following non-generic assumption for the bifurcation. 

\begin{assumption}
	\label{ass-second-bif}
For the given value of $c \in (1,c_*)$, there exists an eigenfunction $\psi_0 \in H^1_{\rm aper}$ (defined up to the scalar multiplication) such that $\mathcal{L}_{\rm aper} \psi_0 = 0$, where 
$$
\mathcal{L}_{\rm aper} : H^1_{\rm aper} \subset L^2_{\rm aper} \to L^2_{\rm aper}
$$ 
is the extension of $\mathcal{L}$ given by (\ref{linBop}) to the $2\pi$-antiperiodic ($4\pi$-periodic) perturbations.  
\end{assumption}

To proceed with analysis, we first reformulate the spectral problem (\ref{mod-Bab-eq}) for $(\mu,\lambda)$ near $\left(\frac{1}{2},0\right)$. 
Since $\mu = \frac{1}{2}$ corresponds to the antiperiodic functions 
represented by the Fourier series 
\begin{equation}
\label{Fourier-series-anti}
f(u) = \sum_{n \in \mathbb{Z}} \hat{f}_n e^{i \left( n + \frac{1}{2} \right) u} : \quad f(u+2 \pi) = -f(u),
\end{equation}
it is in fact easier to substitute $\mu = \frac{1}{2} - \tilde{\mu}$ with 
$\tilde{\mu}$ defined near $0$. It follows from (\ref{Fourier-series-anti}) that 
$\mathcal{H}_{\mu} = \mathcal{H}_{\rm aper}$ if $|\tilde{\mu}| < \frac{1}{2}$, where $\mathcal{H}_{\rm aper}$ is the Hilbert transform acting on anti-periodic functions. Similarly, instead of (\ref{K-mu}), (\ref{M-mu}), 
(\ref{M-star-mu}), and (\ref{L-mu}), we have 
\begin{align*}
\mathcal{K}_{\mu} &= \mathcal{K}_{\rm aper} + i \tilde{\mu} \mathcal{H}_{\rm aper}, \\ 
\mathcal{M}_{\mu} &= \mathcal{M}_{\rm aper}, \quad 
\mathcal{M}_{\mu}^* = \mathcal{M}^*_{\rm aper}, \\
\mathcal{L}_{\mu} &= \mathcal{L}_{\rm aper} + i \tilde{\mu} \mathcal{S}_{\rm aper},
\end{align*}
where the definitions of $\mathcal{K}_{\rm aper}$, $\mathcal{M}_{\rm aper}$, $\mathcal{M}^*_{\rm aper}$, and $\mathcal{L}_{\rm aper}$ acting  on antiperiodic functions follow from (\ref{operator-K}), (\ref{operator-M}), and (\ref{linBop}), and 
\begin{equation*}
\label{operator-S}
\mathcal{S}_{\rm aper} = c^2 \mathcal{H}_{\rm aper} - \eta \mathcal{H}_{\rm aper} - \mathcal{H}_{\rm aper} (\eta \cdot).
\end{equation*}
By using  $\mu = \frac{1}{2} - \tilde{\mu}$ and the linear operators 
introduced above, we can rewrite the stability problem (\ref{mod-Bab-eq}) in the equivalent form
\begin{equation}
\left\{ \begin{array}{ll}
(\mathcal{K}_{\rm aper} + i \tilde{\mu} \mathcal{H}_{\rm aper}) \hat{w} &= \lambda \mathcal{M}_{\rm aper} \hat{v}, \\
(\mathcal{L}_{\rm aper} + i \tilde{\mu} \mathcal{S}_{\rm aper}) \hat{v} &= \lambda (\mathcal{M}^*_{\rm aper} \hat{w} 
- 2 c \mathcal{H}_{\rm aper} \hat{v}),
\end{array}
\right. 
\label{mod-Bab-eq-anti}
\end{equation}
where $(\hat{v},\hat{w}) \in H^1_{\rm aper} \times H^1_{\rm aper}$ and $|\tilde{\mu}| < \frac{1}{2}$.

The following theorem gives the criterion for the corresponding bifurcation. Compared to the case of the co-periodic perturbations in Sections \ref{sec-3}, \ref{sec-5}, and \ref{sec-6}, 
the numerical coefficients of the normal form are not related to the conserved quantities computed at the traveling wave profile $\eta \in C^{\infty}_{\rm per}$. We also do not use Assumptions \ref{ass-kernel} and \ref{ass-first-bif} since they 
determine the dimensions of the kernel and the generalized kernel of the spectral problem (\ref{spec-Bab-eq}) for co-periodic perturbations in $H^1_{\rm per} \times H^1_{\rm per}$. 

\begin{theorem}
	\label{theorem-bif-2}
	Let Assumption \ref{ass-second-bif} be true and define numerical coefficients 
	$\mathcal{A}_1$, $\mathcal{A}_2$, and $\mathcal{A}_3$ by (\ref{coefficients-A}) computed from the wave profile $\eta \in C^{\infty}_{\rm per}$ and the eigenfunction profile $\psi_0 \in H^1_{\rm aper}$. If $\mathcal{A}_1 \neq 0$ and $\mathcal{A}_1 \mathcal{A}_3 + c^2 \mathcal{A}_2^2 > 0$, then there exist two spectral bands of the modulational stability problem (\ref{mod-Bab-eq-anti})  intersecting $\lambda = 0$, which are subsets of $i \R$ and non-degenerate.
\end{theorem} 

\begin{proof}
To expand solutions of the spectral problem (\ref{mod-Bab-eq-anti}) near $(\tilde{\mu},\lambda) = (0,0)$, we introduce the asymptotic expansion 
\begin{equation}
\label{expansion-2-eps}
\tilde{\mu} = \epsilon \tilde{\mu}_1, \quad \lambda = \epsilon \lambda_1 + \mathcal{O}(\epsilon^2)
\end{equation}
and 
\begin{equation}
\label{expansion-2}
\left( \begin{array}{c} \hat{v} \\ \hat{w} \end{array} \right) = \left( \begin{array}{c} \psi_0 \\ 0 \end{array} \right) + \epsilon \left( \begin{array}{c} \hat{v}_1 \\ \hat{w}_1 \end{array} \right) + \epsilon^2 \left( \begin{array}{c} \hat{v}_2 \\ \hat{w}_2 \end{array} \right) + \mathcal{O}(\epsilon^3),
\end{equation}
where $\epsilon$ is a small parameter, $\tilde{\mu}_1$ is fixed, and the correction terms $(\hat{v}_1,\hat{w}_1)$ and $(\hat{v}_2,\hat{w}_2)$ at the $\mathcal{O}(\epsilon)$ and $\mathcal{O}(\epsilon^2)$ orders are to be computed in function space $H^1_{\rm aper} \times H^1_{\rm aper}$. The expansions (\ref{expansion-2-eps}) and (\ref{expansion-2}) are justified due to the orthogonality of $\mathcal{S}_{\rm aper} \psi_0$ to $\psi_0$ in $L^2_{\rm aper}$ since $\mathcal{S}_{\rm aper}$ is skew-adjoint in $L^2_{\rm aper}$. This ensures that the term of the order of $\mathcal{O}(\mu^{1/2})$ in the expansion of $\lambda$ in $\mu$ is identically zero and the next-order term is $\mathcal{O}(\mu)$, see Section  \ref{sec-case-2}.

\underline{\bf Order $\mathcal{O}(\epsilon)$:} We get the linear inhomogeneous system 
\begin{equation}
\left\{ \begin{array}{ll}
\mathcal{K}_{\rm aper} \hat{w}_1 &=  \lambda_1 \mathcal{M}_{\rm aper} \psi_0, \\
\mathcal{L}_{\rm aper} \hat{v}_1 &= -  2 c \lambda_1 \mathcal{H}_{\rm aper} \psi_0  + i \tilde{\mu}_1 \mathcal{S}_{\rm aper} \psi_0.
\end{array}
\right. 
\label{second-level-one}
\end{equation}
Since ${\rm Ker}(\mathcal{L}_{\rm aper}) = {\rm span}(\psi_0)$ in $H^1_{\rm aper}$ by Assumption \ref{ass-second-bif}, whereas 
$\mathcal{H}_{\rm aper}$ and $\mathcal{S}_{\rm aper}$ are skew-adjoint in $L^2_{\rm aper}$, the Fredholm condition is satisfied for the linear inhomogeneous system (\ref{second-level-one}) and we obtain the solution in the form 
\begin{align}
\left( \begin{array}{c} \hat{v}_1 \\ \hat{w}_1 \end{array} \right) =  \lambda_1 
\left( \begin{array}{c} -2c \mathcal{L}_{\rm aper}^{-1} \mathcal{H}_{\rm aper} \psi_0 \\ \mathcal{K}_{\rm aper}^{-1} \mathcal{M}_{\rm aper} \psi_0 \end{array} \right) 
+ i \tilde{\mu}_1 \left( \begin{array}{c} \mathcal{L}_{\rm aper}^{-1} \mathcal{S}_{\rm aper} \psi_0 \\ 0 \end{array} \right).
\label{second-level-one-solution}
\end{align}

\underline{\bf Order $\mathcal{O}(\epsilon^2)$:} We get the linear inhomogeneous system 
\begin{equation}
\left\{ \begin{array}{ll}
\mathcal{K}_{\rm aper} \hat{w}_2 &=  \lambda_1 \mathcal{M}_{\rm aper} \hat{v}_1  + i \tilde{\mu}_1  \mathcal{H}_{\rm aper} \hat{w}_1, \\
\mathcal{L}_{\rm aper} \hat{v}_2 &= \lambda_1 \left( \mathcal{M}^*_{\rm aper} \hat{w}_1 
-  2 c  \mathcal{H}_{\rm aper} \hat{v}_1 \right) + i \tilde{\mu}_1 \mathcal{S}_{\rm aper} \hat{v}_1.
\end{array}
\right. 
\label{second-level-two}
\end{equation}
By Fredholm's theorem, there exists the anti-periodic solution of the linear inhomogeneous system (\ref{second-level-two}) if and only if the following constraint is satisfied:
\begin{align}
\label{lin-alg}
\lambda_1 \langle \psi_0, \mathcal{M}^*_{\rm aper} \hat{w}_1 
-  2 c  \mathcal{H}_{\rm aper} \hat{v}_1 \rangle + i \tilde{\mu}_1 \langle \psi_0, \mathcal{S}_{\rm aper} \hat{v}_1 \rangle = 0.
\end{align} 
With the account of (\ref{second-level-one-solution}), the constraint (\ref{lin-alg}) yields the following quadratic characteristic equation 
\begin{equation}
\label{quad-eq-anti}
\lambda_1^2 \mathcal{A}_{1} + 2 i c \tilde{\mu}_1 \lambda_1 \mathcal{A}_{2} + \tilde{\mu}_1^2 \mathcal{A}_{3} = 0,
\end{equation} 
with the following coefficients:
\begin{align}
\left\{ \begin{array}{l}
\mathcal{A}_1 = \langle \mathcal{M}_{\rm aper} \psi_0, \mathcal{K}_{\rm aper}^{-1} \mathcal{M}_{\rm aper} \psi_0  \rangle - 4 c^2 \langle \mathcal{H}_{\rm aper} \psi_0, \mathcal{L}_{\rm aper}^{-1} \mathcal{H}_{\rm aper} \psi_0 \rangle, \\
\mathcal{A}_2 = \langle \mathcal{H}_{\rm aper} \psi_0,  \mathcal{L}_{\rm aper}^{-1} \mathcal{S}_{\rm aper} \psi_0 \rangle + \langle \mathcal{S}_{\rm aper} \psi_0,  \mathcal{L}_{\rm aper}^{-1} \mathcal{H}_{\rm aper} \psi_0 \rangle, \\
\mathcal{A}_3 = \langle  \mathcal{S}_{\rm aper}  \psi_0, \mathcal{L}_{\rm aper}^{-1} \mathcal{S}_{\rm aper} \psi_0 \rangle, \end{array} \right.
\label{coefficients-A}
\end{align}
where we have used the skew-adjointness of $\mathcal{H}_{\rm aper}$ and $\mathcal{S}_{\rm aper}$ in $L^2_{\rm per}$. 

If $\mathcal{A}_1 \neq 0$, two solutions of the quadratic equation (\ref{quad-eq-anti}) are given by 
\begin{equation}
\label{two-solutions-anti}
\lambda_1^{\pm}(\tilde{\mu}_1) = - \frac{ic \tilde{\mu}_1 \mathcal{A}_2}{\mathcal{A}_1} \pm \frac{i \tilde{\mu}_1}{\mathcal{A}_1} \sqrt{\mathcal{A}_1 \mathcal{A}_3 + c^2 \mathcal{A}_2^2},
\end{equation}
with the discriminant
\begin{align*}
\mathcal{A}_1 \mathcal{A}_3 + c^2 \mathcal{A}_2^2 &= \langle  \mathcal{M}_{\rm aper} \psi_0, \mathcal{K}_{\rm aper}^{-1} \mathcal{M}_{\rm aper} \psi_0  \rangle  \langle  \mathcal{S}_{\rm aper}  \psi_0, \mathcal{L}_{\rm aper}^{-1} \mathcal{S}_{\rm aper} \psi_0 \rangle \\
& \quad + 4 c^2 \left[ |\langle \mathcal{H}_{\rm aper} \psi_0,  \mathcal{L}_{\rm aper}^{-1} \mathcal{S}_{\rm aper} \psi_0 \rangle|^2  \right. \\
& \left. \qquad - \langle \mathcal{H}_{\rm aper} \psi_0, \mathcal{L}_{\rm aper}^{-1} \mathcal{H}_{\rm aper} \psi_0 \rangle \langle  \mathcal{S}_{\rm aper}  \psi_0, \mathcal{L}_{\rm aper}^{-1} \mathcal{S}_{\rm aper} \psi_0 \rangle
\right].
\end{align*}
If $\mathcal{A}_1 \neq 0$ and $\mathcal{A}_1 \mathcal{A}_3 + c^2 \mathcal{A}_2^2 > 0$, then $\lambda_1^{\pm}(\tilde{\mu}_1) \in i \R$ for $\tilde{\mu}_1 \in \R$ and $\lambda_1^+(\tilde{\mu}_1) \neq \lambda_1^-(\tilde{\mu}_1)$ for $\tilde{\mu}_1 \neq 0$. 
Persistence of the two spectral bands crossing $\lambda = 0$ on $i \mathbb{R}$ in the asymptotic expansions (\ref{expansion-2-eps}) and (\ref{expansion-2}) follows from the Hamiltonian symmetry of Lemma \ref{lem-mod-spectrum}. 
\end{proof}

\subsection{Normal form of the bifurcation}
\label{sec-4-2}

The bifurcation induced by Assumption \ref{ass-second-bif} is the period-doubling (subharmonic) bifurcation \cite{BDT2000b}. In addition, it is also related to the 
stability bifurcation for the spectral problem (\ref{mod-Bab-eq-anti}) for $\tilde{\mu} = 0$ in the space of antiperiodic perturbations in $H^1_{\rm aper} \times H^1_{\rm aper}$. The bifurcation is induced by the small eigenvalue of the linearized operator $\mathcal{L}_{\rm aper}$ which crosses $0$ as $c$ changes through the bifurcation point. 

The following lemma describes the transformation of the small eigenvalue 
of $\mathcal{L}_{\rm aper}$. 

\begin{lemma}
	\label{lem-period-doubling}
Let $c_0 \in (1,c_*)$ be the bifurcation point, under which Assumption \ref{ass-second-bif} holds with the $L^2$-normalized eigenfunction $\psi_0 \in H^1_{\rm aper}$. There exists an eigenvalue $\Lambda(c)$ of $\mathcal{L}_{\rm aper}$ with the $L^2$-normalized eigenfunction $\psi(\cdot,c)$, which are 
$C^{\infty}$ functions of $c$ at $c_0$ satisfying $\Lambda(c_0) = 0$ and $\psi(\cdot,c_0) = \psi_0$.
\end{lemma} 

\begin{proof}
	By the theory in \cite{BDT2000b}, the wave profile $\eta \in C^{\infty}_{\rm per}$  satisfies Assumption \ref{ass-kernel} if $c_0 \in (1,c_*)$ is the subharmonic bifurcation point, for which Assumption \ref{ass-second-bif} holds. Hence, the wave profile is continuously differentiable in $c$ at $c_0$ by Lemma \ref{lem-diff}. By the bootstrapping argument, it is actually $C^{\infty}$, so that the linearized operator $\mathcal{L}_{\rm aper}$ has smooth coefficients in $c$ at $c_0$ and we can write 
	\begin{equation}
\label{expansion-2-c}
	\mathcal{L}_{\rm aper} = 	\mathcal{L}_{\rm aper}^{(0)} + (c-c_0) 	\mathcal{L}_{\rm aper}^{(1)} + \mathcal{O}((c-c_0)^2), 
	\end{equation}
	where 
	\begin{align*}
	\mathcal{L}_{\rm aper}^{(0)} &= c_0^2 \mathcal{K}_{\rm aper} - (1 + \mathcal{K} \eta) - \eta \mathcal{K}_{\rm aper} - \mathcal{K}_{\rm aper} (\eta \; \cdot ), \\
	\mathcal{L}_{\rm aper}^{(1)} &= 2 c_0 \mathcal{K}_{\rm aper} -  \left( \mathcal{K} \partial_c \eta |_{c = c_0} \right) - \partial_c \eta |_{c = c_0} \mathcal{K}_{\rm aper} - \mathcal{K}_{\rm aper} \left( \partial_c \eta |_{c = c_0} \; \cdot \right).
	\end{align*}
As follows from the analytic perturbation theory \cite{Kir92}, the eigenvalue--eigenfunction pair $(\Lambda,\psi)$ of $\mathcal{L}_{\rm aper}$ is also smooth near $(0,\psi_0)$, since the kernel of $\mathcal{L}_{\rm aper}$ is simple by Assumption \ref{ass-second-bif}. This gives the assertion of the lemma.
\end{proof}

\begin{remark}
By using the decomposition 
\begin{equation*}
	\Lambda(c) = \Lambda'(c_0) (c-c_0) + \mathcal{O}((c-c_0)^2), \quad 
	\psi(\cdot,c) = \psi_0 + \psi_1 (c - c_0)  + \mathcal{O}((c-c_0)^2)
\end{equation*}
in the solutions of $\mathcal{L}_{\rm aper} \psi(\cdot,c) = \Lambda(c) \psi(\cdot,c)$, we obtain the linear inhomogeneous equation for $\psi_1 \in H^1_{\rm aper}$:
	$$
	\mathcal{L}_{\rm aper}^{(0)} \psi_1 + 	\mathcal{L}_{\rm aper}^{(1)} \psi_0 = \Lambda'(c_0) \psi_0. 
	$$
Projecting this equation to $\psi_0$ yields the Fredholm condition for the existence of $\psi_1 \in H^1_{\rm aper}$:
\begin{equation}
\label{Lambda-der}
	\Lambda'(c_0) = \frac{\langle \mathcal{L}_{\rm aper}^{(1)} \psi_0, \psi_0 \rangle}{\langle \psi_0,\psi_0 \rangle},
\end{equation}
	which is generally nonzero.
\end{remark}

The following theorem gives the normal form for the stability bifurcation in the spectral problem (\ref{mod-Bab-eq-anti}).

\begin{theorem}
	\label{theorem-bif-2-normal-form}
	Under the assumptions of Theorem \ref{theorem-bif-2}, let $c_0 \in (1,c_*)$ be the root of $\Lambda(c)$, where $\Lambda(c)$ is the eigenvalue of $\mathcal{L}_{\rm aper}$ in Lemma \ref{lem-period-doubling}. Assume that $\mathcal{A}_1 \neq 0$ and $\Lambda'(c_0) \neq 0$ in (\ref{coefficients-A}) and (\ref{Lambda-der}). The two spectral bands of the modulational stability problem (\ref{mod-Bab-eq-anti}) near $\lambda = 0$ in Theorem \ref{theorem-bif-2} are obtained from roots of the characteristic 
	equation
	\begin{equation}
	\label{normal-form-two}
	\lambda^2\mathcal{A}_{1} + 2 i c_0 \tilde{\mu} \lambda \mathcal{A}_{2} + \tilde{\mu}^2 \mathcal{A}_{3} 
	- \Lambda'(c_0) (c-c_0) \| \psi_0 \|^2_{L^2} = \mathcal{O}\left( |\tilde{\mu}|^3 \right),
	\end{equation}
	where the asymptotic balance is justified if 
	$c - c_0 = \mathcal{O}( \tilde{\mu}^2)$.
\end{theorem} 

\begin{proof}
We combine the asymptotic expansions (\ref{expansion-2-eps}), (\ref{expansion-2}), and (\ref{expansion-2-c}) at the asymptotic balance $c - c_0 = \mathcal{O}(\tilde{\mu}^2)$. To incorporate the balance, we write 
$$
c = c_0 + \epsilon^2 c_2 + \mathcal{O}(\epsilon^3),
$$
where $c_2$ is fixed. We get the same solution (\ref{second-level-one-solution}) 
with $c = c_0$ at the order $\mathcal{O}(\epsilon)$ and we get the following modification of the linear inhomogeneous system (\ref{second-level-two}) at the order $\mathcal{O}(\epsilon^2)$:
\begin{equation}
\left\{ \begin{array}{ll}
\mathcal{K}_{\rm aper} \hat{w}_2 &=  \lambda_1 \mathcal{M}_{\rm aper} \hat{v}_1  + i \tilde{\mu}_1  \mathcal{H}_{\rm aper} \hat{w}_1, \\
\mathcal{L}_{\rm aper}^{(0)} \hat{v}_2 + c_2 \mathcal{L}_{\rm aper}^{(1)} \psi_0 &= \lambda_1 \left( \mathcal{M}^*_{\rm aper} \hat{w}_1 
-  2 c_0  \mathcal{H}_{\rm aper} \hat{v}_1 \right) + i \tilde{\mu}_1 \mathcal{S}_{\rm aper} \hat{v}_1.
\end{array}
\right. 
\label{second-level-two-mod}
\end{equation}
By Fredholm's theorem, there exists a periodic solution of the linear inhomogeneous system (\ref{second-level-two-mod}) if and only if the following constraint is satisfied:
\begin{align}
\label{lin-alg-mod}
\lambda_1 \langle \psi_0, \mathcal{M}^*_{\rm aper} \hat{w}_1 
-  2 c_0  \mathcal{H}_{\rm aper} \hat{v}_1 \rangle + i \tilde{\mu}_1 \langle \psi_0, \mathcal{S}_{\rm aper} \hat{v}_1 \rangle = c_2 \langle \psi_0, \mathcal{L}_{\rm aper}^{(1)} \psi_0 \rangle.
\end{align} 
With the account of (\ref{quad-eq-anti}) and (\ref{Lambda-der}), the constraint (\ref{lin-alg-mod}) yields the following quadratic characteristic equation 
\begin{equation}
\label{quad-eq-anti-mod}
\lambda_1^2 \mathcal{A}_{1} + 2 i c_0 \tilde{\mu}_1 \lambda_1 \mathcal{A}_{2} + \tilde{\mu}_1^2 \mathcal{A}_{3} = c_2 \Lambda'(c_0) \| \psi_0 \|^2_{L^2},
\end{equation} 
where the coefficients $\mathcal{A}_{1,2,3}$ are defined by (\ref{coefficients-A}). Multiplying (\ref{quad-eq-anti-mod}) by $\epsilon^2$, keeping in mind the remainder term of the order $\mathcal{O}(\epsilon^3)$, and returning to the original variables $\mu$, $\lambda$, and $c$, we obtain (\ref{normal-form-two}).
\end{proof}

\begin{remark}
	\label{rem-anti-per}
If $\mathcal{A}_1 \neq 0$ and $\mathcal{A}_1 \mathcal{A}_3 + c_0^2 \mathcal{A}_2^2 > 0$, then $\lambda^2\mathcal{A}_{1} + 2 i c_0 \tilde{\mu} \lambda \mathcal{A}_{2} + \tilde{\mu}^2 \mathcal{A}_{3} = 0$ has two roots $\lambda^{\pm}(\tilde{\mu}) \in i \mathbb{R}$ such that $\lambda^+(\tilde{\mu}) \neq \lambda^-(\tilde{\mu})$ for $\tilde{\mu} \neq 0$, see (\ref{two-solutions-anti}). If ${\rm sgn}(\Lambda'(c_0)) = {\rm sgn}(\mathcal{A}_1)$, then the spectral problem (\ref{mod-Bab-eq-anti}) with $\tilde{\mu}_1 = 0$ admits two real eigenvalues for $c > c_0$ and two purely imaginary eigenvalues for $c < c_0$ near the origin. Rewriting (\ref{normal-form-two}) as 
	$$
\left( \lambda + \frac{i c_0 \tilde{\mu} \mathcal{A}_2}{\mathcal{A}_1} \right)^2 + 
\tilde{\mu}^2 \frac{\mathcal{A}_1 \mathcal{A}_{3} + c^2 \mathcal{A}_2^2}{\mathcal{A}_1^2}
	- \frac{\Lambda'(c_0) \| \psi_0 \|^2_{L^2}}{\mathcal{A}_1} (c-c_0) = \mathcal{O}\left( |\tilde{\mu}|^3 \right),
	$$
we can see that the roots define two spectral bands, which are subsets of $i \mathbb{R}$ and non-degenerate for $c < c_0$. Compared to the case of $c = c_0$, these two spectral bands do not intersect $\lambda = 0$ for $c < c_0$. If $c > c_0$, the roots define a new circular band around the origin, the radius of which grows with further deviation of $c$ from $c_0$.
\end{remark}

\subsection{Numerical results}
\label{sec-4-3}

The criterion for the ellipse bifurcation is a zero of $\Lambda(c)$, the eigenvalue of $\mathcal{L}_{\rm aper}$ for anti-periodic eigenfunction (see red circles in Figure~\ref{fig:beigs}). The corresponding steepness $s_0$ and speed $c_0$ for the first two zeros of $\Lambda$ are recorded in Table \ref{tab:ellipse}, where we also show the values of $\mathcal{A}_1$, $\mathcal{A}_2$, $\mathcal{A}_1 \mathcal{A}_3 + c_0^2 \mathcal{A}_2^2$, and $\Lambda'(c_0)$ in the normal form (\ref{normal-form-two}).
 
\begin{table}[htb!]
	\centering
	\renewcommand{\arraystretch}{1.3}
	\begin{tabular}{c|c|c|c|c|c}
		$s_0$ & $c_0$ &  $\mathcal{A}_1$ &  $\mathcal{A}_2$ & $\mathcal{A}_1\mathcal{A}_3 + c^2 \mathcal{A}_2^2$ & $\Lambda'(c_0)$ \\ \hline
		$0.12890310$  & $1.08414013$ & $-4.7027$ & $1.3601$ & $0.99841$ & $    -28.87$ \\
		$0.14048675$  & $1.09238125$ & $-4.4533$ & $1.2802$ & $0.89648$ & $2391.7$\\ 
	\end{tabular}
	\caption{Parameters of the first two ellipse bifurcations.}
	\label{tab:ellipse}
\end{table}

We note that  $\mathcal{A}_1 \neq 0$ and $\mathcal{A}_1 \mathcal{A}_3 + c^2 \mathcal{A}_2^2 > 0$ so that the splitting of two spectral bands near $\lambda = 0$ occurs similarly to the description in Remark \ref{rem-anti-per}. The only difference is that the mapping $s \mapsto c(s)$ is increasing at the first 
ellipse bifurcation and decreasing at the second ellipse bifurcation. Therefore, 
the values of $\Lambda'(c_0)$ change the sign in between the two bifurcations, which correspond to the real eigenvalues of the spectral problem  (\ref{mod-Bab-eq-anti}) with $\tilde{\mu}_1 = 0$ for $c > c_0$ for the first bifurcation and $c < c_0$ for the second bifurcation, both regions correspond to 
the steepness $s$ above the critical steepness $s_0$.
	 
There are two equivalent ways to determine the eigenfunction $\psi_0 \in H^1_{\rm aper}$. The first method is to solve the eigenvalue problem for 
$\mathcal{L}_{\rm mu}$ using Floquet multiplier with $\mu = \frac{1}{2}$. The second is to extend the interval for $u$ to $[0,4\pi]$, and
seek the unique odd $4\pi$-periodic eigenfunction of $\mathcal{L}_{\rm aper}$. We chose to follow the second route, where we also use the symmetry of Babenko's equation (\ref{trav-Bab-eq}): if $\eta \in C^{\infty}_{\rm per}$ is a solution for speed $c$, then $\tilde{\eta}(u) = \frac{1}{2} \eta(2u)$ is also a solution for speed $\tilde{c} = \frac{c}{\sqrt{2}}$. This scaling symmetry allows us to recast a $4\pi$-periodic function back to a $2\pi$-periodic function. We obtain $\psi_0$ from the shift-and-invert numerical method applied to $\mathcal{L}_{\rm aper}$, see ~\cite{dyachenko2023quasiperiodic}. The value of $\Lambda'(c_0)$ is determined from the equation~\eqref{Lambda-der}. After $\psi_0 \in H^1_{\rm aper}$ is found, 
we compute the three inner products $A_{1,2,3}$ in relations~\eqref{coefficients-A}, with the inverse operator $\mathcal{L}^{-1}$ applied via the MINRES method. 

Figures \ref{fig:efg1} and \ref{fig:efg2} show 
the instability bands of the spectral problem (\ref{mod-Bab-eq}) near the origin in the spectral $\lambda$-plane obtained numerically (blue dots).  
The three panels in each figure correspond to three consequent values of 
steepness before, at, and after the critical steepness, for which $\Lambda(c_0) = 0$.  The theoretical curves obtained from the normal form~\eqref{normal-form-two} are  shown by red color. A good agreement between the spectral bands given by the normal form (\ref{normal-form-two}) and the numerical approximations of the stability spectra is evident for the first two cycles of bifurcations.

\begin{figure}[htb!]
    \centering
    \includegraphics[width=0.95\linewidth]{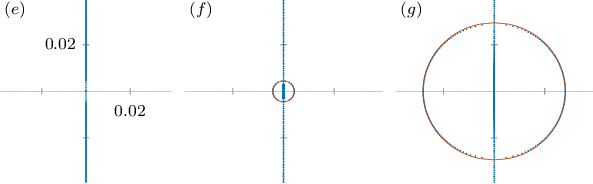}
    \caption{The first instance of the ellipse bifurcation, 
    	which corresponds to panels (e)--(g) of Figure \ref{fig-bif} zoomed near the origin in the $\lambda$-plane. The spectral bands remain on $i \mathbb{R}$ for $c < c_0$ (left panel), remain on $i \mathbb{R}$ and intersect the origin for $c = c_0$ (middle panel), 
    	and form a circular band around the origin for $c > c_0$ (right panel), 
where $c_0 \approx 1.08414013$ is found from the root of $\Lambda$. The dashed lines show the approximations obtained from the normal form~\eqref{normal-form-two}.}
    \label{fig:efg1}
\end{figure}

\begin{figure}[htb!]
	\centering
	\includegraphics[width=0.95\linewidth]{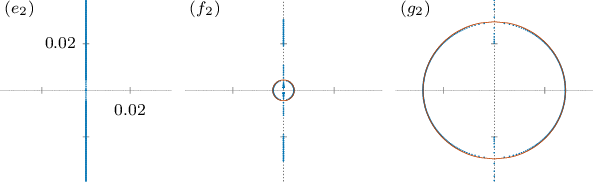}
	\caption{The same as Figure \ref{fig:efg1} but for the second instance 
		of the ellipse bifurcation, which happens when $\Lambda(c_0) = 0$ for $c_0 \approx 1.09238125$.}
	\label{fig:efg2}
\end{figure}

\section{Reconnection of figure-$\infty$ bands}
\label{sec-5}

The corresponding  bifurcation is seen in panels (j)--(l) of Figure \ref{fig-bif}.
For simplicity, it can be called the figure-$\infty$ bifurcation.
Before the bifurcation [panel (j)], the instability band surrounding the origin 
has a peanut shape.  At the bifurcation [panel (k)], the instability 
band reconnects at the origin and exhibits the figure-$\infty$ shape.  After the bifurcation [panel (l)], the figure-$\infty$ band breaks into two instability bubbles along the real axis.

Section \ref{sec-5-1} develops the criterion and the normal form for this bifurcation. Section \ref{sec-5-2} describes 
numerical approximations of the coefficients of the normal form, 
from which we are able to reproduce the bifurcation pattern seen in panels (j)--(l) of Figure \ref{fig-bif}.

\subsection{Criterion and the normal form of the bifurcation}
\label{sec-5-1}

This bifurcation is due to the zero-eigenvalue bifurcation in the spectral problem (\ref{spec-Bab-eq}) for co-periodic perturbations, which was considered in our previous work \cite{DP2025}. The geometric multiplicity of the zero eigenvalue is controlled by Assumption \ref{ass-kernel}, which is still considered to be true. However, the algebraic multiplicity of the zero eigenvalue is increased from four to six, since Assumption \ref{ass-first-bif} is no longer considered to be true. Hence, we replace Assumption \ref{ass-first-bif} with the opposite non-generic assumption.

\begin{assumption}
	\label{ass-third-bif}
	For the given value of $c \in (1,c_*)$, we have $\mathcal{P}'(c) = 0$.
\end{assumption}

Because of Assumption \ref{ass-third-bif}, the generalized kernel of the spectral problem (\ref{mod-Bab-eq}) for $\mu = 0$ is at least six-dimensional. 
It was shown in \cite{DP2025} that it is exactly six-dimensional if $\mathcal{B} \neq 0$, where the numerical coefficient $\mathcal{B}$ is defined by (\ref{B-coef}). 

The following theorem gives the criterion for the corresponding bifurcation. 

\begin{theorem}
	\label{theorem-bif-3}
	Let Assumptions  \ref{ass-kernel} and \ref{ass-third-bif} be true and assume that $\mathcal{B} \neq 0$ in (\ref{B-coef}). There exist six spectral bands of the modulational stability problem (\ref{mod-Bab-eq}) intersecting $\lambda = 0$. Two spectral bands cross the origin along the straight lines in all four complex quadrants of the complex $\lambda$-plane and four spectral bands 
	are subsets of $i \mathbb{R}$ and non-degenerate. 
\end{theorem}

\begin{proof}
	We revisit the proof of Theorem \ref{theorem-bif-1} after replacing Assumption \ref{ass-first-bif} with Assumption \ref{ass-third-bif}. If $\mathcal{P}'(c) = 0$, then the characteristic equation (\ref{quad-eq}) 
	fails because the coefficient in front of $\lambda_1^2$ vanishes. 
	If $\mathcal{B} \neq 0$ in (\ref{B-coef}), then the characteristic 
	equation can be extended to the next term of $\lambda_1^4$. 
According to Section \ref{sec-case-3}, we get the splitting formula for eigenvalues in the form (\ref{eigenvalue-split-5}), where three eigenvalues with 
$\lambda = \mathcal{O}(\mu^{1/3})$ and one eigenvalue with $\lambda = \mathcal{O}(\mu)$ are obtained in the subspace $(a_1,a_2) = (1,0)$ 
and two eigenvalues with $\lambda = \mathcal{O}(\mu^{1/2})$ are obtained 
in the subspace $(a_1,a_2) = (0,1)$. Since the two balances 
	are different, we can still consider these two subspaces separately from each other, similarly to the proof of Theorem \ref{theorem-bif-1}, with two different expansions near $(\mu,\lambda) = (0,0)$.
	
	\vspace{0.2cm}

\underline{Subspace $(a_1,a_2) = (1,0)$:} We introduce the asymptotic expansions
	\begin{equation}
\label{expansion-3-eps}
	\mu = \epsilon^3 \mu_3, \quad \lambda = \epsilon \lambda_1 + \mathcal{O}(\epsilon^2)
	\end{equation}
and 
	\begin{equation}
	\label{expansion-3}
	\left( \begin{array}{c} \hat{v} \\ \hat{w} \end{array} \right) = \left( \begin{array}{c} \eta' \\ 0 \end{array} \right) + \epsilon \left( \begin{array}{c} \hat{v}_1 \\ \hat{w}_1 \end{array} \right) + \epsilon^2 \left( \begin{array}{c} \hat{v}_2 \\ \hat{w}_2 \end{array} \right) + \epsilon^3 \left( \begin{array}{c} \hat{v}_3 \\ \hat{w}_3 \end{array} \right) + \epsilon^4 \left( \begin{array}{c} \hat{v}_4 \\ \hat{w}_4 \end{array} \right) + \mathcal{O}(\epsilon^5),
	\end{equation}
where $\epsilon > 0$ is a small parameter, $\mu_3 \neq 0$ is fixed, and the correction terms up to the order of $\mathcal{O}(\epsilon^4)$ are to be computed recursively in function 
space $H^1_{\rm per} \times H^1_{\rm per}$. 
	
\underline{\bf Order $\mathcal{O}(\epsilon)$:} We get the linear inhomogeneous system 
(\ref{level-one}) without the $\mathcal{O}(\mu_1)$ terms. Hence, instead of the solution (\ref{level-one-solution}), we obtain the truncated solution 
	\begin{align}
	\left( \begin{array}{c} \hat{v}_1 \\ \hat{w}_1 \end{array} \right) = \lambda_1 
	\left[ \left( \begin{array}{c} -\partial_c \eta \\ \mathcal{H} \eta \end{array} \right) \pm i \langle 1, \partial_c \eta \rangle 
	\left( \begin{array}{c} \eta' \\ -c \end{array} \right) 
	\right], \label{third-level-one}
	\end{align}
	where no projections to $(\eta',0)^T$ and $(0,1)^T$ are needed to satisfy the Fredholm conditions in the higher orders.
	
\underline{\bf Order $\mathcal{O}(\epsilon^2)$:} We get the linear inhomogeneous system (\ref{level-two}) without the $\mathcal{O}(\mu_1)$ terms. Moreover, the second term in (\ref{third-level-one}) cancels all the plus/minus projection terms due to Lemma \ref{lem-transformation}. By using (\ref{der-c}) and $\mathcal{P}'(c) = 0$, we confirm that the solvability constraints are satisfied and there exists a periodic solution of the linear inhomogeneous system in the form	
	\begin{align}
	\left( \begin{array}{c} \hat{v}_2 \\ \hat{w}_2 \end{array} \right) = \lambda_1^2  \left( \begin{array}{c} \tilde{v}_2 \\ \tilde{w}_2 \end{array} \right) + a \left( \begin{array}{c} 0 \\ 1 \end{array} \right), \label{third-level-two}
	\end{align}
where $a \in \mathbb{R}$ is arbitrary at this order, the projection to $(\eta',0)^T$ is not included since the eigenfunctions are defined up to the arbitrary scalar multiplication in (\ref{expansion-3}), and 
$(\tilde{v}_2,\tilde{w}_2) \in H^1_{\rm per} \times H^1_{\rm per}$ is uniquely determined from solution of the linear inhomogeneous system
	\begin{equation}
	\left\{ \begin{array}{ll}
	\mathcal{K} \tilde{w}_2 &=  - \mathcal{M} \partial_c \eta, \\
	\mathcal{L} \tilde{v}_2 &=  \mathcal{M}^* \mathcal{H}\eta + 2 c  \mathcal{H}  \partial_c \eta,
	\end{array}
	\right. 
	\label{lambda-squared}
	\end{equation}
under the orthogonality conditions $\langle 1, \tilde{w}_2 \rangle = \langle \eta', \tilde{v}_2 \rangle = 0$. Since the linearized operators in (\ref{lambda-squared}) are parity-preserving, the functions $\tilde{v}_2$ and $\tilde{w}_2$ are spatially odd and even in $u$, respectively. 
	
\underline{\bf Order $\mathcal{O}(\epsilon^3)$:} We get the linear inhomogeneous system 
\begin{equation}
\left\{ \begin{array}{ll}
\mathcal{K} \hat{w}_3 &=  \lambda_1 \mathcal{M} \hat{v}_2, \\
\mathcal{L} \hat{v}_3 &= \lambda_1 \left( \mathcal{M}^* \hat{w}_2 -  2 c \mathcal{H} \hat{v}_2 \right) + i \mu_3 \mathcal{L}_{\pm} \eta'.
\end{array}
\right. 
\label{level-three}
\end{equation}
By using (\ref{kernel-generalized}), (\ref{Upsilon}), (\ref{level-one-solution}), and (\ref{third-level-two}), there exists a periodic solution to (\ref{level-three}) in the form
\begin{align}
\left( \begin{array}{c} \hat{v}_3 \\ \hat{w}_3 \end{array} \right) &= \lambda_1^3 
\left[ \left( \begin{array}{c} \tilde{v}_3 \\ \tilde{w}_3 \end{array} \right) \mp i \langle 1, \tilde{v}_3 \rangle 
\left( \begin{array}{c} \eta' \\ -c \end{array} \right) 
\right]
+ i \mu_3 \left( \begin{array}{c} \Upsilon \\ 0 \end{array} \right) \notag \\
& \qquad \qquad 
+ a \lambda_1 \left[ 
\left( \begin{array}{c} -1 \\ 0 \end{array} \right) \pm i \left( \begin{array}{c} \eta' \\ -c \end{array} \right) \right],
\label{third-level-three}
\end{align}
where projections to $(\eta',0)^T$ and $(0,1)^T$ are not included since the eigenfunctions are defined up to the arbitrary scalar multiplication in (\ref{expansion-3}) and (\ref{third-level-two}), and $(\tilde{v}_3,\tilde{w}_3)  \in H^1_{\rm per} \times H^1_{\rm per}$ is uniquely determined from solution of the linear inhomogeneous system
\begin{equation}
\left\{ \begin{array}{ll}
\mathcal{K} \tilde{w}_3 &=  \mathcal{M} \tilde{v}_2, \\
\mathcal{L} \tilde{v}_3 &=  \mathcal{M}^* \tilde{w}_2 - 2 c  \mathcal{H}  \tilde{v}_2,
\end{array}
\right. 
\label{lambda-cubic}
\end{equation}
under the orthogonality conditions $\langle 1, \tilde{w}_3 \rangle = \langle \eta', \tilde{v}_3 \rangle = 0$. Since the linearized operators in (\ref{lambda-cubic}) are parity-preserving, the functions $\tilde{v}_3$ and $\tilde{w}_3$ are spatially even and odd in $u$, respectively. 
A projection term was added in (\ref{third-level-three}) due to the representation (\ref{transformation}) since $\langle 1, \tilde{v}_3 \rangle \neq  0$ generally.

\underline{\bf Order $\mathcal{O}(\epsilon^4)$:} We get the linear inhomogeneous system 
\begin{equation}
\left\{ \begin{array}{ll}
\mathcal{K} \hat{w}_4 &=  \lambda_1 \left( \mathcal{M} \hat{v}_3 \pm i \langle 1, \hat{v}_3 \rangle \eta' \right) + i \mu_3 \left( \mathcal{H} \hat{w}_1 \pm i \langle 1, \hat{w}_1 \rangle \right), \\
\mathcal{L} \hat{v}_4 &= \lambda_1 \left( \mathcal{M}^* \hat{w}_3 
-  2 c  \mathcal{H} \hat{v}_3 \mp i \langle \eta', \hat{w}_3 \rangle \mp 2 c i \langle 1, \hat{v}_3 \rangle \right) + i \mu_3 \mathcal{L}_{\pm} \hat{v}_1.
\end{array}
\right. 
\label{level-four}
\end{equation}
By Fredholm's theorem, there exists a periodic solution of the linear inhomogeneous system (\ref{level-four}) if and only if the two constraints are satisfied:
\begin{align}
\label{lin-alg-3}
\lambda_1 \langle 1, \mathcal{M} \hat{v}_3 \rangle \mp \mu_3 \langle 1, \hat{w}_1 \rangle &= 0, \\
\lambda_1 \langle \eta', \mathcal{M}^* \hat{w}_3 
-  2 c  \mathcal{H} \hat{v}_3 \rangle + i \mu_3 \langle \eta', \mathcal{L}_{\pm} \hat{v}_1 \rangle &= 0.
\label{lin-alg-4}
\end{align} 
Since $\mathcal{M}^* 1 = 1 + 2 \mathcal{K}\eta$, the first equation (\ref{lin-alg-3}) yields due to (\ref{third-level-one}) and (\ref{third-level-three}):
$$
\lambda_1^4 \langle (1+2\mathcal{K} \eta), \tilde{v}_3 \rangle + 
i \mu_3 \lambda_1 \langle (1 + 2 \mathcal{K} \eta), \Upsilon \rangle - a \lambda_1^2 + i c \mu_3 \lambda_1 \langle 1, \partial_c \eta \rangle = 0.
$$
This equation defines $a$ uniquely as 
\begin{align}
\label{a-third}
a = \lambda_1^2 \langle (1+2\mathcal{K} \eta), \tilde{v}_3 \rangle,
\end{align}	
where the term of $\mu_3 \lambda_1^{-1}$ has zero coefficients due to  (\ref{tech-eq-11}) and (\ref{a-first}) with $\mathcal{P}'(c) = 0$. Similarly, the second equation (\ref{lin-alg-4}) yields due to (\ref{third-level-one}) and (\ref{third-level-three}):
\begin{align}
\label{char-eq-tech}
\lambda_1^4 \left[ \langle \eta', \tilde{w}_3 \rangle - 2 c \langle \mathcal{K} \eta, \tilde{v}_3 \rangle \right] - i \mu_3 \lambda_1 \left[ \langle \eta, \partial_c \eta \rangle + 2 c \langle \mathcal{K}\eta, \Upsilon \rangle \right]  = 0,
\end{align}
We note again that the numerical coefficient $a$ given by (\ref{a-third}) does not contribute to the computation of the characteristic equation (\ref{char-eq-tech}). By using the numerical coefficient $\mathcal{B}$ given by 
\begin{equation}
\label{B-coef}
\mathcal{B} := \langle \eta', \tilde{w}_3 \rangle - 2 c \langle \mathcal{K} \eta, \tilde{v}_3 \rangle,
\end{equation}	
we obtain the quartic characteristic equation from (\ref{char-eq-tech}):
	\begin{equation}
	\label{quartic-eq}
	\lambda_1^4 \mathcal{B} - i \mu_3 \lambda_1 \mathcal{N}'(c)  = 0.
	\end{equation}
Due to the second identity in (\ref{identity1}), if $\mathcal{P}'(c) = 0$, then $\mathcal{N}'(c) = -5 \mathcal{P}(c) < 0$. 
Since $\mathcal{B} \neq 0$ is assumed, the quartic equation admits three nonzero solutions $\lambda_1^{(1,2,3)}(\mu_3)$ given by 
	$$
	\lambda_1^{(j)}(\mu_3) = (\mu_3 \mathcal{B}^{-1} \mathcal{N}'(c))^{1/3} e^{\frac{ij\pi}{6}}, \quad j = 1,2,3, 
	$$
two of which give two straight lines in all four quadrants of the complex $\lambda$-plane and one gives the vertical segment on $i \mathbb{R}$. 
The fourth root in (\ref{quartic-eq}) is zero at this order. According to the analysis in Section \ref{sec-case-3}, the fourth root is of the order $\lambda = \mathcal{O}(\mu)$.

For fixed $\mu_3 \neq 0$, the three eigenvalues $\lambda_1^{(1,2,3)}$ are distinct. The eigenvalue on $i \R$ remains on $i \R$ due to the Hamiltonian symmetry of Lemma \ref{lem-mod-spectrum}. The other two simple eigenvalues are complex-valued and remain in the corresponding quadrants of the complex $\lambda$-plane under the perturbation terms in the asymptotic expansions (\ref{expansion-3-eps}) and (\ref{expansion-3}). The fourth eigenvalue of the order of $\lambda = \mathcal{O}(\mu)$ is simple and belongs to $i \mathbb{R}$ due to Hamiltonian symmetry.

\vspace{0.2cm}

\underline{Subspace $(a_1,a_2) = (0,1)$:} We use the same asymptotic expansions 
(\ref{expansion-1-second-eps}) and (\ref{expansion-1-second}) as in the proof of Theorem \ref{theorem-bif-1}. The same computations give two roots $\lambda_1^{\pm}(\mu_2) = \pm \sqrt{|\mu_2|}$ with two roots on $i \mathbb{R}$, which remain on $i \R$ under the perturbation terms due to the Hamiltonian symmetry of Lemma \ref{lem-mod-spectrum}. 
\end{proof}

\begin{remark}
Let $c = c_0$ be the root of $\mathcal{P}'(c)$ as in Assumption \ref{ass-third-bif} and assume $\mathcal{P}''(c_0) \neq 0$. 
Continuation of the characteristic equation (\ref{quartic-eq}) for $c$ close to $c_0$ was performed in \cite{DP2025}, from which we obtain the normal form 
for the corresponding bifurcation:
	\begin{equation}
	\label{normal-form-three}
	\lambda^4 \mathcal{B} - i \mu \lambda \mathcal{N}'(c_0) + \lambda^2 \mathcal{P}''(c_0) (c - c_0) = \mathcal{O}(\mu^2),	
	\end{equation}
where $c - c_0 = \mathcal{O}(|\mu|^{2/3})$ at the leading-order balance. For any $c \neq c_0$, there exist two roots $\lambda$ of (\ref{normal-form-three}) on $i \mathbb{R}$. Indeed, the transformation $\lambda = i \omega$ yields 
	\begin{equation}
	\label{normal-form-three-real}
	\omega^4 \mathcal{B} + \mu \omega \mathcal{N}'(c_0) - \omega^2 \mathcal{P}''(c_0) (c - c_0) = \mathcal{O}(\mu^2),	
	\end{equation}
where the coefficients in the remainder term are real-valued due to the Hamiltonian symmetry of Lemma \ref{lem-mod-spectrum}. There exists two roots $\omega$ of (\ref{normal-form-three-real}) on $\R$: one root is of the order  $\mathcal{O}(|\mu|^{1/3})$ and the other root is of the order $\mathcal{O}(\mu)$ 
for any $c \neq c_0$. These two real roots for $\omega$ define two roots $\lambda$ of (\ref{normal-form-three}) on $i\R$. The other two roots of (\ref{normal-form-three}) are complex-valued for any $c \neq c_0$.
\end{remark} 

\subsection{Numerical results}
\label{sec-5-2}

The criterion for the bifurcation is a zero of $\mathcal{P}'(c)$, which is equivalent to a zero of $\mathcal{H}'(c)$ due to (\ref{identity1}). The chain of eigenfunctions $(\tilde v_2, \tilde w_2)$ and $(\tilde v_3, \tilde w_3)$ defined by the linear inhomogeneous equations ~\eqref{lambda-squared} and~\eqref{lambda-cubic} only exists at $\mathcal{H}'(c)=0$. It follows from Figure \ref{fig:delta} that there exists exactly one such zero of $\mathcal{H}'(c)$ (red solid line) between each turning point of speed with respect to steepness where $c'(s) = 0$ (vertical lines). The corresponding steepness $s_0$ and speed $c_0$ for the first two zeros of $\mathcal{H}'(c)$ are recorded in Table \ref{tab:inf}, where we also show  
the coefficients $\mathcal{P}''(c_0)$, $\mathcal{B}$, and $\mathcal{N}'(c_0)$ of the normal form~\eqref{normal-form-three}. We note that  $\mathcal{B} > 0$ and the sign of $\mathcal{P}''(c_0)$ alternates between the two bifurcations. In both cases, this leads to new unstable (real) eigenvalues of the spectral problem (\ref{spec-Bab-eq}) in $L^2_{\rm per}$ in the direction of increasing steepness $s$ past the critical steepness $s_0$ \cite{DP2025}.
 
\begin{table}[htb!]
	\centering
	\renewcommand{\arraystretch}{1.3}
	\begin{tabular}{c|c|c|c|c}
        $s_0$ & $c_0$ & $\mathcal{P}''(c_0)$ & $\mathcal{B}$ & $\mathcal{N}'(c_0)$ \\ \hline
        $0.13660355$  & $1.09213793$  &  $-2.0450\times10^{3}$ & $11.018$ & $-0.35594$\\
        $0.14079655$  & $1.09228678$  &   $1.6623\times10^{5}$ & $10.962$ & 
        $-0.35050$\\
	\end{tabular}
	\caption{Parameters of the first two figure-$\infty$ bifurcations.}
	\label{tab:inf}
\end{table}

Figures \ref{fig:jkl1} and \ref{fig:jkl2} show 
the instability bands of the spectral problem (\ref{mod-Bab-eq}) near the origin in the spectral $\lambda$-plane obtained numerically.  
The three panels in each figure correspond to three consequent values of 
steepness before, at, and after the critical steepness, for which $\mathcal{H}'(c_0) = 0$.  The red lines display the complex spectral bands described by the normal form~\eqref{normal-form-three}, and are compared to the spectral bands shown in blue dots. The agreement is perfect for the first two cycles of bifurcations.

\begin{figure}[htb!]
	\centering
	\includegraphics[width=0.95\linewidth]{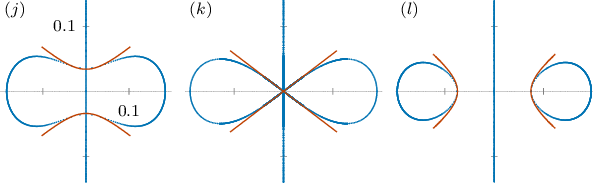}
	\caption{The first instance of the figure-$\infty$ bifurcation, 
		which corresponds to panels (j)--(l) of Figure \ref{fig-bif} zoomed near the origin in the $\lambda$-plane. The spectral band encloses the origin but gets the peanut shape for $c < c_0$ (left panel), reconnects at the origin at the figure-$\infty$ shape for $c = c_0$ (middle panel), 
		and splits into two circular bubbles along the real axis for $c > c_0$ (right panel), where $c_0 \approx 1.09213793$ is found from the root of $\mathcal{H}'(c)$. The dashed lines show the approximations obtained from the normal form~\eqref{normal-form-three}.}
	\label{fig:jkl1}
\end{figure}

\begin{figure}[htb!]
	\centering
	\includegraphics[width=0.95\linewidth]{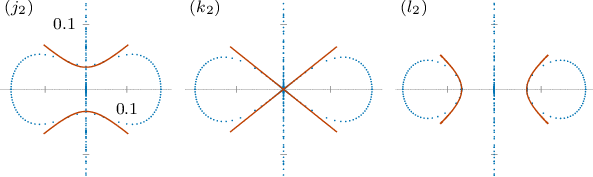}
	\caption{The same as Figure \ref{fig:jkl1} but for the second instance 
		of the figure-$\infty$ bifurcation, which happens when $\mathcal{H}'(c_0) = 0$ for $c_0 \approx 1.09228678$.}
	\label{fig:jkl2}
\end{figure}

\section{New figure-$8$ bands arising from the origin}
\label{sec-6}

The corresponding  bifurcation is seen in panels (m)--(o) of Figure \ref{fig-bif}.
For simplicity, we call it the figure-$8$ bifurcation.
Before the bifurcation [panel (m)], there are no circular bands around the origin. Past the bifurcation [panels (n) and (o)], new figure-$8$ bands appear from the origin and grow in size.

Section \ref{sec-6-1} develops the criterion for this bifurcation. Section \ref{sec-6-2} gives the derivation of the normal form, 
which explains how the bifurcation is unfolded. Section \ref{sec-6-3} describes 
numerical approximations of the coefficients of the normal form, 
from which we are able to reproduce the bifurcation pattern seen in panels (m)--(o) of Figure \ref{fig-bif}.

\subsection{Criterion of the bifurcation} 
\label{sec-6-1}

This bifurcation is due to the double zero eigenvalue of the linear operator 
$$
\mathcal{L} : H^1_{\rm per} \subset L^2_{\rm per} \to L^2_{\rm per},
$$
for which Assumption \ref{ass-kernel} is no longer considered to be true. 
This corresponds to the fold bifurcation in the existence of the traveling waves 
with the profile $\eta \in C^{\infty}_{\rm per}$.
The dependence of wave speed $c$ versus wave steepness $s$ has an extremum (either maximum or minimum). We replace Assumption \ref{ass-kernel} with the opposite non-generic assumption for the fold bifurcation.

\begin{assumption}
\label{ass-fourth-bif}
For the given value of $c \in (1,c_*)$, there exists an eigenfunction $v_0 \in H^1_{\rm per}$ (defined up to the scalar multiplication), such that 
$\mathcal{L} v_0 = 0$, $v_0$ is even in $u$, and $\langle 1, v_0 \rangle \neq 0$.
\end{assumption}

Due to Assumption \ref{ass-fourth-bif}, we have 
${\rm Ker}(\mathcal{L}) = {\rm span}(\eta',v_0)$ so that the decompositions (\ref{kernel}) and (\ref{kernel-generalized}) are no longer true. 
Assumption \ref{ass-fourth-bif} implies that the spectral problem (\ref{mod-Bab-eq-limit}) for $\lambda = 0$ admits a three-dimensional 
kernel spanned as 
\begin{equation}
\label{kernel-extended}
\left( \begin{array}{c} \hat{v}_0 \\ \hat{w}_0 \end{array} \right) = a_1 \left( \begin{array}{c} \eta' \\ 0 \end{array} \right) + 
a_2 \left( \begin{array}{c} v_0 \\ 0 \end{array} \right) + a_3 \left( \begin{array}{c} 0 \\ 1 \end{array} \right).
\end{equation}
where $(a_1,a_2,a_3) \in \mathbb{R}^3$. 

To construct the generalized eigenfunctions and to determine the algebraic multiplicity of the eigenvalue $\lambda = 0$ in the spectral problem (\ref{mod-Bab-eq-limit}), we use the following lemma. 

\begin{lemma}
	\label{lem-tech}
	Let Assumption \ref{ass-fourth-bif} be true. Then, we have 
	$$
	\langle \eta, v_0 \rangle \neq 0 \quad \mbox{\rm and} \quad 
	\langle \mathcal{K} \eta, v_0 \rangle \neq 0.
	$$
\end{lemma}

\begin{proof}
We rewrite $\mathcal{L} v_0 = 0$ for $v_0 \in H^1_{\rm per}$ as 
	\begin{equation}
	\label{eigenfunction-L}
	c^2 \mathcal{K} v_0 - v_0 = (\mathcal{K} \eta) v_0 + \eta \mathcal{K} v_0 + \mathcal{K} (\eta v_0).
	\end{equation}
	Taking the mean value of (\ref{eigenfunction-L}) gives the constraint 
	\begin{equation}
	\label{v_0-constraint-1}
	\langle 1 + 2 \mathcal{K} \eta, v_0 \rangle = 0.
	\end{equation}
Since $\langle 1, v_0 \rangle \neq 0$ by Assumption \ref{ass-fourth-bif}, 
we have $2 \langle \mathcal{K} \eta, v_0 \rangle = -\langle 1, v_0 \rangle \neq 0$. Taking the inner product of (\ref{eigenfunction-L}) with $\eta$ and using (\ref{trav-Bab-eq}) gives another constraint 
	\begin{equation}
	\label{v_0-constraint-2}
	\langle c^2 \mathcal{K} \eta - \eta, v_0 \rangle = 0.
	\end{equation}
Since $\langle \mathcal{K} \eta, v_0 \rangle \neq 0$, we have 
$\langle \eta, v_0 \rangle = c^2 \langle \mathcal{K} \eta, v_0 \rangle \neq 0$.
\end{proof}

\begin{remark}
	Constraints (\ref{v_0-constraint-1}) and (\ref{v_0-constraint-2}) in Lemma \ref{lem-tech} also follow from identities (\ref{L-on-1}) and (\ref{L-on-eta}) 
	and the self-adjointness of $\mathcal{L}$ in $L^2_{\rm per}$. 
\end{remark}

Due to the constraints in Lemma \ref{lem-tech}, there exists only one linearly independent eigenfunction satisfying the linear inhomogeneous system (\ref{gen-Bab-eq}). It corresponds to $a_1 = a_2 = 0$, with the explicit 
solution replacing (\ref{kernel-generalized}) with 
\begin{equation}
\label{kernel-generalized-extended}
\left( \begin{array}{c} \hat{v}_g \\ \hat{w}_g \end{array} \right) =  a_3 \left[ \left( \begin{array}{c} -1 \\ 0 \end{array} \right) 
\pm i \left( \begin{array}{c} \eta' \\ -c \end{array} \right) \right].
\end{equation}
Therefore, the generalized kernel of the spectral problem (\ref{mod-Bab-eq-limit}) is exactly four-dimensional under Assumption \ref{ass-fourth-bif}. This implies particularly that the fold bifurcation does not imply the zero-eigenvalue bifurcation in the spectral problem (\ref{spec-Bab-eq}) for co-periodic perturbations, compared to the bifurcation defined by Assumption \ref{ass-third-bif}, see Remark \ref{rem-stab} and \cite{DP2025}.

We now formulate and prove the criterion for the corresponding bifurcation. 

\begin{theorem}
	\label{theorem-bif-4}
	Let Assumption \ref{ass-fourth-bif} be true. There exists two spectral bands of the modulational stability problem (\ref{mod-Bab-eq}) crossing $\lambda = 0$ according to the roots of the characteristic equation given by 
    \begin{equation}
        \label{char-eq-fourth}
        \left( \lambda - \frac{i \mu c}{2} \right)^2 = |\mu|^3 \left[ \mathcal{N}(c) - \frac{5c}{4} \mathcal{P}(c) \right] + \mathcal{O}(|\mu|^{7/2}),
    \end{equation}
    where $\mathcal{P}(c)$ and $\mathcal{N}(c)$ are defined by (\ref{wave-momentum}) and (\ref{wave-norm}) and $\mathcal{N}(c) - \frac{5c}{4} \mathcal{P}(c) \neq 0$ is assumed. In addition, 
    similar to Theorems \ref{theorem-bif-1} and \ref{theorem-bif-3}, there exist 
    two other spectral bands, which are subsets of $i \R$ and non-degenerate.
\end{theorem}

\begin{proof}
		We revisit the proof of Theorem \ref{theorem-bif-1}, when Assumption \ref{ass-fourth-bif} replaces Assumption \ref{ass-kernel}. Substituting (\ref{kernel-extended}) into (\ref{mod-Bab-eq}) yields the residual terms 
$$
\left\{ \begin{array}{l} 
i \mu \mathcal{H}_{\pm} \hat{w}_0 + \lambda \mathcal{M}_{\pm} \hat{v}_0, \\
i \mu \mathcal{L}_{\pm} \hat{v}_0 + \lambda (\mathcal{M}_{\pm}^* \hat{w}_0 - 2 c \mathcal{H}_{\pm} \hat{v}_0).
\end{array} \right.
$$
Projecting the first equation to ${\rm Ker}(\mathcal{K}) = {\rm span}(1)$ in the first equation  and the second equation to ${\rm Ker}(\mathcal{L}) = {\rm span}(\eta',v_0)$  yields 
\begin{align}
\label{system-tech-extended}
\left\{ \begin{array}{l}
i \mu \langle 1, \mathcal{H}_{\pm} \hat{w}_0 \rangle + \lambda \langle 1, \mathcal{M}_{\pm} \hat{v}_0 \rangle,  \\
i\mu \langle \eta', \mathcal{L}_{\pm} \hat{v}_0 \rangle 
+ \lambda \langle \eta', (\mathcal{M}_{\pm}^* \hat{w}_0 - 2 c \mathcal{H}_{\pm} \hat{v}_0) \rangle, \\
i\mu \langle v_0, \mathcal{L}_{\pm} \hat{v}_0 \rangle 
+ \lambda \langle v_0, (\mathcal{M}_{\pm}^* \hat{w}_0 - 2 c \mathcal{H}_{\pm} \hat{v}_0) \rangle.
\end{array} \right.
\end{align}		
We substitute (\ref{kernel-extended}) into (\ref{system-tech-extended}) and obtain 
\begin{equation}
\label{matrix-projections-extended}
\left( \begin{matrix} 
0 & 0 & \mp \mu \\
0 & i \mu \langle \eta, v_0 \rangle - 2 c \lambda \langle \mathcal{K} \eta, v_0 \rangle & 0 \\
- i \mu \langle \eta, v_0 \rangle + 2 c \lambda \langle \mathcal{K} \eta, v_0 \rangle & 
\mp (2 i c \lambda + c^2 \mu) |\langle 1, v_0 \rangle |^2 \pm 2 \mu \langle \eta, v_0 \rangle \langle 1, v_0 \rangle &  0 
\end{matrix} \right) \left( \begin{matrix} a_1 \\ a_2 \\ a_3 \end{matrix} \right), 
\end{equation}
where we have used (\ref{H-mu}), (\ref{M-mu}), (\ref{L-plus-minus}), (\ref{L-on-eta-prime}), $\mathcal{L}_{\pm}^* = -\mathcal{L}_{\pm}$, and $\mathcal{M}_{\pm} \eta' = \eta'$, and $\langle v_0, \mathcal{M}^* 1 \rangle = \langle v_0, 1 + 2 \mathcal{K} \eta \rangle = 0$ due to (\ref{v_0-constraint-1}). 

Similarly to the system (\ref{matrix-projections}) in the proof of Theorem \ref{theorem-bif-1}, 
the system (\ref{matrix-projections-extended}) suggests two different scalings for the characteristic equations, 
one is related to $(a_1,a_2) = 0$, $a_3 \neq 0$ and the other one is related to $(a_1,a_2) \neq (0,0)$ and $a_3 = 0$. 
The subspace $(a_1,a_2) = 0$, $a_3 \neq 0$ is analyzed in exactly the same way as in the proof of Theorem \ref{theorem-bif-1} 
and it yields two spectral bands with the scaling $\lambda = \mathcal{O}(|\mu|^{1/2})$, which are subsets of $i \R$ and non-degenerate.
On the other hand,  the subspace $(a_1,a_2) \neq (0,0)$ and $a_3 = 0$ yields two spectral bands, which are degenerate at the leading order. It follows from (\ref{matrix-projections-extended}) that the two degenerate spectral bands are given by the leading-order approximation 
\begin{equation}
\label{double-eigenvalue}
\lambda \approx  \frac{i \mu \langle \eta, v_0 \rangle}{2c \langle \mathcal{K} \eta, v_0 \rangle} = \frac{i c \mu}{2}, 
\end{equation}
where we have used  (\ref{v_0-constraint-2}). If $\mu \neq 0$, the third equation in (\ref{matrix-projections-extended}) determines uniquely $a_2$ and the first equation in (\ref{matrix-projections-extended}) determines uniquely $a_3$ since $\pm 2 \mu \langle \eta, v_0 \rangle \langle 1, v_0 \rangle \neq 0$.
Hence, the characteristic equation for the two degenerate spectral bands is obtained from the second equation, where we can set $a_1 = 1$ without loss of generality. 
Based on the justification analysis in Section \ref{sec-case-4}, 
leading to the eigenvalue splitting formulas (\ref{eigenvalue-split-6}) and 
(\ref{eigenvalue-split-7}),  
we introduce the asymptotic expansions:
\begin{equation}
    \label{expansion-mu-4}
	\mu = \epsilon \mu_1, \quad \lambda = \epsilon \lambda_1 + \epsilon^{3/2} \lambda_{3/2} + \mathcal{O}(\epsilon^2)
\end{equation}
and 
	\begin{equation}
	\label{expansion-4}
	\left( \begin{array}{c} \hat{v} \\ \hat{w} \end{array} \right) = \left( \begin{array}{c} \eta' \\ 0 \end{array} \right) + \epsilon^{1/2} a \left( \begin{array}{c} v_0 \\ 0 \end{array} \right) + \epsilon \left( \begin{array}{c} \hat{v}_1 \\ \hat{w}_1 \end{array} \right) + \epsilon^{3/2} \left( \begin{array}{c} \hat{v}_{3/2} \\ \hat{w}_{3/2} \end{array} \right) + \epsilon^2 \left( \begin{array}{c} \hat{v}_2 \\ \hat{w}_2 \end{array} \right) + \mathcal{O}(\epsilon^{5/2}),
	\end{equation}
where $\epsilon > 0$ is a small parameter, $\mu_1 \neq 0$ is fixed, $a \in \mathbb{R}$ is arbitrary at the leading order, and the correction terms up to the order of $\mathcal{O}(\epsilon^2)$ are to be computed recursively in function space $H^1_{\rm per} \times H^1_{\rm per}$. 
		
\underline{\bf Order $\mathcal{O}(\epsilon)$:} We get the linear inhomogeneous system 
\begin{equation}
\left\{ \begin{array}{ll}
\mathcal{K} \hat{w}_1 &=  \lambda_1 \mathcal{M} \eta', \\
\mathcal{L} \hat{v}_1 &= -  2 c \lambda_1 \mathcal{H} \eta' + i \mu_1 \mathcal{L}_{\pm} \eta'.
\end{array}
\right. 
\label{level-two-last}
\end{equation}
Since $\mathcal{M} \eta' = \eta'$, $-\mathcal{H} \eta' = \mathcal{K} \eta$, and $\mathcal{L}_{\pm} \eta' = -\eta$, Fredholm conditions are satisfied for this system if and only if 
\begin{equation}
\label{double-eigenvalue-again}
\lambda_1 = \frac{i c \mu_1}{2},
\end{equation}
in agreement with (\ref{double-eigenvalue}). There exists a periodic solution of the system (\ref{level-two-last}) in the form	
	\begin{align}
	\left( \begin{array}{c} \hat{v}_1 \\ \hat{w}_1 \end{array} \right) =   \left( \begin{array}{c} -i \mu_1 \eta \\ \lambda_1 \mathcal{H} \eta \end{array} \right) + a_1 \left( \begin{array}{c} v_0 \\ 0 \end{array} \right) + b_1 \left( \begin{array}{c} 0 \\ 1 \end{array} \right), \label{fourthlevel-two}
	\end{align}
where we have introduced projections to the second and third elements of the kernel (\ref{kernel-extended}) 
with coordinates $a_1, b_1 \in \mathbb{R}$ and a solution $\eta = -\mathcal{L}^{-1} (c^2 \mathcal{K} \eta - \eta) \in H^1_{\rm per}$ from (\ref{L-on-eta}).
	
\underline{\bf Order $\mathcal{O}(\epsilon^{3/2})$:} We get the linear inhomogeneous system 
\begin{equation}
\left\{ \begin{array}{ll}
\mathcal{K} \hat{w}_{3/2} &=  \lambda_1 a \mathcal{M}_{\pm} v_0 + \lambda_{3/2} \mathcal{M} \eta', \\
\mathcal{L} \hat{v}_{3/2} &= - 2 c \lambda_1 a \mathcal{H}_{\pm} v_0 + i \mu_1 a \mathcal{L}_{\pm} v_0 - 2 c \lambda_{3/2} \mathcal{H} \eta'.
\end{array}
\right. 
\label{level-three-last}
\end{equation}
Since $\mathcal{M}_{\pm} v_0, \mathcal{M} \eta' \perp {\rm Ker}(\mathcal{K})$ due to the constraint (\ref{v_0-constraint-1}), there exists always a solution $\hat{w}_{3/2} \in H^1_{\rm per}$ to the first equation of the system (\ref{level-three-last}). Projection of the right-hand side of the second equation of the system  (\ref{level-three-last}) to $\eta'$ yields 
$$
\langle \eta', - 2 c \lambda_1 a \mathcal{H}_{\pm} v_0 + i \mu_1 a \mathcal{L}_{\pm} v_0 - 2 c \lambda_{3/2} \mathcal{H} \eta'\rangle 
= -2c \lambda_1 a \langle \mathcal{K} \eta, v_0 \rangle + i \mu_1 a \langle \eta, v_0 \rangle = 0, 
$$
where the last equality is due to the constraint (\ref{v_0-constraint-2}) 
and the definition of $\lambda_1$ in (\ref{double-eigenvalue-again}). Finally, projection of the right-hand side of the second equation 
of the system  (\ref{level-three-last}) to $v_0$ yields 
\begin{align*}
& \quad   
\langle v_0, - 2 c \lambda_1 a \mathcal{H}_{\pm} v_0 + i \mu_1 a \mathcal{L}_{\pm} v_0 - 2 c \lambda_{3/2} \mathcal{H} \eta'\rangle \\
&= \mp 2 i c \lambda_1 a |\langle 1, v_0 \rangle |^2 \mp \mu_1 a \left[ c^2 |\langle 1, v_0 \rangle |^2  - 2 \langle \eta, v_0 \rangle \langle 1, v_0 \rangle \right] + 2 c \lambda_{3/2} \langle \mathcal{K} \eta, v_0 \rangle = 0.
\end{align*}
In view of (\ref{double-eigenvalue-again}), we obtain a unique expression for the coefficient $a$ given by 
\begin{equation}
    \label{a-last}
a = \mp \frac{c \lambda_{3/2} \langle \mathcal{K} \eta, v_0 \rangle}{\mu_1 \langle \eta, v_0 \rangle \langle 1, v_0 \rangle } = \mp \frac{\lambda_{3/2}}{c \mu_1 \langle 1, v_0 \rangle },
\end{equation}
where we have used (\ref{v_0-constraint-2}) again. Since the three constraints are satisfied, there exists a periodic solution of the system (\ref{level-three-last}) in the form	
	\begin{align}
	\left( \begin{array}{c} \hat{v}_{3/2} \\ \hat{w}_{3/2} \end{array} \right) =  \lambda_{3/2} \left( \begin{array}{c} \tilde{v}_{3/2} \\ \tilde{w}_{3/2} \end{array} \right) + a_{3/2} \left( \begin{array}{c} v_0 \\ 0 \end{array} \right) + b_{3/2} \left( \begin{array}{c} 0 \\ 1 \end{array} \right), \label{fourthlevel-three}
	\end{align}
where $(\tilde{v}_{3/2},\tilde{w}_{3/2})$ are uniquely defined from (\ref{level-three-last}) with $a$ being defined by  (\ref{a-last}), whereas $a_{3/2}, b_{3/2} \in \R$ are renormalizations of the projections to the second and third element of the kernel (\ref{kernel-extended}). 

\underline{\bf Order $\mathcal{O}(\epsilon^2)$:} We get the linear inhomogeneous system 
\begin{equation}
\left\{ \begin{array}{ll}
\mathcal{K} \hat{w}_2 &=  \lambda_1 \mathcal{M}_{\pm} \hat{v}_1 + \lambda_{3/2} a \mathcal{M}_{\pm} v_0 + i \mu_1 \mathcal{H}_{\pm} \hat{w}_1, \\
\mathcal{L} \hat{v}_2 &= \lambda_1 \left( \mathcal{M}_{\pm}^* \hat{w}_1 - 2 c \mathcal{H}_{\pm} \hat{v}_1 \right) 
+ i \mu_1 \mathcal{L}_{\pm} \hat{v}_1 - 2 c \lambda_{3/2} a \mathcal{H}_{\pm} v_0.
\end{array}
\right. 
\label{level-four-last}
\end{equation}
Projection of the right-hand side of the first equation 
of the system  (\ref{level-four-last}) to $1$ yields 
\begin{align*}
 \langle 1, \lambda_1\mathcal{M}_{\pm} \hat{v}_1 + \lambda_{3/2} a \mathcal{M}_{\pm} v_0 + i \mu_1 \mathcal{H}_{\pm} \hat{w}_1 \rangle = -i \lambda_1 \mu_1 \langle 1, \mathcal{M}_{\pm} \eta \rangle    \mp \mu_1 b_1 = 0.
\end{align*}
This defines the coefficient $b_1$ uniquely by 
\begin{equation}
    \label{b-last}
    b_1 = \mp i \lambda_1 \langle 1 + 2 \mathcal{K} \eta, \eta \rangle
\end{equation}
since $\mu_1 \neq 0$ is assumed. Projection of the right-hand side of the second equation 
of the system  (\ref{level-four-last}) to $\eta'$ yields 
\begin{align*}
& \quad 
\langle \eta', \lambda_1 \left( \mathcal{M}_{\pm}^* \hat{w}_1 - 2 c \mathcal{H}_{\pm} \hat{v}_1 \right) 
+ i \mu_1 \mathcal{L}_{\pm} \hat{v}_1 - 2 c \lambda_{3/2} a \mathcal{H}_{\pm} v_0 \rangle  \\
&= \lambda_1^2 \langle \mathcal{K} \eta, \eta \rangle + 2 i c \mu_1 \lambda_1 \langle \mathcal{K} \eta, \eta \rangle + \mu_1^2 
\langle \eta, \eta \rangle + a_1 (-2 c \lambda_1 \langle \mathcal{K} \eta, v_0 \rangle + i \mu_1 \langle \eta, v_0 \rangle) 
- 2 c a \lambda_{3/2} \langle \mathcal{K} \eta, v_0 \rangle \\
&= - \mu_1^2  \left[  \langle (c^2 \mathcal{K} \eta - \eta, \eta \rangle + \frac{c^2}{4} \langle \mathcal{K} \eta, \eta \rangle \right]
- 2 c a \lambda_{3/2} \langle \mathcal{K} \eta, v_0 \rangle = 0,
\end{align*}
where the last equality is due to the constraint (\ref{v_0-constraint-2}) and the definition of $\lambda_1$ in (\ref{double-eigenvalue-again}). Since $a$ is defined by 
(\ref{a-last}) and the relation (\ref{v_0-constraint-1}) holds, this yields 
\begin{equation}
    \label{lambda-3}
\lambda_{3/2}^2 = \pm \mu_1^3 \left[ \mathcal{N}(c) - \frac{5c}{4} \mathcal{P}(c) \right], \quad \pm \mu_1 > 0,
\end{equation}
which is equivalent to the characteristic equation (\ref{char-eq-fourth}). 
Finally, projection of the right-hand side of the second equation of the system  (\ref{level-four-last}) to $v_0$ yields 
\begin{align*}
& \quad 
\langle v_0, \lambda_1 \left( \mathcal{M}_{\pm}^* \hat{w}_1 - 2 c \mathcal{H}_{\pm} \hat{v}_1 \right) 
+ i \mu_1 \mathcal{L}_{\pm} \hat{v}_1 - 2 c \lambda_{3/2} a \mathcal{H}_{\pm} v_0 \rangle  \\
&= \lambda_1^2 \langle v_0, \mathcal{M}_{\pm}^* \mathcal{H} \eta \rangle + 2 i c \mu_1 \lambda_1 \langle v_0, \mathcal{H}_{\pm} \eta \rangle + \mu_1^2 
\langle v_0, \mathcal{L}_{\pm} \eta \rangle \\
& \qquad 
\mp 2 i c \lambda_1 a_1 |\langle 1, v_0 \rangle |^2  \mp \mu_1 a_1 \left[ c^2 |\langle 1, v_0 \rangle |^2  - 2 \langle \eta, v_0 \rangle \langle 1, v_0 \rangle \right] \mp 2 i c \lambda_{3/2} a |\langle 1, v_0 \rangle |^2 = 0.
\end{align*}
In view of (\ref{double-eigenvalue-again}), (\ref{a-last}), and (\ref{lambda-3}), 
this equation yields a unique expression of $a_1$ in terms of $\mu_1$, which we do not write. 
Since the three constraints are satisfied, there exists a periodic solution of the system (\ref{level-four-last}). The asymptotic expansion (\ref{expansion-4}) 
is defined by (\ref{fourthlevel-two}) and (\ref{fourthlevel-three}). The two roots 
of the characteristic equation (\ref{char-eq-fourth}) are non-degenerate if 
$\mathcal{N}(c) - \frac{5c}{4} \mathcal{P}(c) \neq 0$.
\end{proof}

\begin{remark}
If $\mathcal{N}(c) - \frac{5c}{4} \mathcal{P}(c) < 0$, then the two non-degenerate spectral bands given by roots $\lambda^{\pm}(\mu)$ of the characteristic equation (\ref{char-eq-fourth}) satisfy $\lambda^{\pm}(\mu) \in i \mathbb{R}$ 
and $\lambda^+(\mu) \neq \lambda^-(\mu)$ for $\mu \neq 0$. In this case, there exists no figure-$8$ bands at the fold bifurcation point.
    \label{rem-comp}
\end{remark}

\begin{remark}
We can confirm $\mathcal{N}(c) - \frac{5c}{4} \mathcal{P}(c) < 0$ for the small-amplitude expansions (\ref{eta-small}) and (\ref{c-small}), for which 
    \begin{align*}
        \mathcal{N}(c) &= \frac{1}{2} a^2 + \mathcal{O}(a^4), \\
                \mathcal{P}(c) &= \frac{1}{2} a^2 + \mathcal{O}(a^4),
    \end{align*}
    which yields 
    $$
   \mathcal{N}(c) - \frac{5c}{4} \mathcal{P}(c) = -\frac{1}{8}  a^2 + \mathcal{O}(a^2). 
    $$
However, Assumption \ref{ass-fourth-bif} is not true in the small-amplitude limit because $\eta'(u) \sim \sin(u)$, $v_0(u) \sim \cos(u)$, and 
$\langle 1, v_0 \rangle = 0$. New figure-$8$ bands bifurcating in the small-amplitude limit are studied in Appendix \ref{app-A}, where it is shown 
that the normal form for the bifurcation is different from (\ref{char-eq-fourth}), see equation (\ref{fig-8-small}) also derived in \cite{berti2022full,creedon2022ahigh}.
    \label{rem-small-comp}
\end{remark}

\subsection{Normal form of the bifurcation}
\label{sec-6-2}

In order to obtain the figure-$8$ bands as a result of the bifurcation of Theorem \ref{theorem-bif-4}, 
we need to understand how the bifurcation in the family of traveling waves is unfolded for $c \neq c_0$. 
The following lemma gives the bifurcation result based on the standard method of Lyapunov--Schmidt reductions.

\begin{lemma}
\label{lem-fold}
    Let $c_0 \in (1,c_*)$ be the bifurcation point, under which Assumption \ref{ass-fourth-bif} holds with the profile $\eta_0 \in C^{\infty}_{\rm per}$. If $\langle \mathcal{K} v_0, v_0^2 \rangle \neq 0$, then 
    there exist two solutions of Babenko' equation (\ref{trav-Bab-eq}) for $c \neq c_0$ close to $c_0$ with 
\begin{equation}
    \label{LS-sign}
    {\rm sgn}(c-c_0) = {\rm sgn}\left( \langle \mathcal{K} v_0, v_0^2 \rangle \langle \mathcal{K} \eta_0, v_0 \rangle \right)
\end{equation}
    such that their profiles satisfy 
\begin{equation}
    \label{LS-bound}
\| \eta - \eta_0 \|_{H^1_{\rm per}} \leq C \sqrt{|c-c_0|},
\end{equation}
for a positive constant $C$ in a small neighborhood of $c = c_0$.
\end{lemma}

\begin{proof}
    Let $c^2 = c_0^2 + \delta \varepsilon^2$, where $\delta \in \R$ is to be determined and $\varepsilon \in \mathbb{R}$ is a small parameter unrelated to $\epsilon$. 
    By using the Lyapunov--Schmidt decomposition, we represent an even smooth solution of Babenko's equation (\ref{trav-Bab-eq}) in the form:
\begin{equation}
    \label{LS-decomposition}
    \eta = \eta_0 + \varepsilon v_0 + \varepsilon^2 \mathfrak{v}, \quad \langle v_0, \mathfrak{v} \rangle = 0,
\end{equation}
    where $\mathfrak{v} \in H^1_{\rm per,e}(\mathbb{T}) = \{ f \in H^1_{\rm per} : \; f(-x) = f(x), \quad x \in \mathbb{T} \}$. Babenko's equation (\ref{trav-Bab-eq}) is rewritten in the perturbed form 
\begin{equation}
    \label{LS-1}
\mathcal{L}_0 \mathfrak{v} + \delta \mathcal{K} (\eta_0 + \varepsilon v_0 + \varepsilon^2 \mathfrak{v}) = \frac{1}{2} \mathcal{K} (v_0 + \varepsilon \mathfrak{v})^2 + (v_0 + \varepsilon \mathfrak{v}) \mathcal{K} (v_0 + \varepsilon \mathfrak{v}),
\end{equation}
where $\mathcal{L}_0$ is the linearized Babenko's operator for $c = c_0$ and $\eta = \eta_0$. Since ${\rm Ker}(\mathcal{L}_0) = {\rm span}(v_0)$ 
in $H^1_{\rm per,e}(\mathbb{T})$, we define $\delta$ from the bifurcation equation after projecting (\ref{LS-1}) to $v_0$:
\begin{equation}
    \label{LS-2}
\delta \langle \mathcal{K} (\eta_0 + \varepsilon v_0 + \varepsilon^2 \mathfrak{v}), v_0 \rangle = \frac{1}{2} 
\langle \mathcal{K} (v_0 + \varepsilon \mathfrak{v})^2, v_0 \rangle + \langle \mathcal{K} (v_0 + \varepsilon \mathfrak{v}), v_0 (v_0 + \varepsilon \mathfrak{v})  \rangle.
\end{equation}
By the implicit function theorem in $H^1_{\rm per,e}(\mathbb{T}) |_{\{ v_0 \}^{\perp}}$, there exists a unique 
solution of (\ref{LS-1}) under the constraint (\ref{LS-2}) for every small $\varepsilon \in \R$ such that $\mathfrak{v} \in H^1_{\rm per,e}(\mathbb{T}) |_{\{ v_0 \}^{\perp}}$ and 
$$
\| \mathfrak{v} \|_{H^1_{\rm per}} \leq C \varepsilon^2
$$ 
for some $C > 0$ in a small neighborhood of $\varepsilon = 0$. 
Substituting this solution into (\ref{LS-2}), we obtain 
\begin{equation}
    \label{delta}
\delta = \delta_0 + \mathcal{O}(\varepsilon), \qquad \delta_0 = \frac{3}{2} \frac{\langle \mathcal{K} v_0, v_0^2 \rangle}{\langle \mathcal{K} \eta_0, v_0 \rangle}.
\end{equation}
Since $\langle \mathcal{K} \eta_0, v_0 \rangle \neq 0$ by Assumption \ref{ass-fourth-bif} and 
$\langle \mathcal{K} v_0, v_0^2 \rangle \neq 0$ by the assumption of the lemma, 
we have $\delta_0 \neq 0$, which implies that $\delta \neq 0$ for every small $\varepsilon$. The bound (\ref{LS-bound}) follows from the decomposition (\ref{LS-decomposition}) 
and the bound on $\mathfrak{v} \in H^1_{\rm per,e}(\mathbb{T}) |_{\{ v_0 \}^{\perp}}$. 
Two possible solutions for $\varepsilon \neq 0$ exist for 
${\rm sgn}(c-c_0) = {\rm sgn}(\delta_0)$, which justifies the condition 
(\ref{LS-sign}).
\end{proof}

\begin{remark}
    The bifurcation in Lemma \ref{lem-fold} is referred to as the fold bifurcation. No solutions of Babenko' equation (\ref{trav-Bab-eq}) 
    for $c \neq c_0$ close to $c_0$ exist in a local neighborhood of $\eta_0 \in C^{\infty}_{\rm per}$ if the sign of $(c-c_0)$ is opposite to (\ref{LS-sign}).
\end{remark}

\begin{remark}
\label{rem-splitting}
    Expanding $\mathcal{L}$ at the profile $\eta$ in Lemma \ref{lem-fold} near $\varepsilon = 0$ yields 
\begin{equation}
    \label{expansion-L}
\mathcal{L} = \mathcal{L}_0 + \varepsilon \mathcal{L}_1 + \mathcal{O}(\varepsilon^2),
\end{equation}
where 
\begin{align*}
    \mathcal{L}_0 &= c_0^2 \mathcal{K} - 1 - \mathcal{K} \eta_0 - \eta_0 \mathcal{K} - \mathcal{K}(\eta_0 \cdot),  \\
    \mathcal{L}_1 &= -(\mathcal{K} v_0) - v_0 \mathcal{K} - \mathcal{K}(v_0 \cdot).
\end{align*}
The perturbation theory for the small eigenvalue $\lambda(\varepsilon)$ of $\mathcal{L}$ bifurcating from $\lambda(0) = 0$ at $\varepsilon = 0$ yields 
$$
\lambda(\varepsilon) = -3 \varepsilon \frac{\langle \mathcal{K} v_0, v_0^2 \rangle}{\| v_0 \|^2_{L^2}} + \mathcal{O}(\varepsilon^2).
$$
Since $\langle \mathcal{K} v_0, v_0^2 \rangle \neq 0$ and there are two roots of $\varepsilon$ with opposite signs for every $c \neq c_0$ near $c = c_0$ satisfying (\ref{LS-sign}), 
it is clear that one root corresponds to the small positive eigenvalue $\lambda(\varepsilon) > 0$ and the other root 
corresponds to the small negative eigenvalue $\lambda(\varepsilon) < 0$.
\end{remark}

\begin{remark}
	\label{rem-stab}
    Although $\mathcal{L}$ admits a new small eigenvalue for small $\varepsilon$, see Remark \ref{rem-splitting}, no new small eigenvalues near $\lambda = 0$ bifurcate in the spectral stability problem (\ref{spec-Bab-eq}) for small $\varepsilon$. This is because 
    $\langle \eta', \mathcal{H} v_0 \rangle = \langle \mathcal{K} \eta, v_0 \rangle \neq 0$ due to Assumption \ref{ass-fourth-bif} and Lemma \ref{lem-tech}. 
    The Fredholm condition for the generalized kernel is not satisfied for the subspace $(a_1,a_2)$ in (\ref{kernel-extended}) and no generalized eigenvectors 
    exist in this subspace. There is exactly one generalized eigenvector associated with the subspace $a_3$ as in (\ref{kernel-generalized-extended}). In other words, 
    the quadruple algebraic multiplicity of $\lambda = 0$ in the spectral problem (\ref{spec-Bab-eq}) persists  both for $c = c_0$ and $c \neq c_0$. The only difference between the two cases is the geometric multiplicity of $\lambda = 0$ in the spectral problem (\ref{spec-Bab-eq}): it is equal to three for $c = c_0$ and to two for $c \neq c_0$.
\end{remark}

Finally, we combine the results of Theorem \ref{theorem-bif-4} and Lemma \ref{lem-fold} to get the normal form for the figure-$8$ bifurcation.

\begin{theorem}
	\label{theorem-bif-4-normal-form}
	Let $c = c_0$ be the fold bifurcation point of Lemma \ref{lem-fold}, where Assumption \ref{ass-fourth-bif} holds. The characteristic equation for the two spectral bands crossing $\lambda = 0$ as in Theorem \ref{theorem-bif-4} is given by
    \begin{equation}
        \label{char-eq-fourth-normal-form}
        \left( \lambda - \frac{i \mu c_0}{2} \right)^2 = \mu^2 \left[ \mathcal{N}(c_0) - \frac{5c_0}{4} \mathcal{P}(c_0) \right] 
        \left[ |\mu| + \frac{\delta_0 \varepsilon}{c_0^2 \langle 1, v_0 \rangle} \right] + \mathcal{O}(|\mu|^{7/2}),
    \end{equation}
    where $c^2 = c_0^2 + \delta_0 \varepsilon^2 + \mathcal{O}(\varepsilon^3)$ with $\delta_0$ given by (\ref{delta}) and with $\varepsilon = \mathcal{O}(\mu)$ at the leading-order balance.
\end{theorem}

\begin{proof}
We use (\ref{expansion-mu-4}) and (\ref{expansion-4}) for $(\mu,\lambda) \in \R \times \mathbb{C}$ near $(0,0)$, 
and $(\hat{v},\hat{w}) \in H^1_{\rm per} \times H^1_{\rm per}$. 
Since the profile $\eta$ is also expanded for $c \neq c_0$, we use (\ref{LS-decomposition}) with 
$\varepsilon = \sigma \epsilon$, where $\epsilon > 0$ and $\sigma = \pm 1$ incorporate the different signs 
of $\varepsilon$ at the two branches. By Lemma \ref{lem-fold}, we also have $c^2 = c_0^2 + \delta \epsilon^2$ 
with $\delta = \delta_0 + \mathcal{O}(\varepsilon)$ given by (\ref{delta}). Due to the perturbation expansion (\ref{LS-decomposition}), we 
should use the expansion (\ref{expansion-L}) for $\mathcal{L}$ and similar expansions for $\mathcal{M}_{\pm}$, 
$\mathcal{M}^*_{\pm}$, and $\mathcal{L}_{\pm}$ (which do not affect the outcomes of computations).

\underline{\bf Order $\mathcal{O}(\epsilon)$:} We get the following modification of the linear inhomogeneous system (\ref{level-two-last}):
\begin{equation}
\left\{ \begin{array}{ll}
\mathcal{K} \hat{w}_1 &=  \lambda_1 \mathcal{M} \eta'_0, \\
\mathcal{L}_0 \hat{v}_1 + \sigma \mathcal{L}_1 \eta'_0 &= -  2 c_0 \lambda_1 \mathcal{H} \eta'_0 + i \mu_1 \mathcal{L}_{\pm} \eta'_0,
\end{array}
\right. 
\label{level-two-last-P}
\end{equation}
where $\sigma = \pm 1$ is due to $\varepsilon = \pm \epsilon$. Since 
$$
\mathcal{L}_0 v_0' + \mathcal{L}_1 \eta_0' = 0
$$
follows from the first order in $\mathcal{O}(\varepsilon)$ of $\mathcal{L} \eta' = 0$, 
the solution to the system (\ref{level-two-last-P}) is modified compared to (\ref{fourthlevel-two}) as follows:
	\begin{align*}
	\left( \begin{array}{c} \hat{v}_1 \\ \hat{w}_1 \end{array} \right) =   \left( \begin{array}{c} -i \mu_1 \eta_0 \\ \lambda_1 \mathcal{H} \eta_0 \end{array} \right) + a_1 \left( \begin{array}{c} v_0 \\ 0 \end{array} \right) + b_1 \left( \begin{array}{c} 0 \\ 1 \end{array} \right) + \sigma 
    \left( \begin{array}{c} v_0' \\ 0 \end{array} \right).
	\end{align*}
In addition, we recall that $\lambda_1 = \frac{i c_0 \mu_1}{2}$ from (\ref{double-eigenvalue-again}).
	
\underline{\bf Order  $\mathcal{O}(\epsilon^{3/2})$:} We get  the following modification of the linear inhomogeneous system (\ref{level-three-last}):
\begin{equation}
\left\{ \begin{array}{ll}
\mathcal{K} \hat{w}_{3/2} &=  \lambda_1 a \mathcal{M}_{\pm} v_0 + \lambda_{3/2} \mathcal{M} \eta'_0, \\
\mathcal{L} \hat{v}_{3/2} + a \sigma \mathcal{L}_1 v_0 &= - 2 c_0 \lambda_1 a \mathcal{H}_{\pm} v_0 + i \mu_1 a \mathcal{L}_{\pm} v_0 - 2 c_0 \lambda_{3/2} \mathcal{H} \eta'_0.
\end{array}
\right. 
\label{level-three-last-P}
\end{equation}
The first equation of the system (\ref{level-three-last-P}) is the same as in (\ref{level-three-last}) with the same solution. 
Since $\mathcal{L}_1 v_0$ is even, there is no contribution of the new term of the second equation of the system  (\ref{level-three-last}) 
to the orthogonality condition with respect to $\eta'$. However, we get a modification of the 
orthogonality condition with respect to $v_0$ as follows: 
\begin{align*}
& \quad   
\langle v_0, - 2 c_0 \lambda_1 a \mathcal{H}_{\pm} v_0 + i \mu_1 a \mathcal{L}_{\pm} v_0 - 2 c_0 \lambda_{3/2} \mathcal{H} \eta_0' - a \sigma  \mathcal{L}_1 v_0 \rangle \\
&= \mp 2 i c_0 \lambda_1 a |\langle 1, v_0 \rangle |^2 \mp \mu_1 a \left[ c_0^2 |\langle 1, v_0 \rangle |^2  - 2 \langle \eta_0, v_0 \rangle \langle 1, v_0 \rangle \right] + 2 c_0 \lambda_{3/2} \langle \mathcal{K} \eta_0, v_0 \rangle + 3 a \sigma \langle \mathcal{K} v_0, v_0^2 \rangle 
= 0.
\end{align*}
Similarly to (\ref{a-last}), we use $\lambda_1 = \frac{i c_0 \mu_1}{2}$ and obtain a unique expression for the coefficient $a$ given by 
\begin{equation}
    \label{a-last-P}
a = -\frac{c_0 \lambda_{3/2} \langle \mathcal{K} \eta_0, v_0 \rangle}{|\mu_1| \langle \eta_0, v_0 \rangle \langle 1, v_0 \rangle + \frac{3}{2} \sigma \langle \mathcal{K} v_0, v_0^2 \rangle}
= -\frac{\lambda_{3/2}}{c_0 |\mu_1| \langle 1, v_0 \rangle + \sigma \delta_0},
\end{equation}
where we have used (\ref{delta}). Since the three constraints are satisfied, 
there exists a periodic solution of the system (\ref{level-three-last-P}), which does not affect 
computations at the next order. 

\underline{\bf Order  $\mathcal{O}(\epsilon^2)$:} We get the following modification of the linear inhomogeneous system (\ref{level-four-last}):
\begin{equation*}
\left\{ \begin{array}{ll}
\mathcal{K} \hat{w}_2 &=  \lambda_1 \mathcal{M}_{\pm} \hat{v}_1 + \lambda_{3/2} a \mathcal{M}_{\pm} v_0 + i \mu_1 \mathcal{H}_{\pm} \hat{w}_1 + \lambda_1 \sigma \mathcal{M}_{\pm}^{(1)} \eta_0', \\
\mathcal{L} \hat{v}_2 + \mathcal{L}_1 \hat{v}_1 + \mathcal{L}_2 \eta_0' &= \lambda_1 \left( \mathcal{M}_{\pm}^* \hat{w}_1 - 2 c_0 \mathcal{H}_{\pm} \hat{v}_1 \right) 
+ i \mu_1 \mathcal{L}_{\pm} \hat{v}_1 - 2 c_0 \lambda_{3/2} a \mathcal{H}_{\pm} v_0 + i \mu_1 \sigma \mathcal{L}_{\pm}^{(1)} \eta_0',
\end{array}
\right. 
\end{equation*}
where $\mathcal{L}_{\pm}^{(1)}$ and $\mathcal{M}_{\pm}^{(1)}$ appear in the expansion of $\mathcal{L}_{\pm}$ and $\mathcal{M}_{\pm}$ at the first order of  $\mathcal{O}(\varepsilon)$. Since $\mathcal{M}_{\pm}^{(1)} \eta_0'$ is odd, it does not contribute to the solvability condition for the first 
equation of the system so that the coefficient $b_1$ is still uniquely defined by the expression (\ref{b-last}). 

The second equation of the system is solved under two orthogonality conditions with respect to $\eta_0'$ and $v_0$. 
Recall that the orthogonality condition to $v_0$ defines uniquely the parameter $a_1$, hence it will not be written. 
Finally, the orthogonality condition to $\eta_0'$ is not affected by the term $\mathcal{L}_{\pm}^{(1)} \eta_0'$ since it is even, 
whereas $\eta_0'$ is odd. Similarly, it follows from the second order in $\mathcal{O}(\varepsilon^2)$ of $\mathcal{L} \eta' = 0$ 
that 
$$
\mathcal{L}_0 \mathfrak{v}_0' + \mathcal{L}_1 v_0' + \mathcal{L}_2 \eta_0' = 0,
$$
where $\mathfrak{v}_0$ is the leading order in the expansion (\ref{LS-decomposition}). This implies that 
the odd part of $\mathcal{L}_1 \hat{v}_1 + \mathcal{L}_2 \eta_0'$ does not contribute to the ortogonality to $\eta_0'$. 
The even part of $\mathcal{L}_1 \hat{v}_1 + \mathcal{L}_2 \eta_0'$ does not contribute either, due to 
the different parity with the parity of $\eta_0'$. As a result, we obtain the same orthogonality condition 
as in the proof of Theorem \ref{theorem-bif-4}:
\begin{align*}
- \mu_1^2  \left[  \frac{5 c_0}{4} \mathcal{P}(c_0) - \mathcal{N}(c_0) \right]
- 2 c_0 a \lambda_{3/2} \langle \mathcal{K} \eta_0, v_0 \rangle = 0.
\end{align*}
Substituting (\ref{a-last-P}) yields 
\begin{equation*}
\lambda_{3/2}^2 = \mu_1^2 \left[ \mathcal{N}(c_0) - \frac{5c_0}{4} \mathcal{P}(c_0) \right] \left[  |\mu_1| + \frac{\sigma \delta_0}{c_0 \langle 1, v_0 \rangle}
\right],
\end{equation*}
which is equivalent to the characteristic equation (\ref{char-eq-fourth}) due to the expansion (\ref{expansion-mu-4}) and $c^2 = c_0^2 + \varepsilon^2 \delta_0 + \mathcal{O}(\varepsilon^3)$ with $\varepsilon = \sigma \epsilon$.
\end{proof}

\begin{remark}
	\label{remark-split}
New figure-$8$ bands appear past the bifurcation point $c = c_0$ if
 $$
 \mathcal{N}(c_0) - \frac{5c_0}{4} \mathcal{P}(c_0) < 0
 $$   
see Remark \ref{rem-comp}. If $c = c_0$, both spectral bands in the characteristic equation (\ref{char-eq-fourth-normal-form}) satisfy
$\lambda^{\pm}(\mu) \in i \mathbb{R}$ and $\lambda^+(\mu) \neq \lambda^-(\mu)$ for $\mu \neq 0$. Since the sign of $\varepsilon$ in Lemma \ref{lem-fold} is opposite between the two branches near the fold bifurcation, one branch for $c \neq c_0$ still has $\lambda^{\pm}(\mu) \in i \mathbb{R}$ and $\lambda^+(\mu) \neq \lambda^-(\mu)$ for $\mu \neq 0$, whereas the other branch has the new figure-$8$ bands with the maximum corresponding to 
$$
\mu = \frac{|\delta_0 \varepsilon|}{c_0^2 |\langle 1, v_0 \rangle|} + 
 \mathcal{O}(|\varepsilon|^{3/2}).
$$
\end{remark}

\subsection{Numerical results}
\label{sec-6-3}

The criterion for the bifurcation is a double zero eigenvalue of $\mathcal{L}$ in $L^2_{\rm per}$ which coincides with the extremal point of the speed $c$ versus steepness $s$ shown in vertical line of Figure \ref{fig:delta}. We compute 
eigenvalues of $\mathcal{L}$ in $L^2_{\rm per}$ and show them in Figure \ref{fig:beigs} by green lines. At the green circles, there exists an even 
eigenfunction $v_0 \in H^1_{\rm per}$ in addition to the odd eigenfunction 
$\eta' \in H^1_{\rm per}$.  The numerical solution of eigenvalue problem is obtained via shift-and-invert method, see ~\cite{dyachenko2023quasiperiodic}. Once $v_0$ has been determined, we compute the parameter $\varepsilon$ and the average value $\langle 1, v_0\rangle$ to verify that 
it is nonzero according to the Assumption \ref{ass-fourth-bif}. 

The corresponding steepness $s_0$ and speed $c_0$ for the figure-8
bifurcation are recorded in Table \ref{tab:fold}, where we also show
the coefficients $\langle 1, v_0 \rangle$, $\delta_0$, and $\mathcal{N}(c_0) - \frac{5 c_0}{4} \mathcal{P}(c_0)$ of the normal form~\eqref{char-eq-fourth-normal-form}. We note that  $\langle 1, v_0\rangle \neq 0$ and $\mathcal{N}(c_0) - \frac{5 c_0}{4} \mathcal{P}(c_0) < 0$, so that the new figure-8 appears past the fold bifurcation, according to Remark \ref{remark-split}. 

\begin{table}[htb!]
	\centering
	\renewcommand{\arraystretch}{1.3}
	\begin{tabular}{c|c|c|c|c}
        $s_0$ & $c_0$ & $\langle1,v_0\rangle$ & $\delta_0$ & $\mathcal{N}(c_0) - \frac{5c_0}{4}\mathcal{P}(c_0)$ \\ \hline
        $0.13875329$  & $1.09295138$  & $0.34537$ & $-5.130256\times10^2$ & $-0.26372$\\
	\end{tabular}
	\caption{Parameters of the first nontrivial figure-$8$ bifurcation.}
	\label{tab:fold}
\end{table}

Figure \ref{fig:mno1} shows the instability bands of the spectral problem (\ref{mod-Bab-eq}) near the origin in the spectral $\lambda$-plane obtained numerically. The three panels correspond to three consequent values of steepness before, at, and after the critical steepness. The 
red lines display the figure-$8$ spectral bands obtained from the normal form (\ref{char-eq-fourth-normal-form}) with the excellent agreement with the spectral bands shown in blue dots.

\begin{figure}[htb!]
    \centering
    \includegraphics[width=0.95\linewidth]{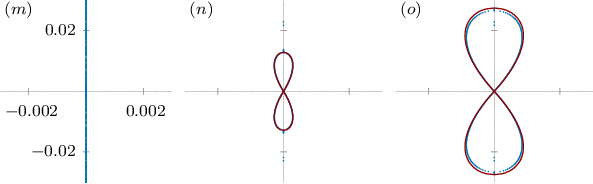}
    \caption{The first instance of the figure-$8$ bifurcation, which 
    	correspond to panels (m)--(o) of Figure \ref{fig-bif} zoomed near the origin. The wave steepness is chosen as $s=0.13875329$ (left panel),  $s = 0.13877279$ (middle panel), and $s = 0.13879503$ (right panel). The red lines show the approximatioins obtained from the normal form ~\eqref{char-eq-fourth-normal-form}. 
 }
    \label{fig:mno1}
\end{figure}

\section*{Acknowledgements} RM and DP acknowledge the support of the Australian Research Council under grant DP210101102. RM also acknowledges J. Bronski for helpful and insightful discussions.

\appendix 
\section{Figure-$8$ in the small-amplitude limit}
\label{app-A}

Here, we recover the figure-$8$ in the small-amplitude limit near $c = 1$, 
which was determined in \cite{creedon2022ahigh} and justified in \cite{berti2022full}. The corresponding result is given by the following theorem. 

\begin{theorem}
	\label{th-appendix}
	Let the Stokes wave be given by the asymptotic expansions  (\ref{eta-small}) and (\ref{c-small}), where $a$ is the small amplitude parameter. There exist  spectral bands of the modulational stability problem (\ref{mod-Bab-eq}) crossing $\lambda = 0$ according to roots of the characteristic equation 
	\begin{equation}
	\left( \lambda - \frac{i \mu}{2} \right)^2 = \frac{1}{64} \mu^2 (8 a^2 - \mu^2) + \mathcal{O}(a^6 + \mu^6).
	\label{fig-8-small}
	\end{equation}
\end{theorem}

\begin{proof}
By using (\ref{eta-small}) and (\ref{c-small}), we write 
\begin{align*}
    \mathcal{L} &= \mathcal{K} - 1 + a \mathcal{L}_1 + a^2 \mathcal{L}_2 + \mathcal{O}(a^3), \\
    \mathcal{L}_{\pm} &= \mathcal{H}_{\pm} + a \mathcal{L}_{\pm}^{(1)} + \mathcal{O}(a^2), \\
    \mathcal{M}_{\pm} &= 1 + a \mathcal{M}_{\pm}^{(1)} + \mathcal{O}(a^2), 
\end{align*}
with 
\begin{align*}
    \mathcal{L}_1 &= -\cos(u) - \cos(u) \mathcal{K}[\cdot] - \mathcal{K} [\cos(u) \cdot ], \\
    \mathcal{L}_2 &= \mathcal{K} - 2 \cos(2u) - \left[\cos(2u) - \frac{1}{2} \right] \mathcal{K}[\cdot] - \mathcal{K}[\cos(2u) \cdot - \frac{1}{2} \cdot],
\end{align*}
and 
\begin{align*}
\mathcal{L}_{\pm}^{(1)} &= -\cos(u) \mathcal{H}_{\pm} - \mathcal{H}_{\pm} [\cos(u) \cdot], \\
\mathcal{M}_{\pm}^{(1)} &= \cos(u) - \sin(u) \mathcal{H}_{\pm},
\end{align*}
where we have used $\mathcal{K} \cos(nu) = n \cos(nu)$, $n \in \mathbb{N}$. For solutions of the modulational stability problem (\ref{mod-Bab-eq}) for 
$(\mu,\lambda)$ near $(0,0)$, we use the scaling 
\begin{equation}
\label{expansion-app-exp}
\mu = a \mu_1, \quad \lambda = a \lambda_1 + a^2 \lambda_2 + \mathcal{O}(a^3), 
\end{equation}
and the expansion
\begin{equation}
\label{expansion-app}
	\left( \begin{array}{c} \hat{v} \\ \hat{w} \end{array} \right) = \left( \begin{array}{c}\hat{v}_0 \\ \hat{w}_0 \end{array} \right) + a \left( \begin{array}{c} \hat{v}_1 \\ \hat{w}_1 \end{array} \right) + a^2 \left( \begin{array}{c} \hat{v}_2 \\ \hat{w}_2 \end{array} \right) + \mathcal{O}(a^3),
\end{equation}
with $\hat{v}_0 = a_1 \sin(u) + a_2 \cos(u)$ and $\hat{w}_0 = a_3$
for arbitrary $(a_1,a_2,a_3) \in \mathbb{R}^3$. 

\underline{\bf Order  $\mathcal{O}(a)$:} We get the linear inhomogeneous system 
\begin{equation}
\left\{ \begin{array}{ll}
\mathcal{K} \hat{w}_1 &=  \lambda_1 \hat{v}_0 + i \mu_1 \mathcal{H}_{\pm} \hat{w}_0, \\
(\mathcal{K}-1) \hat{v}_1 + \mathcal{L}_1 \hat{v}_0 &= \lambda_1 (\hat{w}_0 -  2 \mathcal{H}_{\pm} \hat{v}_0) + i \mu_1 \mathcal{H}_{\pm} \hat{v}_0.
\end{array}
\right. 
\label{level-one-app}
\end{equation}
Projecting the first equation of system (\ref{level-one-app}) to ${\rm Ker}(\mathcal{K}) = {\rm span}(1)$ yields $a_3 = 0$, after 
which we obtain the exact solution 
$$
\hat{w}_1(u) = \lambda_1 (a_1 \sin(u) + a_2 \cos(u)) + \tilde{a}_3,
$$
where $\tilde{a}_3 \in \mathbb{R}$ is a new parameter. The second equation of system (\ref{level-one-app}) 
can be rewritten in the explicit form 
$$
(\mathcal{K}-1) \hat{v}_1 = (i \mu_1 - 2 \lambda_1) [a \cos(u) - a_2 \sin(u)] + 2 a_1 \sin(2u) + 2 a_2 \cos(2u) + a_2,
$$
where we have used $\mathcal{H} \cos(nu) = -\sin(nu)$ and $\mathcal{H} \sin(nu) = \cos(nu)$, $n \in \mathbb{N}$. 
Projecting this equation to ${\rm Ker}(\mathcal{K}-1) = {\rm span}(\sin(u),\cos(u))$ yields 
\begin{equation}
\label{lambda-1-app}
\lambda_1 = \frac{i \mu_1}{2}, 
\end{equation}
after which we obtain the exact solution 
$$
\hat{v}_1(u) = 2 a_1 \sin(2u) + 2 a_2 \cos(2u) - a_2. 
$$

\underline{\bf Order  $\mathcal{O}(a^2)$:} We get the linear inhomogeneous system 
\begin{equation}
\left\{ \begin{array}{ll}
\mathcal{K} \hat{w}_2 &=  \lambda_2 \hat{v}_0 + \lambda_1 (\hat{v}_1 + \mathcal{M}_{\pm}^{(1)} \hat{v}_0) 
+ i \mu_1 \mathcal{H}_{\pm} \hat{w}_1, \\
(\mathcal{K}-1) \hat{v}_2 + \mathcal{L}_1 \hat{v}_1 + \mathcal{L}_2 \hat{v}_0 &= - 2 \lambda_2 \mathcal{H}_{\pm} \hat{v}_0 
+ \lambda_1 (\hat{w}_1 -  2 \mathcal{H}_{\pm} \hat{v}_1) 
+ i \mu_1 (\mathcal{H}_{\pm} \hat{v}_1 + \mathcal{L}_{\pm}^{(1)} \hat{v}_0),
\end{array}
\right. 
\label{level-two-app}
\end{equation}
where we have used that $\hat{w}_0 = 0$. Since 
$$
\hat{v}_1 + \mathcal{M}_{\pm}^{(1)} \hat{v}_0 = 2 a_1 \sin(2u) + 2 a_2 \cos(2u),
$$
projecting the first equation of system (\ref{level-two-app}) to ${\rm Ker}(\mathcal{K}) = {\rm span}(1)$ yields $\tilde{a}_3 = 0$, after which 
we obtain the exact solution 
$$
\hat{w}_2(u) = (\lambda_2 a_1 - i \mu_1 \lambda_1 a_2) \sin(u) + (\lambda_2 a_2 + i \mu_1 \lambda_1 a_1) \cos(u) 
+ \lambda_1 [a_1 \sin(2u) + a_2 \cos(2u)] + \tilde{\tilde{a}}_3,
$$
where $\tilde{\tilde{a}}_3 \in \mathbb{R}$ is a new parameter. For the second equation of system (\ref{level-two-app}), we compute 
\begin{align*}
    \mathcal{L}_1 \hat{v}_1 + \mathcal{L}_2 \hat{v}_0 &= -4 a_1 [\sin(3u) + \sin(u)] - 4 a_2 [\cos(3u) + \cos(u)] + 2 a_2 \cos(u) \\
    & \quad + 2 [a_1 \sin(u) + a_2 \cos(u)] - 2a_1 [\sin(3u)-\sin(u)] - 2 a_2 [\cos(3u) + \cos(u)]
\end{align*}
whereas $\mathcal{H}_{\pm} \hat{v}_1$ and $\mathcal{L}_{\pm}^{(1)} \hat{v}_0$ only return projections to ${\rm span}(1,\sin(2u),\cos(2u))$. 
Projecting the second equation of system (\ref{level-two-app}) to ${\rm Ker}(\mathcal{K}-1) = {\rm span}(\sin(u),\cos(u))$ yields the linear system 
$$
\left[ \begin{matrix} \lambda_1^2 & 2 \lambda_2 \\ -2 \lambda_2 & \lambda_1^2 + 2 \end{matrix} \right] 
\left[ \begin{matrix} a_1 \\ a_2 \end{matrix} \right] = \left[ \begin{matrix} 0 \\ 0 \end{matrix} \right], 
$$
which admits a nontrivial solution for $(a_1,a_2) \in \mathbb{R}^2$ if and only if 
\begin{equation}
\label{lambda-2-app}
\lambda_2^2 = -\frac{1}{4} \lambda_1^2 (2 + \lambda_1^2) = \frac{1}{64} \mu_1^2 (8 - \mu_1^2). 
\end{equation}
Combining (\ref{lambda-1-app}) and (\ref{lambda-2-app}) into (\ref{expansion-app-exp}) yields the characteristic equation (\ref{fig-8-small}). 
Since the spectral bands given by (\ref{fig-8-small}) are non-degenerate, 
persistence of the asymptotic expansion (\ref{expansion-app}) follows from \cite{Hin91,Hor73,Kir92}.
\end{proof}

\begin{remark}
    There exist two additional spectral bands of the modulational stability problem (\ref{mod-Bab-eq}) near the origin associated with the subspace  ${\rm Ker}(\mathcal{K}) = {\rm span}(1)$. Similar to the proof of the second part of Theorem \ref{theorem-bif-1}, the two spectral bands are subsets of $i \mathbb{R}$ associated with the scaling $\mu = a^2 \mu_2$ and $\lambda = a \lambda_1 + \mathcal{O}(a^2)$.
\end{remark}

\section{Numerical methods}
\label{app-B}

For numerical approximations, we compute the profile of the smooth Stokes waves 
from the pseudo--differential equation (\ref{trav-Bab-eq}) and solve the modulational stability problem (\ref{mod-Bab-eq}). 

\subsection{Finding the Stokes wave}
\label{app-B-1}

Finding Stokes waves to high accuracy is a delicate problem that has been studied in detail~\cite{dyachenko2022almost,dyachenko2016branch,korotkevich2022superharmonic}. 
The goal is to determine a nontrivial pair $(c,\eta(u))$, that satisfies the pseudo-differential equation,
\begin{align}
    \left(c^2 \mathcal{K} - 1 \right) \eta - \left( \eta \mathcal{K}\eta + \frac{1}{2}\mathcal{K} \eta^2 \right) = 0.
    \label{trav-Bab-eq2}
\end{align}
Due to the $2\pi$-periodicity and evenness of $\eta = \eta(u)$, we use the truncated Fourier series,
\begin{align}
\label{truncated-Fourier}
    \eta(u) = \sum\limits_{|k| \leq N/2} \hat{\eta}_k e^{iku},
\end{align}
with $\hat{\eta}_k \in \mathbb{R}$, $\forall k$, and $N$ is the number of Fourier modes in the truncated approximation (\ref{truncated-Fourier}). 
The problem (\ref{trav-Bab-eq2}) is solved by a Newton's method coupled with a minimal residual (MINRES) method. The key idea is to take an approximate solution from the linear limit, or a previously computed solution pair, given by $(c^{(j)},\eta^{(j)}(u))$, and apply the continuation method in the $c$ parameter. The wave speed is adjusted to $c^{(j+1)} = c^{(j)} + \delta c^{(j)}$, where $\delta c^{(j)}$ is variable, and the solution, $\eta^{(j)}(u)$, from the preceding step, $j$, is used to initialize Newton iterations on equation~\eqref{trav-Bab-eq2}. After convergence we obtain $(c^{(j+1)},\eta^{(j+1)})$ with a prescribed tolerance, and the next step of continuation method is applied. Once the Stokes wave is found, we define its steepness, $s = (\eta(0) - \eta(\pi))/2\pi$, or equivalently the ratio of crest to trough height to one wavelength, as well as its integral characteristics, like $\mathcal{H}(c)$, $\mathcal{P}(c)$, and $\mathcal{N}(c)$. It becomes convenient to present speed and integral characteristics as functions of $s$ due to their oscillations in $s$.



The method of continuation generally fails at a fold point $c_0 = c(s_0)$, where the linearized operator $\mathcal{L} : H^1_{\rm per} \subset L^2_{\rm per} \to L^2_{\rm per}$ is not invertible on the space of even functions, associated with the bifurcation in Section \ref{sec-6}. In any small neighborhood of $|c-c_0|\leq \delta $, 
there may exist two distinct solutions of equation~\eqref{trav-Bab-eq2} with the same value of speed. This technical difficulty is mitigated by perturbing initial guess for Newton's method, and allows to cross fold points by exploiting basins of attraction of solutions on both sides of the fold point. 

\subsection{Finding eigenvalues}
\label{app-B-2}

Once the Stokes wave has been determined, its linear stability is determined by the eigenvalues of the modulational stability problem 
\begin{align}
\begin{cases}    
\mathcal{K} w = \lambda \mathcal{M} v, \\
\mathcal{L} v = \lambda \left( \mathcal{M}^*w - 2c\mathcal{H}v \right),
\end{cases}
\label{mod-stab-app}
\end{align}
where $\mathcal{L}$ is the linearized Babenko operator~\eqref{lin-Bab-eq}, 
$(v,w) \in H^1(\R) \times H^1(\R)$ is an eigenfunction,
and $\lambda$ is the eigenvalue. By Fourier--Floquet theory, the spectral problem  ~\eqref{mod-stab-app} can be solved by using the truncated Fourier series,
\begin{align}
\begin{cases}
    w(u) = e^{i\mu u} \sum\limits_{|k| \leq N/2} \hat{w}_k e^{iku}, \\
    v(u) =  e^{i\mu u} \sum\limits_{|k| \leq N/2} \hat{v}_k e^{iku},
    \end{cases}
    \label{eigen-truncated}
\end{align}
where the Floquet parameter, $\mu\in[-\frac{1}{2},\frac{1}{2})$ is introduced. We discretize the interval $\mu\in[-\frac{1}{2}, \frac{1}{2}]$ on a grid to approximate $\lambda_j =\lambda_j(\mu)$ for each eigenvalue of~\eqref{mod-stab-app} out of the handful of eigenvalues nearest to zero in the $\lambda$-plane. We refer the reader to~\cite{dyachenko_semenova2022,dyachenko2023quasiperiodic} for a complete recipe of the numerical solution of~\eqref{mod-stab-app} with (\ref{eigen-truncated}), but simply note that the eigenvalue problem is solved via matrix-free orthogonal projection method, where the linear system is solved via minimal residual (MINRES) iterations.

\bibliographystyle{plain}

\bibliography{ModStabStokesClean}

\end{document}